\documentclass[12pt]{article}
\usepackage{graphicx,psfrag,epsf}
\usepackage{enumerate}
\usepackage{natbib}
\usepackage{amsmath}
\usepackage{url} 
\usepackage{hyperref}
\hypersetup{
    colorlinks=true,
    citecolor=blue,
    linkcolor=black,
    filecolor=magenta,      
    urlcolor=cyan,
    pdfpagemode=FullScreen,
    }

\def\boxit#1{\vbox{\hrule\hbox{\vrule\kern5pt\vbox{\kern5pt#1\kern5pt}\kern5pt\vrule}\hrule}}

\def\M{\mathbf{M}}

\def\wh{\widehat}

\newcommand{\ind}{\perp\!\!\!\!\perp}

\def\mE{\mathbb{E}}

\def\cG{\mathcal{G}}

\usepackage{anyfontsize}
\usepackage{amssymb}
\usepackage{amsmath}
\usepackage{amsthm}
\usepackage{bbm}
\usepackage[inline]{enumitem}   
\usepackage[misc]{ifsym}

\newtheorem{theorem}{Theorem}[section]
\newtheorem{proposition}[theorem]{Proposition}
\newtheorem{lemma}[theorem]{Lemma}
\newtheorem{corollary}[theorem]{Corollary}
\newtheorem{definition}[theorem]{Definition}
\newtheorem{assumption}[theorem]{Assumption}

\newcommand{\manyquad}[1][1]{\hspace*{#1em}\ignorespaces}
 \usepackage{multirow}
 \usepackage{booktabs}
 
\usepackage{algorithm}
\usepackage{algpseudocode}
\newcommand{\blind}{1}
\newcommand\numberthis{\addtocounter{equation}{1}\tag{\theequation}}

\begin{document}

\def\spacingset#1{\renewcommand{\baselinestretch}%
{#1}\small\normalsize} \spacingset{1}



\if1\blind
{
  \title{\bf Conditional Local Independence Testing \\ for Itô processes with Applications to Dynamic Causal Discovery}
  \author{Mingzhou Liu$^1$, \, Xinwei Sun$^2$, \, Yizhou Wang$^1$
  \\
  \\
    $^1$School of Computer Science, Peking University\\
    $^2$School of Data Science, Fudan University}
  \maketitle
} \fi


\if0\blind
{
  \bigskip
  \bigskip
  \bigskip
  \begin{center}
    {\LARGE\bf Title}
\end{center}
  \medskip
} \fi

\bigskip
\begin{abstract}
    Inferring causal relationships from dynamical systems is the central interest of many scientific inquiries. Conditional local independence, which describes whether the evolution of one process is influenced by another process given additional processes, is important for causal learning in such systems. In this paper, we propose a hypothesis test for conditional local independence in Itô processes. Our test is grounded in the semimartingale decomposition of the Itô process, with which we introduce a stochastic integral process that is a martingale under the null hypothesis. We then apply a test for the martingale property, quantifying potential deviation from local independence. The test statistics is estimated using the optimal filtering equation. We show the consistency of the estimation, thereby establishing the level and power of our test. Numerical verification and a real-world application to causal discovery in brain resting-state fMRIs are conducted.
\end{abstract}

\vfill

\newpage
{\fontsize{11}{13}\selectfont
\tableofcontents}

\newpage
\spacingset{1.45} 
\section{Introduction}

Dynamic systems governed by Stochastic Differential Equations (SDEs), whose solutions are referred to as \emph{Itô processes}, are prevalent in various fields, including physics \cite{van1976stochastic}, economics \cite{dunbar2019mathematical}, and neuroscience \cite{gilson2019network}. The Dynamic Causal Model (DCM), an extension of Pearl's Structural Causal Model (SCM), considers causal interpretations within dynamic systems. Recently, research on DCM \cite{rubenstein2018deterministic,bellot2021neural,peters2022causal,blom2023causality,boeken2024dynamic,hochsprung2024global} has gathered significant attention.

In dynamic systems, causal relationships are formalized through the concept of \emph{local independence} \cite{florens1989approximate}, which means the history of one process does not directly influence the evolution of another process. More specifically, for two Itô processes $X_\alpha$ and $X_\beta$ within a system $\alpha,\beta \in V$, the process $X_\beta$ is said to be conditionally local independence of $X_\alpha$ given the other processes $X_{V\backslash \alpha}$ if the history $\{X_\alpha(s)\}_{s\leq t}$ does not provide information, beyond that contained in $\sigma\{X_{V\backslash \alpha}(s); s\leq t\}$, about the infinitesimal change of $X_\beta$ at $t$. As such, the conditional local independence defines causality from a mechanism point of view \cite{commenges2009general,aalen2012causality}, with the connection to intervention discussed in \cite{eichler2010granger,hansen2014causal}. When $V$ is completely observed, the test of conditional local independence is then a test of (the absence of) causal relationships.

Another motivation for the conditional local independence test is to learn \emph{the local independence graph} \cite{didelez2008graphical} from data, which serves as a graphical representation for the causal mechanisms within the dynamic system. Equipped with the $\mu$-separation criterion \cite{didelez2008graphical,mogensen2020markov} and the Markov condition \cite{mogensen2018causal}, we can read off all conditional local independence relationships from this graph. Besides, the local independence graph can also be used to determine the identifiability of intervention (i.e., \emph{do}) effect over random processes \cite{didelez2015causal,roysland2025graphical}. Similar to causal directed acyclic graphs (DAGs) \cite{pearl2009causality}, the local independence graph can be recovered from data using constraint-based algorithms \cite{meek2014toward,mogensen2018causal} in conjunction with conditional local independence tests.

Despite the importance, there are no practical tests for conditional local independence in Itô processes. Although \cite{christgau2024nonparametric} proposed a martingale-based test for conditional local independence in counting processes, their procedure relies on the boundedness condition that cannot hold in Itô processes. Furthermore, \cite{lundborg2022conditional,lee2020testing} proposed functional conditional independence tests that can be applied to test the conditional independence of processes in the path space \cite{manten2024signature,manten2025asymmetric}. However, these tests are not applicable to conditional local independence, as the two concepts are not equivalent.

In this paper, we propose a novel test for conditional local independence in Itô processes. We first consider \emph{the semimartingale decomposition} of the process $X_\beta$, showing that the stochastic integral of the cause process $X_\alpha$ relative to the decomposed martingale is itself a martingale under the null hypothesis of local independence. Then, we apply the local covariance test (LCT) \cite{christgau2024nonparametric} to assess the martingale property of the integral process, with the violation of which being evidence for local dependence. Compared to other martingale tests \cite{phillips2014testing,wang2022testing}, the LCT statistic can achieve a $\sqrt{N}$-rate of convergence when estimated using sample splitting, and thus is favored. We relax the boundedness condition in the LCT to a weaker one of moment boundedness, which can then hold in Itô processes given integrability. Finally, for estimation that involves projection processes, we propose an estimator based on the \emph{optimal filtering equation} \cite{liptser2013statistics}, which characterizes the evolution of the projection processes. We show the consistency of the estimation, which then leads to a well-calibrated test.

The rest of the paper is organized as follows. In Sec.~\ref{sec.setup}, we first introduce the general framework of this paper. In Sec.~\ref{sec.cond-loc-ind}, we establish the martingale property of the integral process under the null hypothesis, and introduce the extended local covariance measure for quantifying this property. We outline how this measure can be estimated using sample splitting and optimal filtering in Sec.~\ref{sec.general-estimation}. Then, in Sec.~\ref{sec.asym-ana-ii}, we show the asymptotic properties of the estimated measure, under the assumptions of moment boundedness and consistency of the optimal filtering estimators. Based on these analyses, we drive a test for conditional local independence in Itô processes, and discuss the size and power of the test in Sec.~\ref{sec.lct}. In Sec.~\ref{sec.estimation}, we propose a practical optimal filtering estimator, showing its consistency. We conduct synthetic experiments and apply the proposed test to dynamic causal discovery in Sec.~\ref{sec.exp}. Sec.~\ref{sec.conclusion} concludes the paper and discusses future works. We provide the proofs (\ref{appx.proof-sec-setup}-\ref{appx.asm-part1}), technical lemmas (\ref{appx.tech-lemmas}), and additional simulation results (\ref{app.experiment}) in the appendix.

\section{The Local Covariance Measure for Itô processes}
\label{sec.setup}

In this section, we introduce the general framework of this paper. In Sec.~\ref{sec.cond-loc-ind}, we define the conditional local independence for Itô processes, and introduce an extended Local Covariance Measure for quantifying deviations from conditional local independence. Then, in Sec.~\ref{sec.general-estimation}, we outline how this measure can be estimated from discrete-time observations.

We consider a $d$-dimensional Itô process $X=\{X_1(t),...,X_d(t)\}^\top$, defined on the filtered probability space $(\Omega,\mathcal{F},\{\mathcal{F}_t\}_{0\leq t\leq T},P)$, that satisfies (with probability one):
\begin{equation*}
    X_i(t) = X_i(0) + \int_0^t \lambda_i(s) ds + \sum_{j=1}^{d_1} \int_0^t \Sigma_{ij}(s) dW_j(s), \quad i=1,...,d,
\end{equation*}
where the drift $\lambda(t)=\{\lambda_1(t),...,\lambda_d(t)\}^\top$ and the diffusion matrix $\Sigma(t)=[\Sigma_{ij}(t)]_{d\times d_1}$ consist of adaptive processes, and $W=\{W_1(t),...,W_{d_1}(t)\}^\top$ is a $d_1$-dimensional Wiener process. Let $V:=\{1,...,d\}$. For an element $\alpha\in V$ and a subset $C\subseteq V$, we denote $\mathcal{F}_t^\alpha$ and $\mathcal{F}^C_t$ to be the usual augmentation of the filtrations generated by $X_\alpha$ and $X_C$, respectively.

In particular, we here assume the process is \emph{regular}, in that the drift process $\lambda$ is continuous almost surely and the diffusion coefficient $\Sigma(t)\equiv \Sigma$ is an invertible matrix. Further, we assume that $X_0$ is $L_2$-integrable and $\lambda$ is bounded in $L_2$, that is:
\begin{equation}
    \mathbb{E}\|X_0\|^2_2<\infty, \qquad \sup_{t\in[0,T]} \mathbb{E}\|\lambda(t)\|^2_2 <\infty.
    \label{eq.regularity-lambda-bounded-inL2}
\end{equation}

For brevity, we say the process $X_t$ has the stochastic differential:
\begin{equation*}
    dX_t = \lambda(t)dt + \Sigma dW_t.
\end{equation*}

\subsection{Conditional local independence}
\label{sec.cond-loc-ind}
For the regular Itô process defined above, we introduce the definition of conditional local independence. We note that the $t\mapsto \mathbb{E}\{\cdot |\mathcal{F}_t\}$ below should be interpreted as the optional projection process \cite[Thm. 7.1, vol. 2, p. 319]{rogers2000diffusions}, which is unique up to indistinguishability. This is unproblematic in calculation \cite[Thm. 7.10, vol. 2, p. 320]{rogers2000diffusions}.

\begin{definition}[Conditional local independence]
    We say the process $X_\beta$ is locally independent with $X_\alpha$ given $\mathcal{F}_t^{C}$ ($\alpha\not \in C, \beta\in C$), and write $\alpha\not\to \beta|C$, if the process
    \begin{equation*}
        t\mapsto \mathbb{E}\{\lambda_\beta(t)|\mathcal{F}_t^{C}\}
    \end{equation*}
    is a version of 
    \begin{equation*}
        t\mapsto \mathbb{E}\{\lambda_\beta(t)|\mathcal{F}_t^{\alpha\cup C}\}.
    \end{equation*}
    \label{def.local-ind}
\end{definition}

To ease the notation, we let $\mathcal{F}_t:=\mathcal{F}_t^{C}$ and $\mathcal{G}_t:=\mathcal{F}_t^{\alpha\cup C}$ hereafter. We denote $\mu_t:=\mathbb{E}\{\lambda_\beta(t)|\mathcal{F}_t\}$ and $\boldsymbol{\mu}_t:=\mathbb{E}\{\lambda_\beta(t)|\mathcal{G}_t\}$. We refer to the null hypothesis ``$\mathbb{H}_0: \mu$ is a version of $\boldsymbol{\mu}$'' as local independence in short.

We note that Def.~\ref{def.local-ind} does not imply $\mathbb{E}\{\lambda_\beta(t)|\mathcal{F}_s\}=\mathbb{E}\{\lambda_\beta(t)|\mathcal{G}_s\}, \text{a.s.}, \forall t,s$ -- hence the name \emph{local} independence. We also note that local independence is \emph{asymmetric}, in the sense that ``$\alpha\not\to\beta| C$'' does not imply ``$\beta\not\to\alpha|C$'', which aligns with the asymmetry of causality in time. For $\alpha,\beta\in V$, it follows immediately that $\alpha\not\to\beta|V \backslash \alpha$ if $\lambda_{\beta}$ is $\mathcal{F}_t^{V\backslash \alpha}$-measurable, that is, if $\lambda_{\beta}$ does not depend on the sample path of the $\alpha$-coordinate.

Def.~\ref{def.local-ind} was also adopted in \cite{mogensen2018causal,mogensen2022graphical}. Furthermore, it can be shown that $X_\beta$ is an $\mathcal{F}_t$-semimartingale with the decomposition:
\begin{equation}
    X_\beta(t) = X_\beta(0) + \int_0^t \mu_s ds + M_t,
    \label{eq.F-semimartingale}
\end{equation}
where $M_t:=X_\beta(t)-X_\beta(0) - \int_0^t \mu_s ds$ is an $\mathcal{F}_t$-martingale. Similarly, $X_\beta$ is also a $\mathcal{G}_t$-semimartingale with the decomposition:
\begin{equation}
    X_\beta(t) = X_\beta(0) + \int_0^t \boldsymbol{\mu}_s ds + \mathbf{M}_t,
    \label{eq.G-semimartingale}
\end{equation}
where $\mathbf{M}_t:=X_\beta(t)-X_\beta(0)-\int_0^t \boldsymbol{\mu}_s ds$ is an $\mathcal{G}_t$-martingale. Therefore, Def.~\ref{def.local-ind} is equivalent to requiring that the decompositions \eqref{eq.F-semimartingale} and \eqref{eq.G-semimartingale} coincide, as adopted by \cite{florens1996noncausality}, or that the $\mathcal{F}_t$-martingale $M_t$ is also a $\mathcal{G}_t$-martingale, as used in \cite{christgau2024nonparametric}. We formalize the above remark below.

\begin{proposition}
    Assume \eqref{eq.regularity-lambda-bounded-inL2}, then $X_\beta$ is an $\mathcal{F}_t$ (resp. $\mathcal{G}_t$)-semimartingale with the decomposition \eqref{eq.F-semimartingale} (resp. \eqref{eq.G-semimartingale}), where $M_t$ (resp. $\mathbf{M}_t$) is a continuous, square-integrable $\mathcal{F}_t$ (resp. $\mathcal{G}_t$)-martingale.
    \label{prop.is-martingale}
\end{proposition}

Starting from Prop.~\ref{prop.is-martingale}, we can construct a test of local independence. To be specific, under $\mathbb{H}_0$, $\mu_t$ is a version of $\boldsymbol{\mu}_t$, which means $M_t=\mathbf{M}_t$ is a $\mathcal{G}_t$-martingale. Since $X_\alpha$ is continuous, thus especially $\mathcal{G}_t$-predictable, the stochastic integral process
\begin{equation}
    t\mapsto \int_0^t X_\alpha(s) dM_s
    \label{eq.stoch-integral-test}
\end{equation}
is a $\mathcal{G}_t$-martingale under $\mathbb{H}_0$. A test can then be conducted by detecting whether \eqref{eq.stoch-integral-test} is, indeed, a martingale. For this purpose, we adopt the Local Covariance Measure that was originally proposed by  \cite{christgau2024nonparametric} to test local independence in counting processes. Before introducing this measure, we first introduce the \emph{additive residual process}, which will replace $X_\alpha$ as the integrand in \eqref{eq.stoch-integral-test}. We do so for two reasons. First, to achieve a $\sqrt{N}$-rate via sample splitting, we need the integrand to fulfill \eqref{eq.residual-process-mean-zero} below. Second, this choice of integrand over $X_\alpha$ can lead to a more powerful test.

\begin{definition}[Addictive residual process]
    The addictive residual process $G$ of $X_\alpha$ given $\mathcal{F}_t$ is defined as:
    \begin{equation*}
        G_t:=X_\alpha(t)-\Pi_t,
    \end{equation*}
    where $\Pi_t:=\mathbb{E}\{X_\alpha(t)|\mathcal{F}_t\}$ is the predictable projection of $X_\alpha$ to $\mathcal{F}_t$ \cite[Thm. 19.2, vol. 2, p. 348]{rogers2000diffusions}. 
\end{definition}

To interpret, the additive residual projects $X_\alpha(t)$ onto the orthogonal complement of $L_2(\mathcal{F}_t)$ within $L_2(\mathcal{G}_t)$, retaining information in $X_\alpha$ that is $\mathcal{G}_t$-predictable. We have the process $G_t$ is $\mathcal{G}_t$-predictable and satisfies that:
\begin{equation}
    \mathbb{E}(G_t|\mathcal{F}_t) = 0, \quad 0\leq t\leq T.
    \label{eq.residual-process-mean-zero}
\end{equation}

We can now introduce the Local Covariance Measure for Itô processes.

\begin{definition}[Local Covariance Measure for Itô processes]
    Define for $t\in[0,T]$
    \begin{equation*}
        \gamma_t:=\mathbb{E}(I_t), \quad \text{where} \,\, I_t:=\int_0^t G_s dM_s.
    \end{equation*}
    whenever the expectation is well-defined. We refer to the function $t\mapsto \gamma_t$ as the Local Covariance Measure for Itô processes (LCM-Itô).
    \label{def.lcm}
\end{definition}

Under $\mathbb{H}_0$, the integrator, $M_t$, is a $\mathcal{G}_t$-martingale. Since the integrand, $G_t$, is a $\mathcal{G}_t$-predictable process, we have $I_t$ is a stochastic integral process, and enjoys the zero-mean property $\gamma_t=\mathbb{E}(I_t) = \mathbb{E}(I_0) = 0$. Thus, we can check $\mathbb{H}_0$ by checking the deviation of $\gamma_t$ from zero.

\begin{proposition}
    Assume \eqref{eq.regularity-lambda-bounded-inL2}, then under $\mathbb{H}_0$, the process $I=(I_t)$ is a $\mathcal{G}_t$-martingale with $I_0=0$. Thus, we have $\gamma_t=0$ for all $t\in[0,T]$.
    \label{prop.is-zero-under-h0}
\end{proposition}

To interpret $\gamma$ in the alternative, we see that:
\begin{proposition}
Assume \eqref{eq.regularity-lambda-bounded-inL2}, then for every $0\leq t\leq T$, we have:
\begin{equation*}
    \gamma_t = \int_0^t \mathrm{cov}(G_s,\boldsymbol{\mu}_s-\mu_s)ds.
\end{equation*}
In particular, $\gamma$ is the zero-function if and only $\mathrm{cov}(G_t,\boldsymbol{\mu}_t-\mu_t)=0$ for almost all $t$.
\label{prop.is-covariance}
\end{proposition}

According to Prop.~\ref{prop.is-covariance}, the LCM-Itô quantifies deviations from $\mathbb{H}_0$ in terms of the covariance between the residual process and the difference of the $\mathcal{F}_t$- and $\mathcal{G}_t$-projections of the drift process. In particular, when the drift process $\lambda_t=\Phi X_t + b$ is a linear function of $X_t$ and the diffusion matrix is diagonal, we can show that $\gamma$ is the zero-function if and only if the parameter $\Phi_{\beta \alpha}=0$.

For the additive residual process $G_t=X_\alpha(t)-\Pi_t$, we have:
\begin{equation*}
    \gamma_t = \mathbb{E}\Big(\int_0^t G_s dM_s\Big) = \mathbb{E}\Big\{\int_0^t X_\alpha(s) dM_s\Big\} - \mathbb{E}\Big(\int_0^t \Pi_s dM_s\Big)
\end{equation*}
under the regularity \eqref{eq.regularity-lambda-bounded-inL2}. Since the predictable projection $\Pi_t$ is $\mathcal{F}_t$-predictable, and since $M_t$ is a $\mathcal{F}_t$-martingale, the process $\int_0^t \Pi_s dM_s$ is a $\mathcal{F}_t$-martingale with mean zero. Hence,
\begin{equation*}
    \gamma_t = \mathbb{E}\Big\{\int_0^t X_\alpha(s) dM_s\Big\},
\end{equation*}
which means the addictive residual process defines the same measure $\gamma_t$ as the stochastic integral \eqref{eq.stoch-integral-test} would. Nonetheless, the property \eqref{eq.residual-process-mean-zero} of the addictive residual process will allow us to achieve a $\sqrt{N}$-rate of convergence of the estimator of $\gamma_t$ in cases where the estimator of $\mu$ converges at a slower rate.

\subsection{General estimation procedure}
\label{sec.general-estimation}

In this section, we introduce an estimation procedure for the LCM-Itô based on sample splitting. To this end, we consider $N_0$ i.i.d replications of the process $X$, as similarly assumed in \cite{manten2024signature,christgau2024nonparametric}. Each process is discretely observed on a uniformly-spaced time grid $\{0,\delta,2\delta,...,n\delta\}$, with $T=n\delta$. We denote $\underline{t}:=i\delta$ for $t\in [k\delta,(k+1)\delta)$. 

We split the dataset into two disjoint subsets, indexed by $J$ and $J^c$, respectively, with $N:=|J|, N_c:=|J^c|$, and $J\cup J^c = \{1,...,N_0\}$. Then, we estimate $\Pi$ and $\mu$ using data indexed by $J^c$. By an estimate, $\hat{\Pi}$ (resp. $\hat{\mu}$), of $\Pi$ (resp. $\mu$), we mean a (stochastic) function fitted on samples indexed by $J^c$, and can be evaluated on the basis of $\mathcal{F}_{j,t}$ for each $j\in J$. In practice, these estimators can be constructed with the \emph{optimal filtering equations} \cite[Thm. 8.1, vol. 1. p. 318]{liptser2013statistics}, which are SDEs characterizing the evolutions of $\Pi$ and $\mu$. We call such estimators the optimal filtering estimators. In Sec.~\ref{sec.estimation}, we introduce an example of such estimators and discuss its statistical properties.

After having estimated the functional forms of $\hat{\Pi}$ and $\hat{\mu}$, we can evaluate them on the samples from $J$, for which we denote as $\hat{\Pi}_j, \hat{\mu}_j$ ($j\in J$), respectively. From these, we can obtain an estimate of $G_t$ with
\begin{equation*}
    \hat{G}_{j,t} = \hat{G}_{j,\underline{t}} = X_{\alpha,j}(\underline{t}) - \hat{\Pi}_{j,\underline{t}},
\end{equation*}
and an estimate of $M_t$ by 
\begin{equation*}
    \hat{M}_{j,t} = \hat{M}_{j,\underline{t}} = X_{\beta,j}(\underline{t}) - \int_0^t \hat{\mu}_{j,\underline{s}} ds.
\end{equation*}
It then follows that we can estimate the LCM-Itô by
\begin{align*}
    \hat{\gamma}_t = \frac{1}{N} \sum_{j\in J} \int_0^t \hat{G}_{j,\underline{s}} d\hat{M}_{j,\underline{s}} = \frac{1}{N} \sum_{j\in J} \sum_{l=1}^k \hat{G}_{j,(l-1)\delta} \{\hat{M}_{j,l\delta}-\hat{M}_{j,(l-1)\delta}\}
\end{align*}
for $t\in [k\delta,(k+1)\delta)$. We summarize the above procedure in Alg.~\ref{alg:estimation-lcm}.

\algdef{SE}{Begin}{End}{\textbf{begin}}{\textbf{end}}

\begin{algorithm}[htp]
\caption{Estimation for LCM-Itô}
\label{alg:estimation-lcm}
\begin{algorithmic}[1]
    \State \textbf{inputs:} samples $\{X_j(t)|j=1,...,N_0, t=0,...,n\delta\}$, a partition of indices $J\cup J^c$.
    \State \textbf{options:} optimal filtering estimators for $\Pi$ and $\mu$.
    \State Fit the estimators $\hat{\Pi}$ and $\hat{\mu}$ from the data $\{X_j(t)|j\in J^c, t=0,...,n\delta\}$.
    \State For each $j\in J$ and $k=1,...,n$, compute the estimated residual $\hat{G}_{j,k\delta}=X_{\alpha,j}(k\delta)-\hat{\Pi}_{j,k\delta}$ and the martingale $\hat{M}_{j,k\delta} = X_{\alpha,j}(k\delta) - \delta\sum_{l=1}^k \hat{\mu}_{j,k\delta}$.
    \State For each $k=1,...,n$, compute $\hat{\gamma}_{k\delta} = \frac{1}{N}\sum_{j\in J} \sum_{l=1}^k \hat{G}_{j,(l-1)\delta}\{\hat{M}_{j,l\delta}-\hat{M}_{j,(l-1)\delta}\}$.
    \State \textbf{output:} the estimated LCM-Itô $\hat{\gamma}_t$.
\end{algorithmic}
\end{algorithm}

We note that $\hat{\gamma}_t$ is a double machine learning estimator \cite{chernozhukov2018double} for $\gamma_t$, with the observations indexed by $J^c$ used to learn models of $\Pi$ and $\mu$, and with observations indexed by $J$ to estimate $\gamma_t$ based on these models. In Sec.~\ref{sec.lct}, we will introduce a more efficient estimator based on cross-fitting that can lead to a more powerful test. But for now, we first illustrate the asymptotic properties of $\hat{\gamma}$ using the sample-splitting estimator which is simpler.

\section{Asymptotic analyses}
\label{sec.asym-ana-ii}

In this section, we show asymptotic properties of the LCM-Itô estimator introduced in Sec.~\ref{sec.general-estimation}. Based on these asymptotic results, we will derive a test of $\mathbb{H}_0$ in Sec.~\ref{sec.lct}. To begin with, recall that the estimator $\hat{\gamma}$ satisfies:
\begin{equation*}
    \hat{\gamma}_t = \frac{1}{N} \sum_{j\in J} \int_0^t \hat{G}_{j,s} d\hat{M}_{j,s}.
\end{equation*}

With this form, we consider the decomposition
\begin{equation}
    \sqrt{N}\hat{\gamma} = U + R_1 + R_2 + R_3 + D_1 + D_2,
    \label{eq.decomposition}
\end{equation}
where
\begin{align*}
    &U_t = \frac{1}{\sqrt{N}} \sum_{j\in J} \int_0^t G_{j,s} d\mathbf{M}_{j,s}, \\
    &R_{1,t} = \frac{1}{\sqrt{N}} \sum_{j\in J} \int_0^t G_{j,s} (\mu_{j,s}-\hat{\mu}_{j,s}) ds, \\
    &R_{2,t} =  \frac{1}{\sqrt{N}} \sum_{j\in J} \int_0^t (\hat{G}_{j,s}-G_{j,s}) d\mathbf{M}_{j,s}, \\
    &R_{3,t} =  \frac{1}{\sqrt{N}} \sum_{j\in J} \int_0^t (\hat{G}_{j,s}-G_{j,s}) (\mu_{j,s}-\hat{\mu}_{j,s}) ds,\\
    &D_{1,t} = \frac{1}{\sqrt{N}} \sum_{j\in J} \int_0^t G_{j,s} (\boldsymbol{\mu}_{j,s}-\mu_{j,s})ds, \\
    &D_{2,t} = \frac{1}{\sqrt{N}} \sum_{j\in J} \int_0^t (\hat{G}_{j,s} - G_{j,s})(\boldsymbol{\mu}_{j,s}-\mu_{j,s})ds.
\end{align*}

In the following, we will show that the processes $U$ and $D_1-\sqrt{N}\gamma$ each converge in distribution, and that the remaining terms converge to the zero-processes. In this regard, we have that $\sqrt{N}(\hat{\gamma}-\gamma)$ is stochastically bounded, so the LCM-Itô estimator will asymptotically detect if the $\gamma$ is non-zero. Moreover, under $\mathbb{H}_0$, $\mu$ is a version of $\boldsymbol{\mu}$, which means the processes $D_1$ and $D_2$ are zero-processes almost surely. It then follows that $U$ drives the asymptotic limit distribution of the LCM-Itô estimator under $\mathbb{H}_0$. During the discussion, we will use $\overset{\mathcal{D}}{\to}$ and $\overset{\mathcal{P}}{\to}$ to denote convergence in distribution and convergence in probability, respectively. We denote $C[0,T]$ as the space of continuous functions on $[0,T]$.

We first show the process $U$ converges in distribution to a Gaussian martingale as $N\to\infty$. In what follows, we define a function $\mathcal{V}:[0,T]\mapsto \mathbb{R}^+$ that is the variance of the limiting process.

\begin{definition}
    We define the variance function $\mathcal{V}$ as:
    \begin{equation}
        \mathcal{V}_t = \|\Sigma_\beta\|^2_2 \mathbb{E}\Big(\int_0^t G_s^2 ds \Big),
        \label{eq.def-variance-function}
    \end{equation}
    where $\Sigma_\beta := [\Sigma_{\beta 1},...,\Sigma_{\beta d}]^\top$.
\end{definition}

We note that under the regularity \eqref{eq.regularity-lambda-bounded-inL2}, $\mathcal{V}_t<\infty$ for every $t\in [0,T]$. Moreover, it can be shown that $\mathcal{V}_t$ is the variance of $\int_0^t G_s d\mathbf{M}_s$. We then show that:

\begin{proposition}
    Assume \eqref{eq.regularity-lambda-bounded-inL2}, then $U \overset{\mathcal{D}}{\to} U_0$ in $C[0,T]$ as $N\to \infty$, where $U_0$ is a mean-zero continuous Gaussian martingale on $[0,T]$ with variance function $\mathcal{V}$.
    \label{prop.convergence-u}
\end{proposition}

To control the remainder terms $R_1, R_2, R_3$, we require the estimators $\hat{\mu}$ and $\hat{G}$ to be consistent. Specifically, define:
\begin{equation*}
    \|Z\|_p := \left(\mathbb{E}\int_0^T \|Z_s\|^p ds\right)^{\frac{1}{p}}
\end{equation*}
as the $\ell_p$-norm for any (vector) random process $Z\in L_p(\Omega\times [0,T])$. Then, we assume:

\begin{assumption}
    There exists $p>2$ such that: 
    \begin{equation*}
        \lim_{N \to \infty, \delta\to 0} \max \left\{\|\mu-\hat{\mu}\|_p,\, \|G-\hat{G}\|_p, \, N^{\frac{1}{2}}\|\mu-\hat{\mu}\|_2 \|G-\hat{G}\|_2\right\} = 0.
    \end{equation*}
    \label{asm.Lp-consistency-estimator}
\end{assumption}

We make the following remark on Asm.~\ref{asm.Lp-consistency-estimator}. Compared with \cite{christgau2024nonparametric} that studied discrete-time processes, we here study the continuous-time Itô process. Therefore, the error terms in Asm.~\ref{asm.Lp-consistency-estimator} are also influenced by the granularity of the time grid. Thus, to vanish the error, we require not only a sufficiently large sample size $N\to \infty$, but also a refined time partition $\delta \to 0$. This is similarly adopted in \cite{yoshida1992estimation,denis2020consistent,denis2021ridge}. In Sec.~\ref{sec.estimation}, we will provide a practical estimator based on the optimal filtering equation, such that the Asm.~\ref{asm.Lp-consistency-estimator} is satisfied. 

We also need the following assumption that bounds the moments of $\mu$ and $G$.
\begin{assumption}
    For the $q>2$ such that $\frac{1}{p}+\frac{1}{q}=\frac{1}{2}$ with the $p$ in Asm.~\ref{asm.Lp-consistency-estimator}, the processes $\mu$ and $G$ are bounded in $L_{q}$. That is, we assume:
    \begin{equation*}
        \sup_{t\in [0,T]} \mathbb{E}|\mu_t|^q <\infty, \quad \sup_{t\in [0,T]} \mathbb{E}|G_t|^q <\infty.
    \end{equation*}
    We also assume the process $\boldsymbol{\mu}$, the estimators $\hat{\mu}$ and $\hat{G}$ satisfy the same bounds as $\mu$ and $G$, respectively.
    \label{asm.moment-bound}
\end{assumption}

A similar but stronger assumption on bounding the infinity norms of $\mu$ and $G$ was required in \cite{christgau2024nonparametric}. This kind of assumption is important to establish the stochastic equicontinuity of the remainder processes, and hence achieve convergence uniformly over $t$ using the chaining argument. Nonetheless, unlike \cite{christgau2024nonparametric} that requires $\mu$ and $G$ to be bounded, which can be impractical for Itô processes, we here only bound their moments. These moment bounds can be satisfied by diffusion processes under the global Lipschitz condition and the linear growth condition. Please refer to Corollary~\ref{coro.lp-bounded} in the appendix for details.

With these assumptions, we can establish that the remainder terms converge uniformly to the zero-process.

\begin{proposition}
    Assume Asms.~\ref{asm.Lp-consistency-estimator}, \ref{asm.moment-bound}, then we have $\sup_{t\in [0,T]}|R_{i,t}| \overset{\mathcal{P}}{\to} 0$ as $N\to \infty, \delta\to 0$, for $i=1,2,3$.
   \label{prop.R1-3-zero}
\end{proposition}

To control the asymptotic behavior of the LCM estimator in the alternative $\mathbb{H}_1$, we also need to control the two terms $D_1$ and $D_2$. To be specific, let $Q_t := \int_0^t G_{s}(\boldsymbol{\mu}_{s}-\mu_{s}) ds - \gamma_t$, we show that:

\begin{proposition}
    Assume Asms.~\ref{asm.Lp-consistency-estimator}, \ref{asm.moment-bound}, then the process $D_1 - \sqrt{N} \gamma$ converges in distribution in $C[0,T]$ to a Gaussian process with covariance function $(s,t)\mapsto \mathrm{Cov}(Q_s,Q_t)$ as $N\to\infty,\delta\to 0$. Moreover, we have $\sup_{t\in [0,T]} |D_{2,t}|\overset{\mathcal{P}}{\to} 0$ as $N\to\infty,\delta\to 0$.
    \label{prop.asym-D1-D2}
\end{proposition}

We can then combine the above propositions into a single theorem regarding the asymptotic behavior of the LCM-Itô estimator:

\begin{theorem}
    Assume \eqref{eq.regularity-lambda-bounded-inL2} and Asms.~\ref{asm.Lp-consistency-estimator}, \ref{asm.moment-bound}, then 
    \begin{enumerate}[label=(\roman*)]
        \item  under $\mathbb{H}_0$, it holds that
        \begin{equation*}
            \sqrt{N} \hat{\gamma} \overset{\mathcal{D}}{\to} U_0
        \end{equation*}
        in $C[0,T]$ as $N\to\infty, \delta\to 0$, where $U_0$ is a mean zero continuous Gaussian martingale on $[0,T]$ with variance function $\mathcal{V}$.
        \item For every $\epsilon>0$, there exists $K>0$ such that
        \begin{equation*}
            \limsup_{N\to\infty,\delta\to 0} P(\sqrt{N} \sup_{t\in [0,T]}|\hat{\gamma}_t-\gamma_t|> K) < \epsilon.
        \end{equation*}
    \end{enumerate}
    \label{thm.main-4.6}
\end{theorem}

Thus, we have established the weak asymptotic limit of $\sqrt{N} \gamma$ under $\mathbb{H}_0$. However, the variance function $\mathcal{V}$ of the limiting Gaussian martingale is unknown and must be estimated from data. We consider the following plug-in estimator:
\begin{equation}
    \hat{\mathcal{V}}_t = \|\hat{\Sigma}_\beta\|^2 \frac{1}{N} \sum_{j\in J} \int_0^t \hat{G}_{j,s}^2 ds,
    \label{eq.variance-estimator}
\end{equation}
where $\hat{\Sigma}_\beta$ is from the estimator $\hat{\Sigma}$ of the diffusion matrix. To establish the consistency of $\hat{\mathcal{V}}$, we need to assume $\hat{\Sigma}$ is consistent:
\begin{assumption}
    We assume $\hat{\Sigma}\overset{\mathcal{P}}{\to} \Sigma$ as $N\to \infty, \delta\to 0$.
    \label{asm.Sigma-convergence-in-probability}
\end{assumption}

We then have the following result:
\begin{proposition}
     Assume Asms.~\ref{asm.Lp-consistency-estimator}, \ref{asm.moment-bound}, \ref{asm.Sigma-convergence-in-probability}, then we have $\sup_{t\in [0,T]} |\hat{\mathcal{V}}_t - \mathcal{V}_t| \overset{\mathcal{P}}{\to} 0$ as $N\to\infty, \delta\to 0$.
     \label{prop.consistency-vt}
\end{proposition}
\section{Conditional local independence testing}
\label{sec.lct}

In this section, we introduce a test based on the LCM-Itô estimator. Using the asymptotic properties derived in Sec.~\ref{sec.asym-ana-ii}, we show the asymptotic distribution of our testing statistic under the null hypothesis, and thus establish the level of our test. In addition, we show a power result for alternatives where $\gamma$ decays at an order of at most $N^{-\frac{1}{2}}$. Finally, we extend the test to be based on a cross-fitted LCM-Itô estimator and establish the corresponding level result.

To be specific, we consider the Local Covariance Test statistic introduced by \cite{christgau2024nonparametric}:
\begin{equation}
    \hat{T}_N := \frac{\sqrt{N}}{\sqrt{\hat{\mathcal{V}}_T}} \sup_{t\in [0,T]} |\hat{\gamma}_t|.
    \label{eq.lct-statistic}
\end{equation}

This statistic captures the deviation of $\gamma$ from zero at any $t\in [0,T]$ by taking the infinity norm of $\hat{\gamma}$. Moreover, it is normalized over $\hat{\mathcal{V}}_T$, which is the limiting variance $\sqrt{N} \sup_t |\hat{\gamma}_t|$ under $\mathbb{H}_0$, so that the limiting null distribution does not depend on $\mathcal{V}$.

\subsection{Type I and type II error control}

In this section, we establish the asymptotic level and power of our test. For that purpose, we first need to assume the variance at $t=T$ is bounded away from zero:
\begin{assumption}
    There exists a $\nu_0>0$ such that $\mathcal{V}_T\geq \nu_0$.
    \label{asm.bound-away-from-zero}
\end{assumption}

We proceed to show that under $\mathbb{H}_0$, the statistic $\hat{T}_N$ is asymptotically distributed as $S:=\sup_{t\in [0,T]} |W_t|$, where $(W_t)$ is a standard Wiener process. We note that the Cumulative Density Function (CDF) of $S$ can be written as:
\begin{equation}
    F_S(x) = P(S\leq x) = \frac{4}{\pi} \sum_{i=0}^\infty \frac{(-1)^i}{2i+1} \exp\left\{-\frac{\pi^2 (2i+1)^2T}{8x^2}\right\}, \quad x>0.
    \label{eq.sup_t-W_t-cdf}
\end{equation}

The $p$-value for a test of $\mathbb{H}_0$ then equals $1-F_S(\hat{T}_N)$, and since the series in \eqref{eq.sup_t-W_t-cdf} converges at an exponential rate, the $p$-value can be computed with high numerical precision. Furthermore, for a significance level $\alpha\in (0,1)$, the $(1-\alpha)$-quantile of $F_S$, denoted as $z_{1-\alpha}$, exists and is unique since the right-hand side of \eqref{eq.sup_t-W_t-cdf} is strictly increasing and continuous.

Then, based on Thm.~\ref{thm.main-4.6}, we can deduce the asymptotic level of our test:
\begin{theorem}
    Under \eqref{eq.regularity-lambda-bounded-inL2} and  Asms.~\ref{asm.Lp-consistency-estimator}, \ref{asm.moment-bound}, \ref{asm.Sigma-convergence-in-probability}, and \ref{asm.bound-away-from-zero}, we have:
    \begin{equation*}
        \hat{T}_N \overset{\mathcal{D}}{\to} S
    \end{equation*}
    as $N\to\infty, \delta\to 0$. As a consequence, for any $\alpha\in (0,1)$, we have:
    \begin{equation*}
        \limsup_{N\to\infty, \delta\to 0} P(\hat{T}_N>z_{1-\alpha}) \leq \alpha.
    \end{equation*}
    In other words, our test has an asymptotic level $\alpha$.
    \label{thm.asym-level}
\end{theorem}

In general, we can not expect the test to have power against all alternatives to $\mathbb{H}_0$. This is analogous to other types of conditional independence tests \cite{lundborg2022conditional,cai2022distribution}. However, we do have the following result that establishes power against local alternatives in which $\sup_t |\gamma_t|$ decays at an order of at most $N^{-\frac{1}{2}}$:

\begin{theorem}
    Assume \eqref{eq.regularity-lambda-bounded-inL2} and Asms.~\ref{asm.Lp-consistency-estimator}, \ref{asm.moment-bound}, \ref{asm.Sigma-convergence-in-probability}. Then, for any $0<\alpha<\beta<1$, there exists $c>0$ such that:
    \begin{equation*}
        \liminf_{N\to\infty,\delta\to 0} \inf_{\gamma \in\mathcal{A}_{c,N}} P(\hat{T}_N>z_{1-\alpha})\geq \beta,
    \end{equation*}
    where $\mathcal{A}_{c,N} := \left\{ \sup_{t\in [0,T]} |\gamma_t | \geq c N^{-\frac{1}{2}} \right\}$. In other words, the type II error rate of the test can be controlled to arbitrarily small.
    \label{thm.asym-power}
\end{theorem}

\subsection{Extension to cross-fitting}

In this section, we improve the efficiency of the LCM-Itô estimator using cross-fitting \cite{chernozhukov2018double}, which then leads to a more powerful test. 

To be specific, we partite the dataset into $K$ disjoint folds $J_1 \cup \cdots \cup J_K = \{1,...,N_0\}$ with $N_k:=|J_k|$. Then, for each $k=1,...,K$, we estimate the models using data indexed by $J_k^c=\{1,...,N_0\}\backslash J_k$, and subsequently estimate the LCM-Itô using data indexed by $J_k$. Heuristically, the obtained $K$ estimators are approximately independent, thus their average should be a more efficient estimator. Following this, we define the $K$-fold Cross-fitted LCM-Itô estimator as:
\begin{equation*}
    \hat{\gamma}_K = \frac{1}{K} \sum_{k=1}^K \frac{1}{N_k} \sum_{j\in J_k} \int_0^t \hat{G}_{j,s}^k d\hat{M}_{j,s}^k,
\end{equation*}
where $\hat{G}_{j,s}^k, \hat{M}_{j,s}^k$ are the model prediction of $G$ and $M$, respectively, based on training data index by $J_k^c$. We also define the $K$-fold version of the variance estimator:
\begin{equation*}
    \hat{\mathcal{V}}_K(t) = \frac{1}{K} \sum_{k=1}^K \|\hat{\Sigma}_\beta^k\|^2 \frac{1}{N_k} \sum_{j\in J_k} \int_0^t (\hat{G}_{j,s}^k)^2 ds,
\end{equation*}
where $\hat{\Sigma}^k_\beta$ is from the estimator $\hat{\Sigma}^k$ of the diffusion matrix using data index by $J_k^c$.

Now, we can define the testing statistic as:
\begin{equation*}
    \hat{T}_K = \frac{\sqrt{N_0}}{\sqrt{\hat{\mathcal{V}}_K(T)}} \sup_{t\in [0,T]} |\hat{\gamma}_K|.
\end{equation*}

We summarize the computation of $\hat{T}_K$ in Alg.~\ref{alg:estimation-crossfit-lcm}.  Assume the partition has a uniform asymptotic density, i.e.,  $N_k/N_0\to \frac{1}{K}$ for each $k=1,2,...,K$, then the asymptotic analysis of $\hat{\gamma}$ in Sec.~\ref{sec.asym-ana-ii} generalized directly to $\hat{\gamma}_K$. For simplicity, we omit to restate all results for the K-fold cross-fitting estimator and focus on the asymptotic level of the test.

\begin{theorem}
    Suppose that Asms.~\ref{asm.Lp-consistency-estimator}, \ref{asm.Sigma-convergence-in-probability} are satisfied for every sample split $J_k \cup J_k^c$. Under\eqref{eq.regularity-lambda-bounded-inL2} and  Asms.~\ref{asm.moment-bound}, \ref{asm.bound-away-from-zero}, we have:
    \begin{equation*}
        \hat{T}_K \overset{\mathcal{D}}{\to} S
    \end{equation*}
    as $N\to\infty, \delta\to 0$. As a consequence, for any $\alpha\in (0,1)$, we have:
    \begin{equation*}
        \limsup_{N\to\infty, \delta\to 0} P(\hat{T}_K>z_{1-\alpha}) \leq \alpha.
    \end{equation*}
    \label{thm.asym-level-Kfold}
\end{theorem}

A similar result regarding the power of the test, as established in Thm.~\ref{thm.asym-power}, can also be demonstrated, although we will not delve into the details here. Notably, the cross-fitting approach recovers full efficiency in the sense that the rate is $\sqrt{N_0}$ instead of $\sqrt{N}$, which leads to a more powerful test. Moreover, the asymptotic distribution of $\hat{T}_K$ does not rely on $K$, and any difference between various choices of $K$ can thus be attributed to finite sample errors. In our experiments, we found that setting $K$ to $3,4$ or $5$ yields consistent results.

\algdef{SE}{Begin}{End}{\textbf{begin}}{\textbf{end}}

\begin{algorithm}[htp]
\caption{$K$-fold cross-fitted local independence test}
\label{alg:estimation-crossfit-lcm}
\begin{algorithmic}[1]
    \State \textbf{inputs:} samples $\{X_j(t)|j=1,...,N_0, t=0,...,n\delta\}$, a partition of indices $J_1\cup \cdots \cup J_K$, the significance level $\alpha$.
    \For{$k=1,...,K$}
        \State Apply Alg.~\ref{alg:estimation-lcm} on the sample split $J_k \cup J_k^c$ to compute $\hat{\gamma}_k(l\delta), l=1,...,n$.
        \State Apply Eq.~\eqref{eq.variance-estimator} on the sample split $J_k \cup J_k^c$ to compute $\hat{\mathcal{V}}_k(T)$.
    \EndFor
    \State Let $\hat{\gamma}_K(i\delta) = \frac{1}{K}\sum_{k=1}^K \hat{\gamma}_k(l\delta), l=1,...,n$ and $\hat{\mathcal{V}}_K(T) = \frac{1}{K}\sum_{k=1}^K \hat{\mathcal{V}}_k(T)$.
    \State Compute the $K$-fold statistic $\hat{T}_K=\sqrt{N_0/\hat{\mathcal{V}}_K(T)}\max_l |\hat{\gamma}_K(l\delta)|$.
    \State Compute the $p$-value $\hat{p}=1-F_S(\hat{T}_K)$ with \eqref{eq.sup_t-W_t-cdf}.
    \State \textbf{output:} the $p$-value $\hat{p}$.
\end{algorithmic}
\end{algorithm}
\section{An optimal filtering estimator}
\label{sec.estimation}

In this section, we estimate the processes $\Pi$ and $\mu$ such that the consistency in Asm.~\ref{asm.Lp-consistency-estimator} is satisfied. Recall that $\Pi_t=\mathbb{E}\{X_\alpha(t)|\mathcal{F}_t\}$ and $\mu=\mathbb{E}\{\lambda_\beta(t)|\mathcal{F}_t\}$ are both projection processes. Hence, their evolutions follow the \emph{optimal filtering equations} \cite[Thm.~8.1, vol. 1, p. 318]{liptser2013statistics}, which are essentially SDEs with drift and diffusion related to $\lambda$ and $\Sigma$. For instance, consider the following diffusion process, which is the most common special case of the Itô process:
\begin{equation*}
    dX_t = \lambda(t,X_t)dt + \Sigma dW_t.
\end{equation*}

Let $h_t:=h(t,X_t)$, with $h$ a measurable function such that $\sup_t \mathbb{E}|h_t|<\infty$, then the process $\pi_t(h)=\mathbb{E}(h_t|\mathcal{F}_t)$ has an evolution characterized by\footnote{Please refer to \eqref{eq.filtering-eq-complete-form} for the complete form.}:
\begin{equation}
    d\pi_t(h) = \pi_t(\mathcal{L}h) dt + \left\{\pi_t(\mathcal{N} h) \Sigma_C + \pi_t(\lambda_{C} h)-\pi_t(\lambda_{C}) \pi_t(h)\right\} \Psi^{-1} \{dX_{C,t}-\pi_t(\lambda_C)dt\},
    \label{eq.filter-eq-general}
\end{equation}
where $\Psi = \Sigma_{C,V\backslash C}\Sigma_{C,V\backslash C}^\top + \Sigma_C \Sigma_C^\top$. Then let $h(t,X_t)=X_\alpha(t)$ and $h(t,X_t)=\lambda_\beta(t,X_t)$, we can get the evolution equations for $\Pi$ and $\mu$, respectively. In this regard, we can first estimate $\lambda$ and $\Sigma$, then use the Euler-Maruyama method to solve these equations with plugins to estimate $\Pi$ and $\mu$. Nonetheless, one may notice that \eqref{eq.filter-eq-general} is not always \emph{closed}, i.e., it may contain terms that cannot be expressed (in general) as functions of $h_t$. Indeed, it is shown that \eqref{eq.filter-eq-general} is closed only in particular cases \cite{bain2009fundamentals}. For cases where it is not closed, however, it is possible to find numerically tractable approximations to the solution \cite{crisan2002survey,budhiraja2007survey}.

In this paper, we consider an important model where \eqref{eq.filter-eq-general} is closed, i.e., the Ornstein-Uhlenbeck (OU) process that is defined by
\begin{equation}
        dX_t = \Phi X_t dt + \Sigma dW_t \quad \text{with initial value $X_0$},
        \label{eq.multivariate-ou}
    \end{equation}
where $\Phi$ is a $d\times d$ matrix and $\Sigma$ is the matrix of the diffusion coefficients. In particular, we consider isometric noise $\Sigma=\sigma I_d$ as in \cite{bellot2022neural}. We assume the initial value $X_0$ is sampled from the following linear structural equation model:
\begin{equation}
    X_0 = \Phi X_0 + e \quad e\sim \mathcal{N}(0,\sigma^2 I_d),
    \label{eq.initial-value-linear-model}
\end{equation}
which is compatible with the OU process \cite{mogensen2022graphical}. To ensure the identifiability of $X_0$, we assume the matrix $(I-\Phi)$ is positive definite.

The OU process is also known as the (multi-factor) Vasicek model in economics \cite{vasicek1977equilibrium,langetieg1980multivariate}. It also has important applications in physics, chemistry \cite{van1992stochastic}, and neuroscience \cite{wilmer2010method,gilson2019network}. For the OU model, we can first estimate the parameters, then estimate $\Pi$ and $\mu$ by plugging in the filtering equation. We will introduce the details in the subsequent two sections.

\subsection{OU model estimation}

In this section, we estimate the parameters $\Phi$ and $\Sigma$ of the OU process, with data indexed by $J^c$. For technical reasons, we let $\delta_c$, which can be larger than $\delta$, be the sample interval of our data. In practice, this can be achieved by subsampling the original series.

For estimation, \cite{phillips1972structural} showed that the discrete-time observation follows the Vector Auto-Regression (VAR) process:
\begin{equation}
    X_t = F X_{t-1} + \varepsilon_t, \quad t=0,\delta_c,...
    \label{eq.var(1)}
\end{equation}
where $F=e^{\Phi\delta_c}$ defined by matrix exponential and $\{\varepsilon_t\}$ is a martingale difference sequence with covariance matrix $\Omega = \mathbb{E} (\varepsilon_t \varepsilon_t^\top) = \int_0^{\delta_c} e^{\Phi s} \Sigma e^{\Phi^\top s} ds$. \cite{PHILLIPS1973351} showed that $(\Phi,\Sigma)$ is identifiable from $(F,\Omega)$ if and only if the matrix $\Phi$ is identifiable from $F=e^{\Phi\delta_c}$. \cite[Thm. 5]{hansen1983dimensionality} gave a condition where this identifiability holds for sufficiently small $\delta_c$, for which we assume below:

\begin{assumption}
    Denote $Q$ as the matrix composed of eigenvectors of $\Phi$ and denote $\bar{Q}$ as the complex conjugate of $Q$. We assume the matrix $R_0=Q^{-1} \Sigma (\bar{Q}^{-1})^\top$ does not have any zero element.
    \label{asm.identifiable-ou}
\end{assumption}

For the process \eqref{eq.var(1)}, it can be shown that:
\begin{equation*}
    F = \Gamma(-1) \Gamma(0)^{-1},
\end{equation*}
where $\Gamma(-1):=\mathbb{E} X_{\delta_c} X_0^\top$ and $\Gamma(0):= EX_0 X_0^\top$ are auto-covariance matrices. Hence, we can estimate $F$ by:
\begin{equation}
    \hat{F} = C_N(-1)C_N(0)^{-1},\label{eq.hatF}
\end{equation}
where $C_{N}(-1):=\frac{1}{N_c}\sum_{j\in J^c} X_{\delta_c}(j) X_0(j)^\top$ and $C_{N}(0):=\frac{1}{N_c}\sum_{j\in J^c} X_0(j) X_0(j)^\top$ are sample covariance matrices. 

To obtain an estimate of $\Phi$ from $\hat{F}$, we can consider the Taylor approximation:
\begin{equation*}
    F = e^{\Phi \delta_c} = \sum_{i=0}^\infty \frac{(\Phi\delta_c)^i}{i!} = I + \Phi\delta_c + \frac{\Phi\delta_c}{2}(F-I) + \eta, \quad \text{where} \,\,\, \eta=\sum_{i=3}^\infty \frac{(i-2)(\Phi\delta_c)^i}{2 i!}
\end{equation*}
Consequently, we have:
\begin{equation*}
    \Phi = \frac{2}{\delta_c}(F-I)(F+I)^{-1} - \frac{2}{\delta_c} \eta (F+I)^{-1}.
\end{equation*}

After neglecting terms smaller than $O(\delta_c^2)$, we obtain an estimator:
\begin{equation*}
    \hat{\Phi} = \frac{2}{\delta_c}(\hat{F}-I)(\hat{F}+I)^{-1}.
\end{equation*}

In practice, to achieve consistency, we introduce two refinements to the estimator. \emph{First}, we truncate the singular values of the matrix $(I+\hat{F})$ into the interval $[1,3]$. This truncation bounds the norms of $\hat{F}$ and $(I+\hat{F})^{-1}$, which is essential for proving the convergence of $\mathbb{E}\|\hat{\Phi}-\Phi\|^p_\mathrm{sp}$  ($p\geq 1$). We justify this step by showing that the singular values of $(I+\hat{F})$ fall into $[1,3]$ almost surely (Lem.~\ref{lem.order-P(E)}). \emph{Second}, we truncate the singular values of $(I-\hat{\Phi})$ into $[\frac{1}{u},u]$, where $u$ is a threshold that goes to infinity with some appropriate rate. This modification controls the norm of $\hat{\Phi}$ and $(I-\hat{\Phi})^{-1}$, ensuring the convergence of the plug-in estimators. We summarize the estimation of $\Phi$ in Alg.~\ref{alg:estimation-ou} (line 5-9).

Furthermore, the covariance matrix $\Omega$ of the VAR(1) process \eqref{eq.var(1)} can be estimated by the sample covariance of the fitting residual:
\begin{equation*}
    \hat{\Omega} = \frac{1}{N_c-d-1} \sum_{j\in J^c} \left\{X_t(j)-\hat{F} X_{t-1}(j)\right\}\left\{X_t(j)-\hat{F} X_{t-1}(j)\right\}^\top.
\end{equation*}

\cite{PHILLIPS1973351} introduced a connection between $\Omega$ and the diffusion matrix $\Sigma$:
\begin{equation*}
    \mathrm{vec}\,\Sigma = \left\{F\otimes F - I\otimes I\right\}^{-1} \left\{\Phi\otimes I + I\otimes \Phi\right\} \mathrm{vec}\,\Omega,
\end{equation*}
where $\otimes$ denotes the Kronecker product. Therefore, we can estimate $\Sigma$ by:
\begin{equation}
    \mathrm{vec}\,\hat{\Sigma} = \left\{\hat{F}\otimes \hat{F} - I\otimes I\right\}^{-1} \left\{\hat{\Phi}\otimes I + I\otimes \hat{\Phi}\right\} \mathrm{vec}\,\hat{\Omega}.
    \label{eq.def-hat-Sigma}
\end{equation}

\begin{algorithm}[htp]
\caption{OU model estimation}
\label{alg:estimation-ou}
\begin{algorithmic}[1]
    \State \textbf{input:} samples $\{X_j(t)|j\in J^c, t=0,\delta_c,...\}$.
    \State \textbf{option:} threshold $u>0$.
    \State Compute $\hat{F}$ with \eqref{eq.hatF} and compute $\hat{\Sigma}$ with \eqref{eq.def-hat-Sigma}. 
    \State Apply the singular value decomposition (SVD): $I+\hat{F}=USV^\top$.
    \State Truncate the elements of $S$ into $[1,3]$. Call the resulting matrix $\bar{S}$.
    \State Let $\bar{F}=U\bar{S}V^\top - I$ and let $\bar{\Phi}= \frac{2}{\delta_c}(\bar{F}-I)(\bar{F}+I)^{-1}$.
    \State Apply the SVD: $I-\bar{\Phi}=U_1 S_1 V_1^\top$.
    \State Truncate the elements of $S_1$ into $[\frac{1}{u},u]$. Call the resulting matrix $\Tilde{S}_1$.
    \State Let $\Tilde{\Phi}=I-U_1 \Tilde{S}_1 V_1^\top$.
    \State\textbf{outputs:} $\Tilde{\Phi},\hat{\Sigma}$.
    \end{algorithmic}
\end{algorithm}

For consistency of the estimation, we first show that:
\begin{proposition}
    Assume Asm.~\ref{asm.identifiable-ou}. Let $p\geq 1$, then there exists a positive constant $K=K(p,\Phi,\Sigma)$ such that:
    \begin{equation*}
        \mathbb{E} \|\Phi-\Tilde{\Phi}\|^p_\mathrm{sp} \leq K u^p \left(\delta_c^{-p} N_c^{-\frac{p}{2}} + \delta_c^{2p} \right)
    \end{equation*}
    for sufficiently large $N_c, u$ and sufficiently small $\delta_c$.
    \label{prop.consistency-ou}
\end{proposition}

Furthermore, we have that:
\begin{proposition}
    Assume Asm.~\ref{asm.identifiable-ou}, then we have $\hat{\Sigma} \overset{\mathcal{P}}{\to} \Sigma$ if $N_c\to \infty, \delta_c\to 0, u\to \infty$ and $\exists p\geq 1$ such that $u^p \{\delta_c^{-p} N_c^{-\frac{p}{2}} + \delta_c^{2p} \} \to 0$.
    \label{prop.consistency-hat-Sigma}
\end{proposition}

\subsection{Estimation of the projection processes}

We then estimate the projection processes $\Pi$ and $\mu$. To this end, we first show an optimal filtering equation related to $\Pi$ and $\mu$.

\begin{proposition}
    \label{prop.form-of-filtering-eq-in-OU}
    The processes $m_t = \mathbb{E}\left\{X_{V\backslash C}(t)|\mathcal{F}_t\right\}, y_t = \frac{1}{\sigma^2}\mathrm{Var}\left\{X_{V\backslash C}(t)|\mathcal{F}_t\right\}$ satisfy the following system of filtering equations:
    \begin{subequations}
    \label{eq.pi-t}
        \begin{align}
            &dm_t = \{A_t X_C(t) + B_t m_t\} dt + C_t d X_C(t),\\
            &\dot{y}_t = \Phi_{V\backslash C} y_t + y_t \Phi_{V\backslash C}^\top + I - y_t \Phi_{C, V\backslash C}^\top \Phi_{C, V\backslash C} y_t,
        \end{align}
    \end{subequations}
    where we denote $\Phi_{V\backslash C}$
 in short of $\Phi_{V\backslash C,V\backslash C}$ and             \begin{equation*}
        A_t := \Phi_{V\backslash C, C} - y_t  \Phi_{C, V\backslash C}^\top  \Phi_{C}, \quad B_t := \Phi_{V\backslash C} - y_t \Phi_{C, V\backslash C}^\top\Phi_{C, V\backslash C}, \quad C_t := y_t \Phi_{C, V\backslash C}^\top.
    \end{equation*}

    Moreover, we have $X_0\sim\mathcal{N}(0,\sigma^2 \Upsilon)$, with $\Upsilon=(I-\Phi)^{-1} (I-\Phi^\top)^{-1}$. Therefore, the initial values of (\ref{eq.pi-t}a)-(\ref{eq.pi-t}b) are (respectively):
    \begin{equation}
        m_0 = \Upsilon_{V\backslash C,C} \Upsilon_C^{-1} X_C(0), \quad y_0 = \Upsilon_{V\backslash C} - \Upsilon_{V\backslash C,C} \Upsilon_C^{-1} \Upsilon_{C,V\backslash C}.
        \label{eq.m0-y0}
    \end{equation}
\end{proposition}

We also have the following result that ensures the solution obtained from forwarding (\ref{eq.pi-t}a)-(\ref{eq.pi-t}b) is exactly the projection progress $m_t$ and the covariance process $y_t$.

\begin{proposition}
    The system of equations (\ref{eq.pi-t}a)-(\ref{eq.pi-t}b), subject to the initial condition on $m_0, \gamma_0$, has a unique, continuous, $\mathcal{F}_t$-adapted solution for any $t\in [0,T]$.
    \label{prop.filtering-eq-uniqueness-solution}
\end{proposition}

For the solution, we first note that (\ref{eq.pi-t}b) is a symmetric Riccati Differential Equation (RDE), which can be solved using the negative exponential method \cite[Sec. 2.3]{choi1990survey}. Specifically, denote 
\begin{equation*}
    M = \begin{bmatrix}
        -\Phi_{V\backslash C}^\top & \Phi_{C, V\backslash C}^\top \Phi_{C, V\backslash C}\\
        I & \Phi_{V\backslash C}
     \end{bmatrix} \quad \text{and} \quad  W^{-1} M W = \begin{bmatrix}
        \Lambda & 0 \\
        0 & -\Lambda
     \end{bmatrix},
\end{equation*}
which are a Hamiltonian matrix and its Jordan decomposition, respectively, with $\Lambda$ being the real Jordan canonical form with right-half-plane eigenvalues. Write the matrix $W$ in the form of a $2\times 2$ block matrix:
\begin{equation*}
    W=\begin{bmatrix}
        W_{11} & W_{12}\\
        W_{21} & W_{22} 
     \end{bmatrix},
\end{equation*}
and let $R := -(W_{22}-y_0 W_{12})^{-1} (W_{21}-y_0 W_{11})$. Then, \cite{choi1990survey} showed that:
\begin{align}
    y_t = (W_{21} + W_{22} e^{-t\Lambda} R e^{-t\Lambda}) (W_{11}+W_{12} e^{-t\Lambda} R e^{-t\Lambda})^{-1} \quad \forall t\in [0,T].\label{eq.compute-solution-rde}
\end{align}

The matrix exponential $e^{-t\Lambda}$ of the real Jordan form $\Lambda$ can be computed with closed forms. Therefore, we can compute the solution of (\ref{eq.pi-t}b) at any $t$ with \eqref{eq.compute-solution-rde}. After solving (\ref{eq.pi-t}b), we can solve (\ref{eq.pi-t}a) with the Euler-Maruyama method. More specifically, let $\{\bar{m}_{k\delta}\}$ ($\{\bar{m}_k\}$ for brevity) be the Euler-Maruyama approximation for the solution of (\ref{eq.pi-t}a) that is defined by the following difference equation:
\begin{equation*}
    \bar{m}_k - \bar{m}_{k-1} = (A_{k-1} X_{C,k-1} + B_{k-1} \bar{m}_{k-1}) \delta + C_{k-1} (X_{C,k} - X_{C,k-1}), \,\, \bar{m}_0 = \Upsilon_{V\backslash C,C} \Upsilon_C^{-1} X_{C,0}
\end{equation*}
where we denote $A_k:=A_{k\delta}, B_k:=B_{k\delta}, C_k:=C_{k\delta}$ and $ X_{C,k}:=X_C(k\delta)$.

Plugging-in the estimator $\Tilde{\Phi}$ obtained in Alg.~\ref{alg:estimation-ou}, we then obtain an estimator for $m_t$, namely $\hat{m}_t=\hat{m}_{k\delta}$ for $k\delta\leq t< (k+1)\delta$, with the sequence $\{\hat{m}_k\}$ defined by:
\begin{equation}
    \hat{m}_k - \hat{m}_{k-1} = (\Tilde{A}_{k-1} X_{C,k-1} + \Tilde{B}_{k-1} \hat{m}_{k-1}) \delta + \Tilde{C}_{k-1} (X_{C,k} - X_{C,k-1}), \,\, \hat{m}_0 = \Tilde{\Upsilon}_{V\backslash C,C} \Tilde{\Upsilon}_C^{-1} X_{C,0},\label{eq.hat-m-t}
\end{equation}
where $\Tilde{A}_t, \Tilde{B}_t, \Tilde{C}_t, \Tilde{\Upsilon}$ are the plugins of $A_t,B_t,C_t,\Upsilon$, respectively, 

After solving $\hat{m}$, we can take the $\alpha$-th dimension of $\hat{m}$ to be $\hat{\Pi}$. To estimate $\mu$, we note: 
\begin{align*}
    \mu_t :=\mathbb{E}\left\{\lambda_\beta(t)|\mathcal{F}_t\right\} &= \mathbb{E}\left\{\Phi_{\beta,V\backslash C} X_{V\backslash C}(t) + \Phi_{\beta,C} X_{C}(t) |\mathcal{F}_t\right\} \\
    &= \Phi_{\beta,V\backslash C} m_t + \Phi_{\beta,C} X_C(t).
\end{align*}

Hence, we can estimate $\mu_t$ by $\hat{\mu}_t = \hat{\mu}_{k\delta}$ for $k\delta\leq t<(k+1)\delta$, with $\hat{\mu}_k$ defined by:
\begin{equation*}
    \hat{\mu}_k = \Tilde{\Phi}_{\beta,V\backslash C} \hat{m}_k + \Tilde{\Phi}_{\beta,C} X_{C,k}.
\end{equation*}

We summarize the estimation procedure in Alg.~\ref{alg:estimation-pi-mu}.

\begin{algorithm}[htp]
\caption{Optimal filtering estimators}
\label{alg:estimation-pi-mu}
\begin{algorithmic}[1]
    \State \textbf{inputs:} training samples $\{X_j(t)|j\in J^c, t=0,\delta_c,...\}$, evaluation samples $\{X_j(t)|j\in J, t=0,\delta,...\}$, cause, effect and the conditioning set $\alpha,\beta \in V, C\subset V$.
    \State Compute $\Tilde{\Phi}$ with the training samples via Alg.~\ref{alg:estimation-ou}.
    \For{$j\in J$}
        \For{$k=1,...,n$}
            \State Compute $\hat{y}_{k\delta}$ by plugging  $\Tilde{\Phi}$ into \eqref{eq.compute-solution-rde}.
            \State Compute $\hat{m}_j(k\delta)$ with \eqref{eq.hat-m-t}.
            \State Let $\hat{\Pi}_{j}(k\delta)$ be the $\alpha$-th dimension of $\hat{m}_j(k\delta)$.
            \State Let $\hat{\mu}_j(k\delta) = \Tilde{\Phi}_{\beta,V\backslash C} \hat{m}_j(k\delta) + \Tilde{\Phi}_{\beta,C} X_{C,j}(k\delta).$
        \EndFor
    \EndFor
    \State \textbf{outputs:} $\{\hat{\Pi}_j(k\delta),\hat{\mu}_j(k\delta) | j\in J, k=0,...,n\}$.
    \end{algorithmic}
\end{algorithm}

For consistency of the estimation, we can show that:

\begin{proposition}
    Assume Asm.~\ref{asm.identifiable-ou}, then for $p>2$, we have:
    \begin{equation*}
        \max \left\{\|\mu-\hat{\mu}\|_p,\, \|G-\hat{G}\|_p, \, N^{\frac{1}{2}}\|\mu-\hat{\mu}\|_2 \|G-\hat{G}\|_2\right\} \to 0
    \end{equation*}
    as $N_c\to\infty, \delta_c=\Theta(N_c^{-\frac{1}{6}}), u = O(\frac{1}{4T} \log\log N_c^{\frac{1}{6}}), \delta=O(N_c^{-\frac{2}{3}})$. In this regard, the requirements in Asm.~\ref{asm.Lp-consistency-estimator} are satisfied.
    \label{prop.consistency-projection-estimator}
\end{proposition}

We sketch the outline of the proof below. The key is to control the error $\|m-\hat{m}\|_p$, with the statements in Prop.~\ref{prop.consistency-projection-estimator} follow. We can decompose the error into:
\begin{equation}
    \|m-\hat{m}\|_p \leq \|m-\Tilde{m}\|_p + \|\Tilde{m}-\hat{m}\|_p,
    \label{eq.decompose-m-hatm}
\end{equation}
with $\Tilde{m}$ being the continuous process satisfying:
\begin{equation*}
    d\Tilde{m}_t = \{\Tilde{A}_t X_C(t) + \Tilde{B}_t \Tilde{m}_t\}dt + \Tilde{C}_t d X_C(t), \quad \tilde{m}_0 = \hat{m}_0.
\end{equation*}

On the right-hand side of \eqref{eq.decompose-m-hatm}, the first term $\|m-\Tilde{m}\|_p$ represents the estimation error of plugging $\Tilde{\Phi}$ in $\Phi$. We can show that $\|m-\Tilde{m}\|_p^p=o\{\exp(e^{2puT})\mathbb{E}\|\Phi-\Tilde{\Phi}\|^p\}$ (Lem.~\ref{lem.m-tildem-consistency}), where the term $\exp(e^{2puT})$ originates from the Lipschitz constant of the mapping $(t,\Phi)\mapsto m(t,\Phi)$ on $[0,T]\times [\frac{1}{u},u]$. Combine with $\mathbb{E} \|\Phi-\Tilde{\Phi}\|^p = O\left\{u^p \left(\delta_c^{-p} N_c^{-\frac{p}{2}} + \delta_c^{2p} \right)\right\}$ shown in Prop.~\ref{prop.consistency-ou}, we can show the order of $\|m-\Tilde{m}\|_p$ and make it vanish by managing $N_c,\delta_c$ and $u$. For the second term $\|\Tilde{m}-\hat{m}\|_p$, which is the approximation error of the Euler-Maruyama method, we can show that $\|\Tilde{m}-\hat{m}\|_p\leq K_{u,T} \delta^{\frac{p}{2}}$, similar to the $\delta^{\frac{p}{2}}$-rate of strong convergence result of the Euler-Maruyama method. Hence, by setting $\delta$ to be sufficiently small, we can control the error $\|\Tilde{m}-\hat{m}\|_p$. 
\section{Experiment}
\label{sec.exp}

\subsection{Local independence test}
\label{sec.exp-CL-test}
In this section, we evaluate our test on synthetic data. For all experiments, we set the significance level $\alpha=0.05$ and the number of folds $K=3$.

We first investigate the type I error rate and recall of the test under different sample sizes and scales. For data generation, we first sample the elements of $\Phi$ from $\mathrm{Bern}(0.3)$ to simulate the presence and absence of causal relationships. We then replace the nonzero entries of $\Phi$ with independent realizations of a $\mathrm{Unif}(0,1)$ r.v. to simulate various causal strengths. To ensure $(I-\Phi)$ is nonsingular, the diagonal of $\Phi$ is set to $2$. To simulate data from $\mathbb{H}_0$ and $\mathbb{H}_1$, $\Phi_{\beta \alpha}$ is set to $0$ and $0.75$, respectively. In addition, we use $\sigma=1,\delta=0.01,T=1$. With these parameters, we use the Euler-Maruyama scheme to sample data from \eqref{eq.multivariate-ou} and \eqref{eq.initial-value-linear-model}. We consider different sample sizes $N_0\in \{50,100,150,200,250\}$ and scales $d\in \{10,15\}$. For each setting, we replicate $5$ random $\Phi$s. For each $\Phi$, we repeat $50$ times to estimate the type I error rate and recall, respectively.

We compare our proposed test with the following procedures: the Granger's test of noncausality \cite{granger1969investigating} and the test based on linear regression (denoted as LinReg), which applies our test with $\hat{m}$ estimated by historical linear regression.

The results are shown in Fig.~\ref{fig:exp-1}. We see that the empirical size of our test are close to the significance level, which means the type I error can be effectively controlled. Moreover, our test achieves high recall, particularly when implemented with cross-fitting. This demonstrates the power of the proposed procedure. In contrast, the Granger's test, which quantifies the deviation of $F_{\beta \alpha}=(e^{\Phi\delta})_{\beta \alpha}$ from zero, suffers from low recall, since $F_{\beta \alpha}\to 0$ as $\delta\to 0$. We also observe that the linear regression test loses control of the type I error, which can be attributed to the violation of the consistency Asm.~\ref{asm.Lp-consistency-estimator}.

\begin{figure}[htp]
    \centering
    \includegraphics[width=0.75\linewidth]{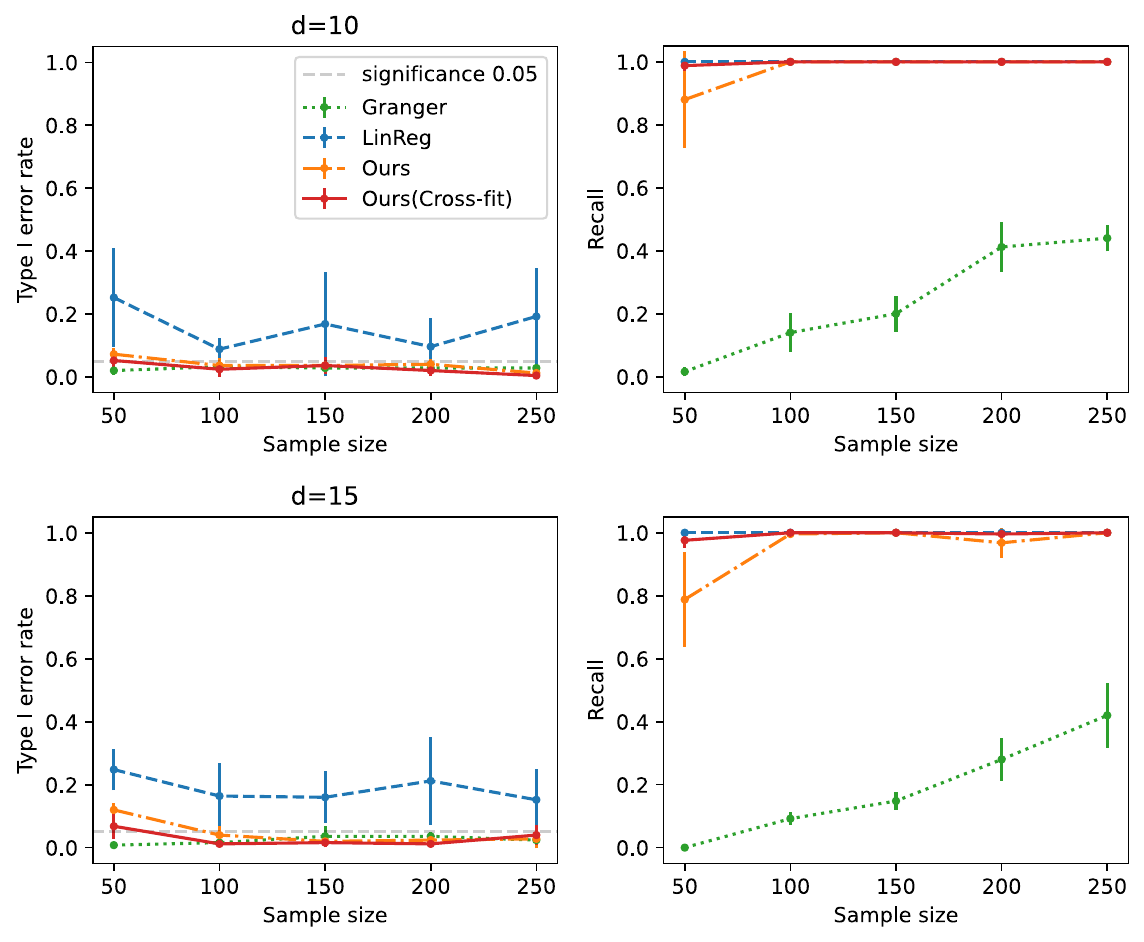}
    \caption{Type I error rate and recall of local independence tests with $d=10, 15$.}
    \label{fig:exp-1}
\end{figure}

\begin{figure}[htp]
    \centering
    \includegraphics[width=0.8\linewidth]{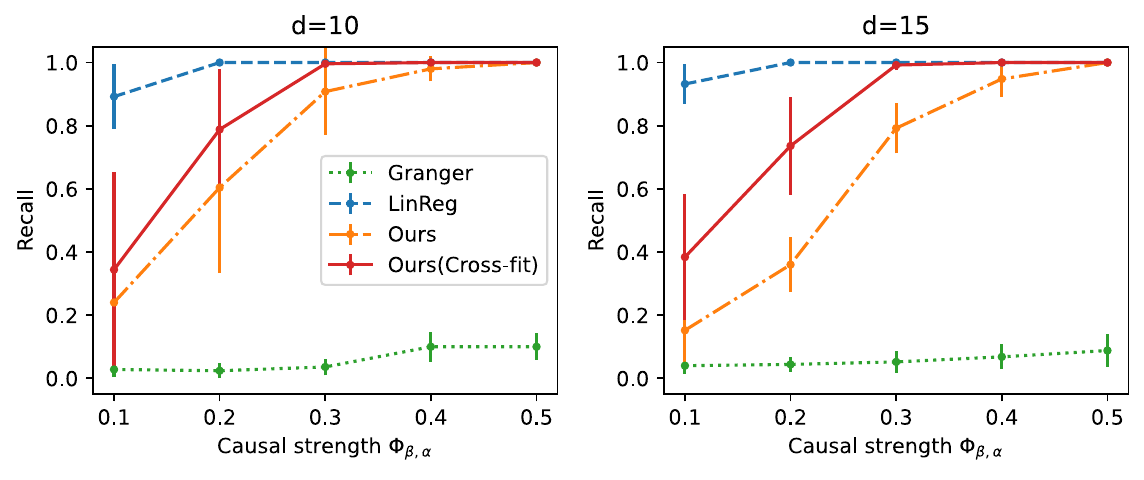}
    \caption{Recall of local independence tests under different causal strengths.}
    \label{fig:exp-2}
\end{figure}

We then study the power of local independence tests under different causal strengths. To this end, we vary the parameter $\Phi_{\beta \alpha}\in \{0.1,0.2,0.3,0.4,0.5\}$ to simulate different strengths, with the sample size $N_0=150$ kept fixed. The results are shown in Fig.~\ref{fig:exp-2}. As expected, the recalls of all procedures increase as $\Phi_{\beta \alpha}$ increases. In particular, we see that the power of our test with cross-fitting increases most rapidly, reaching $1$ when $\Phi_{\beta \alpha}=0.3$. This can be attributed to the fact that the cross-fitting recovers the full root-$N_0$ efficiency, which leads to a more powerful test.

Finally, we verify the necessity of requiring $\delta\to 0$ as $N\to\infty$, as discussed in Prop.~\ref{prop.consistency-projection-estimator}. In details, in Fig.~\ref{fig:exp-5} (left), we report the type I error rate as $N_0$ increases from $\{10^2,10^3,10^4,1.5\times 10^4\}$ and $\delta=0.01$ kept fixed. We observe that the type I error of the test loses control, which can be attributed to the fact that the discretization error $\sqrt{N} \|\Tilde{m}-\hat{m}\|_2 \to\infty$ if $N\to\infty$ and $\delta$ is kept unchanged. Then, in Fig.~\ref{fig:exp-5} (right), we further let $\delta\in \{0.01,0.009,0.008,0.007,0.006\}$ converge to zero. As shown, the type I error rate gradually decreases to the significance level.  These observations demonstrate the validity of our study on the consistency of the filtering estimator.

\begin{figure}[htp]
    \centering
    \includegraphics[width=0.8\linewidth]{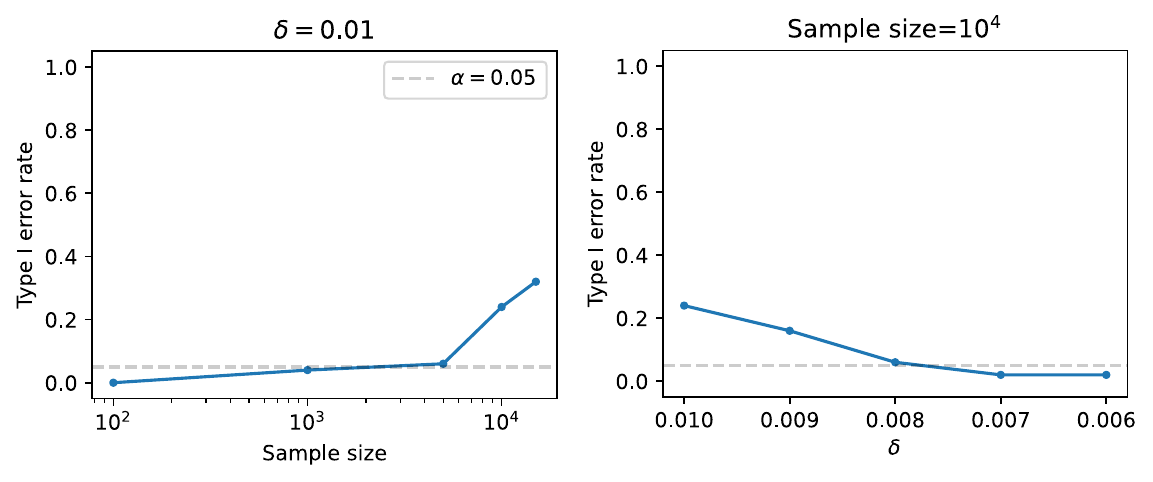}
    \caption{Type I error rate of the proposed test. Left: the sample size $N_0$ goes to infinity while the sampling interval $\delta$ kept fixed. Right: the sampling interval $\delta$ goes to zero.}
    \label{fig:exp-5}
\end{figure}

\subsection{Application to causal discovery}

In this section, we consider a real application of the proposed test in recovering the Local Independence Graph (LIG) \cite{didelez2008graphical}. Consider an Itô process $X$ with a diagonal diffusion matrix, a LIG is a graph $G=(V,E)$, where the nodes $V=\{1,...,d\}$ correspond to the processes $X_1,...,X_d$, and the set of edges $E$ contains $\alpha\to\beta$ if and only if the local independence $\alpha\not\to\beta|V\backslash \alpha$ does not hold. In $G$, two subsets of vertices $A$ and $B$ are $\mu$-separated by a subset of vertices $C$ (denoted as $A\ind_\mu B|C$) if every path between them is blocked by $C$. Please refer to \cite[Def. 4]{mogensen2018causal} for a formal definition. We say $X$ satisfies the \emph{global Markov condition} with respect to $G$ if for all $A,B,C\subseteq V$, $A\ind_\mu B|C$ if and only if $\alpha\not\to\beta |C$ for every $\alpha\in A,\beta \in B$. Under the global Markov condition, the LIG gives a causal representation of a dynamical system \cite{commenges2009general}. In the following, we will recover the LIG from real data by iteratively applying the proposed local independence test. Additional simulation results are relegated to Appx.~\ref{app.experiment}.

We consider the resting-state functional Magnetic Resonance Imaging (rs-fMRI) data from the Human Connectome Project (HCP) \cite{van2013wu}. The subjects included in this dataset underwent $4$ sessions of $15$-minute multi-band accelerated (TR=$0.72$s) rs-fMRI scans on a 3 Tesla scanner with an anisotropic spatial resolution of $2$ mm, for a total of $4800$ volumes \cite{smith2013resting}. Preprocessing details, which include corrections for spatial distortions and head motion, registration to the T1-weighted structural image, resampling to $2$ mm Montreal Neurological Institute (MNI)
space, global intensity normalization, and high-pass filtering, can be found in \cite{glasser2013minimal}. In total, we use $N=1070$ subjects from the \href{https://www.humanconnectome.org/study/hcp-young-adult/document/hcp-young-adult-2025-release}{HCP-Young Adult 2025 Release}. For all subjects, we parcellate the brain into $44$ cortical regions and $19$ subcortical regions, according to the HCP-MMP-1 atlas \cite{glasser2016multi} (please refer to their supplementary material for details).

\begin{figure}[htp]
    \centering
    \includegraphics[width=\linewidth]{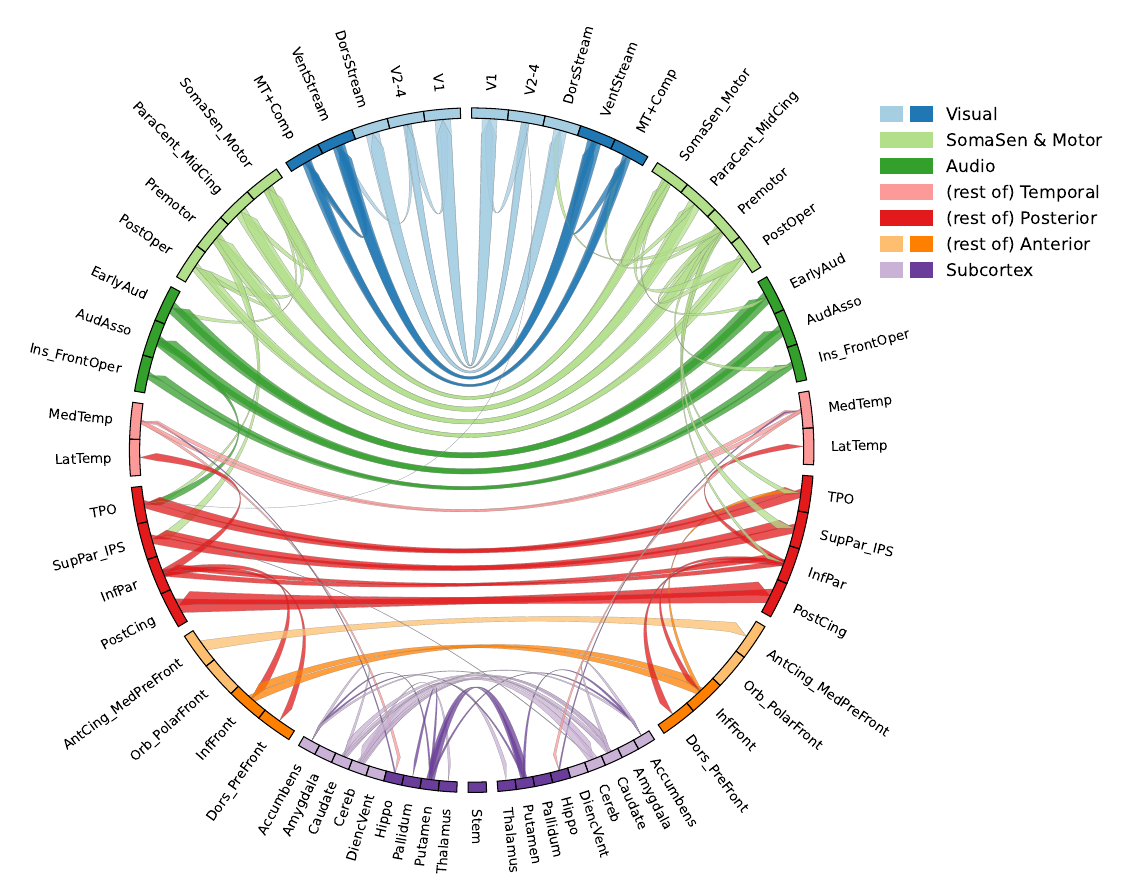}
    \caption{The recovered causal graph from resting-state fMRIs. We have omitted self-loops for clarity. The upper $44$ regions are cortical regions and the bottom $19$ regions (colored in purple) are subcortical regions. The cortical regions are further grouped into $6$ subgroups: Visual, Somatosensory-Motor, Auditory, (rest of) Temporal, (rest of) Posterior, and (rest of) Anterior cortex, according to \cite{glasser2016multi}. The left and right sides represent the left and right hemispheres, respectively. The width of each edge is proportional to the strength of the causal interaction obtained from $\Tilde{\Phi}$.}
    \label{fig:rsfMRI-result}
\end{figure}

The recovered causal structure is shown in Fig.~\ref{fig:rsfMRI-result}. For illustration, we omit all self-loops. We also group the $22$ cortical regions in each hemisphere into $6$ subgroups: Visual, Somatosensory-Motor, Auditory, (rest of) Temporal, (rest of) Posterior, and (rest of) Anterior cortex, according to \cite{glasser2016multi}. To interpret, we see that the graph is symmetric, with dominating links in symmetrically-located regions and regions in the same subgroup. In particular, we observe that the visual cortex and the somatosensory-motor cortex are two subgroups with high connectivity, which can be related to the activation of visual \cite{wang2008spontaneous} and sensory functions \cite{colonnese2012spontaneous,wang2013relationship} in rs-fMRIs. Besides, we note a connection from the inferior partial cortex ($\mathrm{InfPar}$) to the inferior frontal cortex and the dorsolateral prefrontal cortex ($\mathrm{InfFront}, \mathrm{Dors\_PreFront}$), which are known to be related to working memory \cite{baldo2006role,friedman1994coactivation}. The connection between $\mathrm{InfPar}$ and the lateral temporal cortex ($\mathrm{LatTemp}$) can be attributed to the fact that they are both components of the Default Mode Network (DMN) \cite{yusoff2018resting}. Finally, we can observe a clear distinction in causal links among cortical regions and subcortical ones (colored in purple), except for the bidirected edges between hippocampus ($\mathrm{Hippo}$) and the medial temporal lobe ($\mathrm{MedTemp}$), which is also found in \cite[Fig. 4]{arab2025whole}. Nonetheless, we note the graph may present some false negative results, e.g., the edge from the primary visual area ($\mathrm{V1}$) to other visual areas ($\mathrm{V1}$-$\mathrm{V4}$) is missing. This may be due to the low Signal-Noise-Ratio (SNR) of the rs-fMRI data \cite{smith2013resting}.

\begin{table}[b]
    \caption{Causal graph stability across three random dataset splits (S1, S2, S3).}
    \centering
    \begin{tabular}{c c c}
        \toprule
        Splits & Structural Hamming distance & Edge difference\\
        \midrule
        S1 v.s. S2 & 1 & Amygdala(R) $\to$ TPO(R)\\
        \hline
        \multirow{3}{*}{S2 v.s. S3} & \multirow{3}{*}{3} & Amygdala(R) $\to$ TPO(R)\\
        & & InfFront(R)$\to$TPO(R)\\
        & & SomaSen\_Motor(R)$\to$ EarlyAud(R) \\
        \hline
       \multirow{2}{*}{S3 v.s. S1} & \multirow{2}{*}{2} & InfFront(R)$\to$TPO(R) \\
       & & SomaSen\_Motor(R)$\to$ EarlyAud(R)\\
       \bottomrule
    \end{tabular}
    \label{tab:stability-fMRI}
\end{table}

We also analyze the stability of the recovered causal relationships across different dataset splits. In details, we consider three random splits (S1, S2, S3) and show the Structural Hamming Distances (SHDs) between graphs recovered from these splits in Tab.~\ref{tab:stability-fMRI}. As shown, the SHDs are consistently small, showing the graphs differ by only 1–3 edges. This result confirms the stability of our cross-fitting test, with differing edges may be due to the finite sample error.

\section{Conclusion}
\label{sec.conclusion}

In this paper, we propose a conditional local independence test for Itô processes. Our test is based on the semimartingale decomposition, which allows us to establish the martingale property of the stochastic integral process under the null hypothesis. We estimate our test statistic through sample splitting and the optimal filtering equation. We demonstrate the asymptotic level and power of our test. We verify our test on synthetic data and apply it for dynamic causal discovery in resting-state fMRI. In the future, we plan to investigate local independence tests based on observations from a single trajectory of the process. Additionally, we are interested in extending optimal filtering estimation to encompass general diffusion processes.
\bibliographystyle{agsm}
\bibliography{reference}

\newpage
\begin{center}
{\huge\bf Appendix}
\end{center}
\appendix
\section{Proof of Section~\ref{sec.setup}}
\label{appx.proof-sec-setup}
\subsection{Proof of Prop.~\ref{prop.is-martingale}}

\textbf{Proposition \ref{prop.is-martingale}.} \emph{Assume \eqref{eq.regularity-lambda-bounded-inL2}, then $X_\beta$ is an $\mathcal{F}_t$ (resp. $\mathcal{G}_t$)-semimartingale with the decomposition \eqref{eq.F-semimartingale} (resp. \eqref{eq.G-semimartingale}), where $M_t$ (resp. $\mathbf{M}_t$) is a continuous, square-integrable $\mathcal{F}_t$ (resp. $\mathcal{G}_t$)-martingale.}

\begin{proof}
    We show $X_\beta$ is an $\mathcal{F}_t$-semimartingale with the decomposition \eqref{eq.F-semimartingale}. The proof for the $\mathcal{G}_t$-semimartingale can be similarly derived.

    We first show that the process $t\mapsto \int_0^t \mu_s ds$ is a cadlag adapted process with bounded variation. The process is $\mathcal{F}_t$-adapted and continuous. To show it has a bounded variation, note that:
    \begin{equation*}
        V_0^T\left(t\mapsto \int_0^t \mu_s ds\right) = \int_0^T |\mu_s| ds,
    \end{equation*}
    where $V_a^b(f)$ denotes the total variation of $f$ on the interval $[a,b]$. Under \eqref{eq.regularity-lambda-bounded-inL2}, we have $\mathbb{E} \int_0^T |\mu_s| ds<\infty$ (Lem.~\ref{prop.tech-lem-L2-boundedness-under-(4)}), which means the process has bounded variation almost surely.

    Next, we show that the process $M_t$ is a continuous, square-integrable $\mathcal{F}_t$-martingale. First, by the fact that both $X_\beta(t)$ and $\mu_t$ are $\mathcal{F}_t$-adapted ($\beta\in C$), $M_t$ is $\mathcal{F}_t$-adapted. Furthermore, we have $M_t$ is square-integrable under \eqref{eq.regularity-lambda-bounded-inL2} (Lem.~\ref{prop.tech-lem-L2-boundedness-under-(4)}).
    
    In the following, we show $\mathbb{E}(M_t|\mathcal{F}_u) = M_u$ for $0\leq u\leq t$. Denote
    \begin{equation*}
        \Lambda_t = \int_0^t \lambda_\beta(s)ds, \quad A_t = \int_0^t \lambda_\beta(s) - \mathbb{E}\{\lambda_\beta(s)|\mathcal{F}_s\} ds, \quad B_t = \sum_{j=1}^d \Sigma_{\beta j} W_j(t),
    \end{equation*}
    so that $X_\beta(t)=X_\beta(0) + \Lambda_t + B_t$ and $M_t = A_t + B_t$.
    
    With these notations, we first have:
    \small
    \begin{align*}
    \allowdisplaybreaks
        \mathbb{E}(A_t|\mathcal{F}_u) &= \mathbb{E}\left[\int_0^t \lambda_\beta(s) - \mathbb{E}\{\lambda_\beta(s)|\mathcal{F}_s\} ds \Big| \mathcal{F}_u \right] \\
        &= \int_0^t \mathbb{E}\{\lambda_\beta(s)|\mathcal{F}_u\} - \mathbb{E}\left[ \mathbb{E}\{\lambda_\beta(s)|\mathcal{F}_s\}|\mathcal{F}_u \right] ds \\
        &= \int_0^u \mathbb{E}\{\lambda_\beta(s)|\mathcal{F}_u\} - \mathbb{E}\left[ \mathbb{E}\{\lambda_\beta(s)|\mathcal{F}_s\}|\mathcal{F}_u \right] ds + \int_u^t \mathbb{E}\{\lambda_\beta(s)|\mathcal{F}_u\} - \mathbb{E}\left[ \mathbb{E}\{\lambda_\beta(s)|\mathcal{F}_s\}|\mathcal{F}_u \right] ds \\
        & = \int_0^u \mathbb{E}\{\lambda_\beta(s)|\mathcal{F}_u\} - \mathbb{E}\{\lambda_\beta(s)|\mathcal{F}_s\} ds + \int_u^t \mathbb{E}\{\lambda_\beta(s)|\mathcal{F}_u\} - \mathbb{E}\{\lambda_\beta(s)|\mathcal{F}_u\}ds \\
        & = \int_0^u \mathbb{E}\{\lambda_\beta(s)|\mathcal{F}_u\} - \mathbb{E}\{\lambda_\beta(s)|\mathcal{F}_s\} ds, 
    \end{align*}
    \normalsize
    where in the second equation, we use the conditional Fubini's theorem as both $\lambda_\beta(t)$ and $\mathbb{E}\{\lambda_\beta(t)|\mathcal{F}_t\}$ are progressively measurable (Lem.~\ref{lem.fubini-progressive-measurable}). As a result, we have:
    \begin{align*}
        \mathbb{E}(A_t|\mathcal{F}_u) - A_u &= \int_0^u \mathbb{E}\{\lambda_\beta(s)|\mathcal{F}_u\} - \mathbb{E}\{\lambda_\beta(s)|\mathcal{F}_s\} ds  - \int_0^u \lambda_\beta(s) - \mathbb{E}\{\lambda_\beta(s)|\mathcal{F}_s\} ds \\
        &= \int_0^u \mathbb{E}(\lambda_\beta(s)|\mathcal{F}_u) - \lambda_\beta(s) ds = \mathbb{E}(\Lambda_u|\mathcal{F}_u) - \Lambda_u.
    \end{align*}
    Further, observe that:
    \begin{equation*}
        \mathbb{E}\{X_\beta(u)|\mathcal{F}_u\} = \mathbb{E}\{X_\beta(0)+\Lambda_u+B_u|\mathcal{F}_u\} = X_\beta(0) + \mathbb{E}(\Lambda_u|\mathcal{F}_u) + \mathbb{E}(B_u|\mathcal{F}_u).
    \end{equation*}
    On the other hand, since $X_\beta(u)$ is $\mathcal{F}_u$-adapted, we have:
    \begin{equation*}
        \mathbb{E}\{X_\beta(u)|\mathcal{F}_u\} = X_\beta(u)= X_\beta(0) + \Lambda_u + B_u,
    \end{equation*}
    which gives $\mathbb{E}(\Lambda_u|\mathcal{F}_u) - \Lambda_u = B_u - \mathbb{E}(B_u|\mathcal{F}_u)$.  Now, for $0\leq u\leq t$,
    \begin{align*}
        \mathbb{E}(M_t|\mathcal{F}_u) & = \mathbb{E}(A_t+B_t|\mathcal{F}_u) \\
        & = \mathbb{E}(A_t|\mathcal{F}_u) + \mathbb{E}(B_t|\mathcal{F}_u) \\
    & = A_u + \mathbb{E}(\Lambda_u|\mathcal{F}_u) - \Lambda_u + \mathbb{E}(B_t|\mathcal{F}_u) \\
    & = A_u + B_u - \mathbb{E}(B_u|\mathcal{F}_u) + \mathbb{E}(B_t|\mathcal{F}_u) \\
    & = M_u + \mathbb{E}(B_t - B_u|\mathcal{F}_u) = M_u,
    \end{align*}
    where the last equation uses the fact that:
    \begin{align*}
        \mathbb{E}(B_t-B_u|\mathcal{F}_u) &= \sum_{j=1}^d \Sigma_{\beta j} \mathbb{E} \{W_j(t)-W_j(u)|\mathcal{F}_u\} = \sum_{j=1}^d \mathbb{E} \{W_j(t)-W_j(u)\} = 0,
    \end{align*}
    as $\{W_j(t)-W_j(u)\}$ is independent with $\mathcal{F}_u^V$, therefore also independent with $\mathcal{F}_u\subset \mathcal{F}_u^V$.

    We now discuss the continuity of $M_t$. Recall that $M_t:=X_\beta(t) - X_\beta(0) -\int_0^t \mu_s ds$. The sample path of the Itô process $X_\beta(t)$ is continuous. To show $t\mapsto \int_0^t \mu_s ds$ is continuous, it is sufficient to show almost all paths of $\mu$ are in $L_1[0,T]$. Note that $\mu$, as the optional projection of $\lambda_\beta$, is a cadlag \cite[Lem.~7.8, vol. 2, p. 320]{rogers2000diffusions}. Therefore, the sample paths of $\mu$ are cadlag functions on $[0,T]$, which means they are bounded on $[0,T]$ \cite[Lem. 1 in Sec. 12, p. 122]{billingsley2013convergence}, thus integrable on $[0,T]$. This shows that $M_t$ is continuous.
\end{proof}

\subsection{Proof of Prop.~\ref{prop.is-zero-under-h0}}

\textbf{Proposition \ref{prop.is-zero-under-h0}.} \emph{Assume \eqref{eq.regularity-lambda-bounded-inL2}, then under $\mathbb{H}_0$, the process $I=(I_t)$ is a $\mathcal{G}_t$-martingale with $I_0=0$. Thus, we have $\gamma_t=0$ for $t\in[0,T]$.}

\begin{proof}
   We first show $I_t = \int_0^t G_s dM_s$ is a $\mathcal{G}_t$-martingale using Thm.~\ref{thm.tech-lem-integral-process-is-martingale}. Recall that we have shown $M_t$ is a continuous, square-integrable $\mathcal{G}_t$-martingale under $\mathbb{H}_0$ (Prop.~\ref{prop.is-martingale}). According to Thm.~\ref{thm.tech-lem-integral-process-is-martingale}, what's left is to show $G_t$ is a $\cG$-predictable process with $\mathbb{E} \int_0^T |G_t|^2 d[M]_t<\infty$, where $[M]$ is the quadratic variation of $M$ (Def.~\ref{def.tech-lem-quadratic-variation}). 
   
   By the continuity of $X_\alpha(t)$ and that $\Pi(t)$ is also predictable \cite[Thm. 19.6, vol. 2, p. 348]{rogers2000diffusions}, we have $G_t:=X_\alpha(t)-\Pi(t)$ is $\mathcal{G}_t$-predictable. Further, according to Lem.~\ref{lem.tech-lem-[M]_t}, we have $[M]_t=\|\Sigma_\beta\|^2_2 t$, where $\Sigma_\beta:=[\Sigma_{\beta 1}, ..., \Sigma_{\beta d}]^\top$. Therefore, under \eqref{eq.regularity-lambda-bounded-inL2},
   \begin{align*}
        \mathbb{E}\int_0^T |G_t|^2 d[M]_t = \Vert \Sigma_\beta \Vert_2^2 \mathbb{E} \int_0^T |G_t|^2 dt<\infty.
    \end{align*}
    according to Lem.~\ref{prop.tech-lem-L2-boundedness-under-(4)}.
    
    We then show $I_0=0$. Recall that a simple $\mathcal{G}_t$-predictable process $H_t$ is of the form:
    \begin{equation*}
        H_t(\omega) = \sum_i Z_i(\omega) 1_{(t_{i-1},t_i]} (t),
    \end{equation*}
    where $0\leq t_0\leq ... \leq t_n$ and $Z_i$ is a random variable measurable to $\mathcal{G}_t$. For such processes, the stochastic integral is defined as \cite[(25.2), vol. 2, p. 45]{rogers2000diffusions}:
    \begin{equation*}
        (H\cdot M)_t := \int_0^t H_s dM_s := \sum_i Z_i \{M(t \wedge t_i) - M(t \wedge t_{i-1})\},
    \end{equation*}
    where $t \wedge t_i := \min(t,t_i)$. So we have $ (H\cdot M)_0=\sum_i Z_i \{M(0 \wedge t_i) - M(0 \wedge t_{i-1}) = 0$.

    For the (non-simple) process $G$, the integral $(G\cdot M)$ is defined as:
    \begin{equation*}
        (G\cdot M) := \lim_n (G_n\cdot M),
    \end{equation*}
    where $\{G_n\}$ is a sequence of simple predictable processes that converge to $G$. Since for each $n$, $(G_n\cdot M)_0=0$, we have $I_0=(G\cdot M)_0=0$. If $I_t$ is a martingale, then $\gamma_t:=\mathbb{E} I_t = \mathbb{E} I_0 = 0$. This has concluded the proof.
\end{proof}

\subsection{Proof of Prop.~\ref{prop.is-covariance}}

\textbf{Proposition \ref{prop.is-covariance}.} \emph{Assume \eqref{eq.regularity-lambda-bounded-inL2}, then for every $0\leq t\leq T$, we have:
\begin{equation*}
    \gamma_t = \int_0^t \mathrm{cov}(G_s,\boldsymbol{\mu}_s-\mu_s)ds.
\end{equation*}
In particular, $\gamma$ is the zero-function if and only $\mathrm{cov}(G_t,\boldsymbol{\mu}_t-\mu_t)=0$ for almost all $t$.}

\begin{proof}
    Recall $d X_\beta(t) = \boldsymbol{\mu}_t dt + d\mathbf{M}_t$. We have shown $\mathbb{E} \int_0^t G_s d\M_s = 0$ (Prop.~\ref{prop.is-zero-under-h0}). Thus,
    \allowdisplaybreaks
    \begin{align*}
        \mathbb{E} I_t = \mathbb{E} \int_0^t G_s dM_s &= \mathbb{E} \int_0^t G_s d\left\{X_\beta(s) - X_\beta(0) - \int_0^s \mu_u du\right\} \\
        & = \mathbb{E} \int_0^t G_s d X_\beta(s) - \mathbb{E}\int_0^t G_s\mu_s ds \\
        & = \mathbb{E} \int_0^t G_s \boldsymbol{\mu}_s ds +  \mathbb{E} \int_0^t G_s d\mathbf{M}_s - \mathbb{E}\int_0^t G_s\mu_s ds \\
        & = \mathbb{E} \int_0^t G_s(\boldsymbol{\mu}_s - \mu_s) ds \\
        & = \int_0^t \mathbb{E}[G_s(\boldsymbol{\mu}_s - \mu_s)] ds = \int_0^t \mathrm{cov}(G_s, \boldsymbol{\mu}_s - \mu_s) ds,
    \end{align*}
    where the last equation is due to $\mathbb{E} G_t \equiv 0$ for the residual process.
\end{proof}
\section{Proof of Section \ref{sec.asym-ana-ii}}
\label{appx.proof-asym-ana-ii}

\subsection{Proof of Prop.~\ref{prop.convergence-u}}
\label{appx.proof-prop.convergence-u}

\noindent \textbf{Proposition~\ref{prop.convergence-u}.} \emph{Assume \eqref{eq.regularity-lambda-bounded-inL2}, then $U \overset{\mathcal{D}}{\to} U_0$ in $C[0,T]$ as $N\to \infty$, where $U_0$ is a mean-zero continuous Gaussian martingale on $[0,T]$ with variance function $\mathcal{V}$.}

\begin{proof}
    Recall
    \begin{equation*}
        U_t = \frac{1}{\sqrt{N}} \sum_j \int_0^t G_{j,s} d\mathbf{M}_{j,s}
    \end{equation*}

    We prove the proposition by applying the martingale central limit theorem (CLT) (Thm.~\ref{thm.martingale-clt}). We first show that $U$ is a mean-zero, square integrable martingale. Specifically, follow the proof of Prop.~\ref{prop.is-zero-under-h0}, we can show that each component of $U_t$ is a mean-zero, square integrable $\Tilde{\mathcal{G}}^c(j)_t$-martingale under \eqref{eq.regularity-lambda-bounded-inL2}. By independence of the observations for each $j$, we can enlarge the filtration for each term, and conclude that they are also mean-zero, square integrable $\Tilde{\mathcal{G}}^c_t$-martingale. Thus, $U$ is also a mean-zero, square integrable $\Tilde{\mathcal{G}}_t$-martingale.

    Second, apply Thm.~\ref{thm.tech-lem-integral-process-is-martingale} and Lem.~\ref{lem.tech-lem-[M]_t}, we have:
    \begin{equation*}
        \Big[\int_0^t G_s d\mathbf{M}_s\Big]_t = \int_0^t G_s^2 d[\mathbf{M}]_s  = \|\Sigma_\beta\|^2 \int_0^t G_s^2 ds.
    \end{equation*}

    Since the components of $U$ are independent, their co-quadratic variations are zero. Therefore, 
    \begin{equation*}
        [U]_t = \frac{1}{N} \|\Sigma_\beta\|^2 \sum_{j\in J} \int_0^t G_{j,s}^2 ds.
    \end{equation*}

    Under \eqref{eq.regularity-lambda-bounded-inL2}, we have that the random variable $\int_0^t G_s^2 ds$ is integrable (Prop.~\ref{prop.tech-lem-L2-boundedness-under-(4)}). Thus, applying the Law of Large Numbers, we have:
    \begin{equation*}
        [U]_t =\frac{1}{N} \|\Sigma_\beta\|^2 \sum_{j\in J} \int_0^t G_{j,s}^2 ds \overset{\mathcal{P}}{\to} \|\Sigma_\beta\|^2 \int_0^t G_{s}^2 ds := \mathcal{V}_t.
    \end{equation*}

    Applying the martingale CLT (Thm.~\ref{thm.martingale-clt}) then concludes the proof.
\end{proof}

\subsection{Proof of Prop.~\ref{prop.R1-3-zero}}
\label{appx.proof-R1-3-zero}

In this section, we prove Prop.~\ref{prop.R1-3-zero}. For clarity, we divide the proof into three lemmas (Lem.~\ref{lem.convergence-R1}, Lem.~\ref{lem.convergence-R2}, and Lem.~\ref{lem.convergence-R3}) for each of the remainder terms $R_1, R_2$ and $R_3$.

In the proof, we denote $\Tilde{\mathcal{G}}_t^c$ as the smallest right continuous and complete filtration generated by the filtrations $\{\mathcal{G}_{j,t}\mid j\in J^c\}$. 

\begin{lemma}
     Assume Asms.~\ref{asm.Lp-consistency-estimator}, \ref{asm.moment-bound}, then $\sup_{t\in [0,T]}|R_{1,t}| \overset{\mathcal{P}}{\to} 0$ as $N\to \infty, \delta\to 0$.
     \label{lem.convergence-R1}
\end{lemma}

\begin{proof}
    Recall 
    \begin{equation*}
        R_{1,t} = \frac{1}{\sqrt{N}} \sum_{j\in J} \int_0^t G_{j,s} (\mu_{j,s}-\hat{\mu}_{j,s}) ds.
    \end{equation*}

    We first show $R_{1,t}\overset{\mathcal{P}}{\to} 0$ for every $t$. Specifically, we note that:
    \small
    \begin{equation*}
        \mathbb{E}\left\{G_s(\mu_s-\hat{\mu}_s) \mid \Tilde{\mathcal{G}}_1^c\right\} =  \mathbb{E}\left[\mathbb{E}\left\{G_s(\mu_s-\hat{\mu}_s) \mid \mathcal{F}_{s} \lor \Tilde{\mathcal{G}}_1^c\right\} \mid \Tilde{\mathcal{G}}_1^c\right] = \mathbb{E}\left\{\mathbb{E}(G_s|\mathcal{F}_{s}) (\mu_s-\hat{\mu}_s) \mid \Tilde{\mathcal{G}}_1^c \right\} = 0.
    \end{equation*}
    \normalsize 
    
    Since the collection of r.v.'s 
    $\big\{\int_0^t G_{j,s}(\mu_{j,s}-\hat{\mu}_{j,s}) ds \mid j\in J\big\}$
    are i.i.d conditionally on $\Tilde{\mathcal{G}}_1^c$, we have $\mathbb{E}(R_{1,t}|\Tilde{\mathcal{G}}_1^c)=0$. Then, for $p,q>1$ such that $\frac{1}{p}+\frac{1}{q}=\frac{1}{2}$, we have:
    \begin{align*}
        \mathrm{Var}(R_{1,t}) =\mathbb{E}\{\mathrm{Var}(R_{1,t}|\Tilde{\mathcal{G}}_1^c)\} &= \mathbb{E}\Big[ \frac{1}{N} \sum_{j\in J} \mathrm{Var}\big\{\int_0^t G_{j,s}(\mu_{j,s}-\hat{\mu}_{j,s}ds) \mid \Tilde{\mathcal{G}}_1^c\big\} \Big] \\
        &= \mathbb{E} \left(\mathbb{E}\Big[\big\{\int_0^t G_s(\mu_s-\hat{\mu}_s)ds\big\}^2 \mid \Tilde{\mathcal{G}}_1^c \Big]\right) &\\
        &= \mathbb{E} \left\{\int_0^t G_s(\mu_s-\hat{\mu}_s)ds\right\}^2 \overset{(1)}{\leq} T \mathbb{E}\int_0^T |G_s(\mu_s-\hat{\mu}_s)|^2 ds \\
        &= T \|G(\mu-\hat{\mu})\|_2^2 \overset{(2)}{\leq} T(\|G\|_p^2 + \|\mu-\hat{\mu}\|_q^2)\\
        &= T \left\{\mathbb{E}\int_0^T |G_s|^{p} ds\right\}^{\frac{2}{p}} \left\{\mathbb{E}\int_0^T |\mu_s-\hat{\mu}_s|^q ds\right\}^{\frac{2}{q}} \\
        &\leq T \Big(T\sup_{t\in[0,T]}\mathbb{E}|G_t|^p\Big)^\frac{2}{p} \|\mu-\hat{\mu}\|_q^2 \overset{(3)}{\to} 0,
    \end{align*}
    where ``(1)'' applies Lem.~\ref{lem.tech-lem-p-norm-of-integral}, ``(2)'' applies the Hölder inequality, and ``(3)'' applies Asm.~\ref{asm.Lp-consistency-estimator} and Asm.~\ref{asm.moment-bound}. Hence, we have $R_{1,t}\overset{\mathcal{P}}{\to} 0$ for every $t$. 
    
    Then, we show $R_1$ is stochastic equicontinuous (Def.~\ref{def.stoch-equicontinuous}), the desired statement then follows from Thm.~\ref{thm.pointwise-convergence+stoch-equicon=uniform-convergence}. Use Coro.~\ref{coro.max-ineq-2.2.4}, we just need to show $\mathbb{E} |R_{1,t}-R_{1,s}|^2 \leq C|t-s|^2$ for some positive  constant $C$. In detail, we have:
        \begin{align*}
        \mathbb{E}\left(|R_{1,t}-R_{2,s}|^2 \mid \Tilde{\mathcal{G}}_1^c\right) &= N^{-1} \mathbb{E}\Big\{\big|\sum_{j\in J} \int_s^t G_{j,s}(\mu_{j,\tau}-\hat{\mu}_{j,\tau}) d\tau\big|^2 \mid \Tilde{\mathcal{G}}_1^c\Big\} \\
        &= N^{-1} \sum_{j\in J} \mathbb{E}\Big\{\big|\int_s^t G_{j,\tau}(\mu_{j,\tau}-\hat{\mu}_{j,\tau}) d\tau\big|^2 \mid \Tilde{\mathcal{G}}_1^c\Big\} \\
        &= \mathbb{E}\Big\{\big|\int_s^t G_\tau(\mu_\tau-\hat{\mu}_\tau) d\tau\big|^2 \mid \Tilde{\mathcal{G}}_1^c\Big\}.
    \end{align*}
    
    Hence, for $p,q>1$ such that $\frac{1}{p}+\frac{1}{q}=\frac{1}{2}$,
    \begin{align*}
        \mathbb{E}|R_{1,t}-R_{2,s}|^2 &=  \mathbb{E} \left|\int_s^t G_\tau(\mu_\tau-\hat{\mu}_\tau) d\tau\right|^2 \overset{(1)}{\leq} (t-s)\int_s^t \mathbb{E} |G_\tau (\mu_\tau-\hat{\mu}_\tau)|^2 d\tau \\
        &\overset{(2)}{\leq} (t-s)\int_s^t (\mathbb{E} |G_\tau|^p)^{\frac{2}{p}} (\mathbb{E}|\mu_\tau-\hat{\mu}_\tau)|^q)^{\frac{2}{q}} d\tau \\
        &\leq (t-s)^2 \big(\sup_\tau \mathbb{E} |G_\tau|^p\big)^{\frac{2}{p}} \big(2^q \sup_\tau \mathbb{E}|\mu_\tau|^q + 2^q \sup_\tau \mathbb{E}|\hat{\mu}_\tau|^q\big)^{\frac{2}{q}} \overset{(3)}{=:} C(t-s)^2,
    \end{align*}
    where ``(1)'' applies Lem.~\ref{lem.tech-lem-p-norm-of-integral}, ``(2)'' applies the Hölder inequality, and ``(3)'' applies Asm.~\ref{asm.moment-bound}. The desired statement then follows from Thm.~\ref{thm.pointwise-convergence+stoch-equicon=uniform-convergence}.
\end{proof}

\begin{lemma}
    Assume Asms.~\ref{asm.Lp-consistency-estimator}, \ref{asm.moment-bound}, then $\sup_{t\in [0,T]}|R_{2,t}| \overset{\mathcal{P}}{\to} 0$ as $N\to \infty, \delta\to 0$.
     \label{lem.convergence-R2}
\end{lemma}

\begin{proof}
    Recall
    \begin{equation*}
        R_{2,t} =  \frac{1}{\sqrt{N}} \sum_{j\in J} \int_0^t (\hat{G}_{j,s}-G_{j,s}) d\mathbf{M}_{j,s}.
    \end{equation*}

    We first show the process $R_{2,t}$ is a martingale conditionally on $\Tilde{\mathcal{G}}_T^{c}$. To be specific, according to Prop.~\ref{prop.is-martingale}, we have $t\mapsto \int_0^t G_s d\mathbf{M}_s$ and $t\mapsto \int_0^t \hat{G}_s d\mathbf{M}_s$ are both square integrable, mean-zero $\mathcal{G}_t$-martingale conditionally on $\Tilde{\mathcal{G}}_T^{c}$. Therefore, each of the terms in $R_{2,t}$ is a $\mathcal{G}_{j,t}^{c}$-martingale conditionally on $\Tilde{\mathcal{G}}_T^{c}$. Since $j\in J$ are i.i.d, we have $R_{2,t}$ is a $\Tilde{\mathcal{G}}_t^{c}$-martingale conditionally on $\Tilde{\mathcal{G}}_T^{c}$.
    
    Hence, the squared process $(R_{2,t})^2$ is a $\Tilde{\mathcal{G}}^c_t$-submartingale conditionally on $\Tilde{\mathcal{G}}_T^{c}$. Apply the Doob submartingale inequality (Lem.~\ref{lem.doob-submartingale-ineq}), we then have:
    \begin{align*}
        P\Big(\sup_{t\in[0,T]} |R_{2,t}|\geq \epsilon \Big) &= P\Big(\sup_{t\in[0,T]} |R_{2,t}|^2\geq \epsilon^2 \Big) \\
        &= \mathbb{E} \left\{ P\Big(\sup_{t\in[0,T]} |R_{2,t}|^2\geq \epsilon^2 \mid \Tilde{\mathcal{G}}_1^{c} \Big) \right\}\leq \frac{\mathbb{E}\{\mathrm{Var} (R_{2,T}|\Tilde{\mathcal{G}}_1^{c})\}}{\epsilon^2}
    \end{align*}
    for any $\epsilon>0$. Since the collection of r.v.'s $\left\{\int_0^1 (G_{j,s} - \hat{G}_{j,s}) d\mathbf{M}_{j,s} \mid j\in J\right\}$ are i.i.d conditionally on $\Tilde{\mathcal{G}}_T^c$, we have:
    \begin{align*}
        \mathrm{Var}(R_{2,T}|\Tilde{\mathcal{G}}_1^{c}) &= \frac{1}{N} \sum_{j\in J} \mathrm{Var}\left\{\int_0^T (G_{j,s} - \hat{G}_{j,s}) d\mathbf{M}_{j,s} \mid \Tilde{\mathcal{G}}_T^{c}) \right\} \\
        &= \mathbb{E}\left\{\int_0^T (G_s-\hat{G}_s)^2 d[\mathbf{M}]_s \mid \Tilde{\mathcal{G}}_T^{c}\right\} =  \mathbb{E}\left\{\int_0^T (G_s-\hat{G}_s)^2 \|\Sigma_\beta\|^2 ds \mid \Tilde{\mathcal{G}}_T^{c}\right\}.
    \end{align*}
    
    Thus,
    \begin{equation*}
        \mathbb{E}\{\mathrm{Var}(R_{2,1}|\Tilde{\mathcal{G}}_1^{c})\} \leq \|\Sigma_\beta\|^2 \, \mathbb{E} \int_0^T (G_s-\hat{G}_s)^2 ds = \|\Sigma_\beta\|^2 \|G-\hat{G}\|_2^2.
    \end{equation*}

    We then conclude that for $q>2$,
    \begin{equation*}
        P\Big(\sup_{t\in[0,T]} |R_{2,t}|\geq \epsilon \Big) \leq \frac{\|\Sigma_\beta\|^2 \|G-\hat{G}\|_2^2}{\epsilon^2} \leq \frac{\|\Sigma_\beta\|^2 \|G-\hat{G}\|_q^2}{\epsilon^2}\to 0
    \end{equation*}
    as $N\to\infty, \delta\to 0$, due to the Hölder inequality and Asm.~\ref{asm.Lp-consistency-estimator}. This concludes the proof.
\end{proof}

\begin{lemma}
     Assume Asms.~\ref{asm.Lp-consistency-estimator}, \ref{asm.moment-bound}, then $\sup_{t\in [0,T]}|R_{3,t}| \overset{\mathcal{P}}{\to} 0$ as $N\to \infty, \delta\to 0$.
     \label{lem.convergence-R3}
\end{lemma}

\begin{proof}
    Recall
    \begin{equation*}
        R_{3,t} =  \frac{1}{\sqrt{N}} \sum_{j\in J} \int_0^t (\hat{G}_{j,s}-G_{j,s}) (\mu_{j,s}-\hat{\mu}_{j,s}) ds.
    \end{equation*}
    
    We have:
    \begin{align*}
        \mathbb{E}\Big(\sup_{t\in[0,T]} |R_{3,t}|\Big) &=  \mathbb{E}\, \sup_{t\in[0,T]} \Big|\frac{1}{\sqrt{N}} \sum_{j\in J} \int_0^t (G_{j,s}-\hat{G}_{j,s})(\mu_{j,s}-\hat{\mu}_{j,s})ds\Big| \\
        &\leq \frac{1}{\sqrt{N}} \sum_{j\in J} \mathbb{E} \Big(\sup_{t\in[0,T]}\int_0^t |G_{j,s}-\hat{G}_{j,s}| |\mu_{j,s}-\hat{\mu}_{j,s}|ds \Big)\\
        &= \sqrt{N} \mathbb{E} \int_0^T |(G_{s}-\hat{G}_{s})(\mu_{s}-\hat{\mu}_{s})|ds \\
        &= \sqrt{N} \|(G-\hat{G})(\mu-\hat{\mu})\|_1 \leq \sqrt{N} \|G-\hat{G}\|_2 \|\mu-\hat{\mu}\|_2 \to 0,
    \end{align*}
    according to the Hölder inequality and Asm.~\ref{asm.Lp-consistency-estimator}.
\end{proof}

\subsection{Proof of Prop.~\ref{prop.asym-D1-D2}}
\label{appx.proof-prop.asym-D1-D2}

In this section, we prove Prop.~\ref{prop.asym-D1-D2}. For clarity. we separate the discussion for $D_1$ and $D_2$ into two lemmas (Lem.~\ref{lem.convergence-D1} and Lem~\ref{lem.convergence-D2}), which together amount to Prop.~\ref{prop.asym-D1-D2}.

\begin{lemma}
    Assume Asm.~\ref{asm.moment-bound}, then the process $D_1 - \sqrt{N} \gamma$ converges in distribution in $C[0,T]$ to a Gaussian process with covariance function $(s,t)\mapsto \mathrm{Cov}(Q_s,Q_t)$ as $N\to\infty,\delta\to 0$. 
    \label{lem.convergence-D1}
\end{lemma}

\begin{proof}
    Recall
    \begin{equation*}
        D_{1,t} = \frac{1}{\sqrt{N}} \sum_{j\in J} \int_0^t G_{j,s} (\boldsymbol{\mu}_{j,s}-\mu_{j,s})ds.
    \end{equation*}

    Let $\Tilde{D}_1:= D_1 - \sqrt{N}\gamma, Q_j(t):=\int_0^t G_{j,s}(\boldsymbol{\mu}_{j,s}-\mu_{j,s}) ds - \gamma_t$, then $\Tilde{D}_1 = \frac{1}{\sqrt{N}} \sum_{j\in J} Q_j$, where $\{Q_j \mid j\in J\}$ are i.i.d. copies of the random process $Q_t := \int_0^t G_{s}(\boldsymbol{\mu}_{s}-\mu_{s}) ds - \gamma_t$. We note that the mean of $Q$ is zero since $(\boldsymbol{\mu}_t-\mu_t) dt = d(M_t -\mathbf{M}_t)$ and $\mathbb{E} \int_0^t G_s d\mathbf{M}_s = 0$.

    Denote $\Gamma$ as a Gaussian process with mean zero and covariance function
    \begin{equation}
        (s,t)\mapsto \mathrm{Cov}(Q_s,Q_t),
        \label{def.Gamma-process}
    \end{equation}
    which is well-defined by the computations shown below. In the following, we will show that $\Tilde{D}_1\overset{\mathcal{D}}{\to}\Gamma$ in $C[0,T]$ by applying Lem.~\ref{lem.tight+identification}. Specifically, we first show the convergence of the finite-dimensional distribution as required in Lem.~\ref{lem.tight+identification}. That is, $\forall\, 0\leq t_1 < \cdots <t_k \leq T$,
    \begin{equation}
        \Tilde{\mathbf{D}} := \left\{\Tilde{D}_1(t_1), ..., \Tilde{D}_1(t_k)\right\} \overset{\mathcal{D}}{\to}\left\{\Gamma(t_1),...,\Gamma(t_k)\right\}
        \label{eq.finite-dim-converge}
    \end{equation}
    as random vectors in $\mathbb{R}^k$. We show this by applying the Central Limit Theorem (CLT) to the sequence of random vectors $\Tilde{\mathbf{D}}^{(n)}$, which are normalized sums of the zero-mean i.i.d. random vectors $\mathbf{Q}:=(Q_{t_1},...,Q_{t_k})$. We observe that for $p,q>1$ with $\frac{1}{p}+\frac{1}{q}=\frac{1}{2}$,
    \begin{align*}
        \mathrm{Var}(Q_t) = \mathrm{Var}(Q_t+\gamma_t) &= \mathrm{Var} \int_0^t G_s(\boldsymbol{\mu}_s-\mu_s) ds \\
        &= \mathbb{E} \Big\{\int_0^t G_s(\boldsymbol{\mu}_s-\mu_s) ds\Big\}^2 \leq 2T \|G\|_p^2 \left(\|\boldsymbol{\mu}\|_q^2 + \|\mu\|_q^2\right) \leq \infty
    \end{align*}
    according to the Hölder inequality and Asm.~\ref{asm.moment-bound}. Hence, the CLT concludes that:
    \begin{equation*}
        \Tilde{\mathbf{D}} \overset{\mathcal{D}}{\to}\mathcal{N}(0,\mathrm{Var}(\mathbf{Q})),
    \end{equation*}
    which, by the definition of the process $\Gamma$, is equivalent to \eqref{eq.finite-dim-converge}.

    In the following, we show that both $\Tilde{D}_1$ and $\Gamma$ are tight, which means the second condition Lem.~\ref{lem.tight+identification} holds. We show this by showing that both $\Tilde{D}_1$ and $\Gamma$ are stochastic equicontinuous. Prop.~\ref{prop.stoch-equi<=>tight} then implies that they are tight. Specifically, for $\Gamma$, we have:
     \begin{align*}
        \mathbb{E}\left(\Gamma_t - \Gamma_s\right)^2 &= \mathbb{E}(Q_t-Q_s)^2 = \mathbb{E} \Big\{\int_s^t G_{\tau}(\boldsymbol{\mu}_{\tau}-\mu_{\tau}) d\tau - \mathbb{E}\int_s^t G_\tau dM_\tau \Big\}^2 \\
        &\overset{(1)}{\leq} 2(t-s)\mathbb{E} \int_s^t |G_{\tau}(\boldsymbol{\mu}_{\tau}-\mu_{\tau})|^2 d\tau + 2\Big(\mathbb{E} \int_s^t \mathcal{G}_\tau dM_\tau d\tau\Big)^2 \\
        &\overset{(2)}{\leq} 2(t-s)\mathbb{E} \int_s^t |G_{\tau}(\boldsymbol{\mu}_{\tau}-\mu_{\tau})|^2 d\tau + 2\mathbb{E} \Big(\int_s^t \mathcal{G}_\tau dM_\tau d\tau\Big)^2 \\
        &\overset{(3)}{=} 2(t-s)\mathbb{E} \int_s^t |G_{\tau}(\boldsymbol{\mu}_{\tau}-\mu_{\tau})|^2 + 2\mathbb{E} \int_s^t |\mathcal{G}_\tau|^2 d\tau\\
        &\overset{(4)}{\leq} (t-s)^2 \big(\sup_\tau \mathbb{E} |G_\tau|^p\big)^{\frac{2}{p}}\Big\{ \big(2^q \sup_\tau \mathbb{E}|\boldsymbol{\mu}_\tau|^q + 2^q \sup_\tau \mathbb{E}|{\mu}_\tau|^q\big)^{\frac{2}{q}} +1 \Big\} \overset{(3)}{=:} C(t-s)^2, \numberthis \label{eq.gammat-gammas}
    \end{align*}
    where ``(1)'' applies Lem.~\ref{lem.tech-lem-p-norm-of-integral}, ``(2)'' applies the Jensen inequality, ``(3)'' applies the Ito's Isometry, and ``(4)'' applies Asm.~\ref{asm.moment-bound}. Therefore, applying Coro.~\ref{coro.max-ineq-2.2.4} shows $\Gamma$ is stochastic equicontinuous. For $\Tilde{D}_1$, we note that:
    \begin{align*}
        \mathbb{E}|\Tilde{D}_1(t)-\Tilde{D}_1(s)|^2 &= \frac{1}{N} \mathbb{E} \Big|\sum_{j\in J} Q_j(t) - Q_j(s)\Big|^2 \\
        &= \frac{1}{N} \sum_{j\in J} \mathbb{E} |Q_j(t) - Q_j(s)|^2 =  \mathbb{E}|Q_t - Q_s|^2 \leq C(t-s)^2.
    \end{align*}
    Hence, $\Tilde{D}_1$ is also stochastic equicontinuous.

    To summarize, we have shown that the two conditions of Lem.~\ref{lem.tight+identification} both hold true. Applying the lemma then concludes the proof.
\end{proof}

Before moving on to $D_2$, we note that Lem.~\ref{lem.convergence-D1} implies $\Tilde{D}_1$ is stochastically bounded.
\begin{lemma}
    Assume Asm.~\ref{asm.moment-bound}, then the process $\Tilde{D}_1:=D_1 - \sqrt{N} \gamma$ is stochastically bounded, that is, for any $\epsilon>0$, there exists $K>0$ such that
    \begin{equation*}
        \limsup_{N\to\infty,\delta\to 0} P(\|\Tilde{D}_1^{(n)}\|_{\infty}>K) <\epsilon.
    \end{equation*}
    \label{lem.Tilde-D>k}
\end{lemma}

\begin{proof}
    We have established in Lem.~\ref{lem.convergence-D1},  under the same conditions, that $\Tilde{D}_1 \overset{\mathcal{D}}{\to} \Gamma$ in $C[0,T]$. Since $\|\cdot\|_\infty$ is a continuous function from $C[0,T]$ to $\mathbb{R}^+$, we have $\|\Tilde{D}_1\|_\infty \overset{\mathcal{D}}{\to} \|\Gamma\|_\infty$ due to the continuous mapping theorem. Hence,
    \begin{equation*}
        \limsup_{N\to\infty,\delta\to 0} P(\|\Tilde{D}_1^{(n)}\|_{\infty}>K) = P(\|\Gamma\|_\infty>K) \leq \frac{\mathbb{E}\|\Gamma\|_\infty}{K}.
    \end{equation*}

    Hence it suffices to argue that $\mathbb{E}\|\Gamma\|_\infty$ is bounded. Under \eqref{eq.gammat-gammas}, the boundedness of $\mathbb{E}\|\Gamma\|_\infty$ has been shown in the proof of \cite[Lem. A.9]{christgau2024nonparametric}. 
\end{proof}

\begin{lemma}
    Assume Asms.~\ref{asm.Lp-consistency-estimator}, \ref{asm.moment-bound}, then $D_2\overset{\mathcal{P}}{\to} 0$ in $C[0,T]$ as $N\to\infty, \delta\to 0$.
    \label{lem.convergence-D2}
\end{lemma}

\begin{proof}
    Recall
    \begin{equation*}
        D_2(t) = \frac{1}{\sqrt{N}} \sum_{j\in J} \int_0^t (G_{j,s}-\hat{G}_{j,s})(\boldsymbol{\mu}_s-\mu_s) ds.
    \end{equation*}

    We first show$D_2(t)\overset{\mathcal{P}}{\to} 0$ for every $t$. Specifically, conditionally on $\mathcal{G}_1^{c}$, the terms in $D_2$ are i.i.d. with the same distribution as the process $\zeta_t := \int_0^t (G_{s}-\hat{G}_{s})(\boldsymbol{\mu}_s-\mu_s) ds$. Furthermore, we observe that:
    \begin{align*}
        \mathbb{E} \left\{(\hat{G}_t - G_t)(\boldsymbol{\mu}_t - \mu_t) \mid \Tilde{\mathcal{G}}_1^c \right\} &= \mathbb{E} \left[\mathbb{E} \left\{(\hat{G}_t - G_t)(\boldsymbol{\mu}_t - \mu_t) \mid \Tilde{\mathcal{G}}_1^c \lor \mathcal{F}_t \right\} \mid \Tilde{\mathcal{G}}_1^c \right] \\
        &= \mathbb{E} \left\{(\hat{G}_t - G_t) \mathbb{E} \left(\boldsymbol{\mu}_t - \mu_t \mid \Tilde{\mathcal{G}}_1^c \lor \mathcal{F}_t \right) \mid \Tilde{\mathcal{G}}_1^c \right\} \\
        &= \mathbb{E} \left[(\hat{G}_t - G_t) \left\{\mathbb{E} \left(\boldsymbol{\mu}_t \mid \Tilde{\mathcal{G}}_1^c \lor \mathcal{F}_t \right)- \mu_t \right\} \mid \Tilde{\mathcal{G}}_1^c \right] \\
        &= \mathbb{E} \left[(\hat{G}_t - G_t) \left\{\mathbb{E} \left(\boldsymbol{\mu}_t \mid \mathcal{F}_t \right)- \mu_t \right\} \mid \Tilde{\mathcal{G}}_1^c \right] = 0.
    \end{align*}

    Hence, we have $E(\zeta_t\mid \Tilde{\mathcal{G}}_1^c) = 0$. It follows that:
    \begin{align*}
        \mathrm{Var}(D_{2,t} \mid \Tilde{\mathcal{G}}_1^c) = \mathbb{E}\Big[\big\{ \int_0^t (G_{s}-\hat{G}_{s})(\boldsymbol{\mu}_s-\mu_s) ds\big\}^2 \mid \Tilde{\mathcal{G}}_1^c \Big].
    \end{align*}

    Furthermore, since $\mathbb{E}(D_{2,t}\mid \Tilde{G}_1^c)=0$, we have:
    \small
    \begin{align*}
        \mathrm{Var}(D_{2,t}) &= \mathbb{E} \left\{\mathrm{Var}(D_{2,t} \mid \Tilde{\mathcal{G}}_1^c) \right\} = \mathbb{E} \left\{ \int_0^t (G_{s}-\hat{G}_{s})(\boldsymbol{\mu}_s-\mu_s) ds\right\}^2 \\
        &\leq T \int_0^t \mathbb{E} (G_s-\hat{G}_s)^2 (\boldsymbol{\mu}_s-\mu_s)^2 ds \leq T \left(\|\boldsymbol{\mu}\|_p^2 + \|\mu\|_p^2\right) \|G-\hat{G}\|_q^2 \leq C \|G-\hat{G}\|_q^2 \to 0
    \end{align*}
    \normalsize
    under Asms.~\ref{asm.Lp-consistency-estimator}, \ref{asm.moment-bound}. Therefore, we have $D_2(t)\overset{\mathcal{P}}{\to} 0$ for every $t$.

    Then, we show $D_2$ is stochastic equicontinuous, which implies $\sup_t D_2(t) \to_\mathcal{P} 0$ due to Thm.~\ref{thm.pointwise-convergence+stoch-equicon=uniform-convergence}. Specifically, we observe that:
    \small
    \begin{align*}
        &\mathbb{E} |D_{2,t}-D_{2,s}|^2 = \mathbb{E} \Big\{ \int_s^t (G_s-\hat{G}_s)(\boldsymbol{\mu}_s-\mu_s) ds\Big\}^2 \leq (t-s) \int_s^t \mathbb{E} (G_s-\hat{G}_s)^2 (\boldsymbol{\mu}_s-\mu_s)^2 ds \\
        &\leq (t-s)^2 \big(2^p\sup_\tau \mathbb{E} |G_\tau|^p+2^p\sup_\tau \mathbb{E} |\hat{G}_\tau|^p\big)^{\frac{2}{p}}\big(2^q \sup_\tau \mathbb{E}|\boldsymbol{\mu}_\tau|^q + 2^q \sup_\tau \mathbb{E}|{\mu}_\tau|^q\big)^{\frac{2}{q}} =: C(t-s)^2
    \end{align*}
    \normalsize
    due to Asm.~\ref{asm.moment-bound}. We then conclude from Coro.~\ref{coro.max-ineq-2.2.4} that $\Tilde{D}_2$ is stochastic equicontinuous. This concludes the proof.
\end{proof}

\subsection{Proof of Thm.~\ref{thm.main-4.6}}
\label{appx.proof-thm.main-4.6}

In this section, we prove Thm.~\ref{thm.main-4.6}. To this end, we first show the Gaussian martingale $U_0$ from Prop.~\ref{prop.convergence-u} is tight (Def.~\ref{def.tightness}) in $C[0,T]$.

\begin{lemma}
    Let $U_0$ be a mean-zero continuous Gaussian martingale on $[0,T]$ with variance function $\mathcal{V}$. Then, under \eqref{eq.regularity-lambda-bounded-inL2}, $U_0$ is tight in $C[0,T]$.
    \label{lem.U0-is-tight}
\end{lemma}

\begin{proof}
   
    According to Prop.~\ref{prop.C.2.LCT2024}, $U_0$ has a distributional representation as a time-transformed Wiener process 
    \begin{equation}
        \left\{U_0(t)\right\}_{t\in [0,T]}\overset{\mathcal{D}}{=}\left\{W_{\mathcal{V}(t)}\right\}_{t\in [0,T]},
        \label{eq.U_0-dist-representation}
    \end{equation}
    where $W=(W_t)$ is a standard Wiener process.

    In the following, we will use \eqref{eq.U_0-dist-representation} and Prop.~\ref{prop.stoch-equi<=>tight} to show $U_0$ is tight. Specifically, we have $P(U_0=0)=1$ trivially holds. Further, for any $\epsilon>0$, 
    \begin{equation*}
        \lim_{\delta\to 0^+} P\Big(\sup_{|t-s|<\delta} \left|U_t-U_s\right|>\epsilon\Big) = \lim_{\delta\to 0^+} P\Big\{\sup_{|t-s|<\delta} \left|W_{\mathcal{V}(t)}-W_{\mathcal{V}(s)}\right|>\epsilon\Big\}
    \end{equation*}
    due to \eqref{eq.U_0-dist-representation}. Since the Wiener process is $\alpha$-Hölder continuous for $\alpha\in \left(0,\frac{1}{2}\right)$:
    \begin{equation*}
        |W_t - W_s| \leq C|t-s|^\alpha, \quad a.s. 
    \end{equation*}
    we further have that:
    \begin{align*}
        \lim_{\delta\to 0^+} P\Big(\sup_{|t-s|<\delta} \left|U_t-U_s\right|>\epsilon\Big) &\leq \lim_{\delta\to 0^+} P\Big\{C \sup_{|t-s|<\delta} \left|\mathcal{V}(t)-\mathcal{V}(s)\right|^\alpha>\epsilon\Big\} \\
        &\leq \lim_{\delta\to 0^+} P\Big\{C L^\alpha \sup_{|t-s|<\delta} \left|t-s\right|^\alpha>\epsilon\Big\} \\
        &\leq \lim_{\delta\to 0^+} P\left\{C L^\alpha \delta^\alpha >\epsilon\right\} = 0,
    \end{align*}
    where we have used $\mathcal{V}(t)$ is $L$-Lipschitz continuous under Asm.~\ref{asm.moment-bound} (ref. to \eqref{eq.V(t)-is-Lipschitz-continuous}). Therefore, the second condition of Prop.~\ref{prop.stoch-equi<=>tight} is satisfied and we conclude that $U_0$ is tight.
\end{proof}

\noindent\textbf{Theorem~\ref{thm.main-4.6}.} \emph{Assume \eqref{eq.regularity-lambda-bounded-inL2} and Asms.~\ref{asm.Lp-consistency-estimator}, \ref{asm.moment-bound}, then 
    \begin{enumerate}[label=(\roman*)]
        \item  under $\mathbb{H}_0$, it holds that
        \begin{equation*}
            \sqrt{N} \hat{\gamma} \overset{\mathcal{D}}{\to} U_0
        \end{equation*}
        in $C[0,T]$ as $N\to\infty, \delta\to 0$, where $U_0$ is a mean zero continuous Gaussian martingale on $[0,T]$ with variance function $\mathcal{V}$.
        \item For every $\epsilon>0$, there exists $K>0$ such that
        \begin{equation*}
            \limsup_{N\to\infty,\delta\to 0} P(\sqrt{N} \sup_{t\in [0,T]}|\hat{\gamma}_t-\gamma_t|> K) < \epsilon.
        \end{equation*}
    \end{enumerate}}

\begin{proof}
    For (i), we first note that under $\mathbb{H}_0$, $\mu=\boldsymbol{\mu}$, which implies $D_1$ and $D_2$ are both zero processes. Combine Prop.~\ref{prop.R1-3-zero} and Prop.~\ref{prop.convergence-u}, use the Slutsky theorem, we see that:
    \begin{equation*}
        \sqrt{N} \gamma = \underbrace{U}_{\overset{\mathcal{D}}{\to}U_0} + \underbrace{R_1 + R_2 + R_3}_{\to_\mathcal{P} 0} + \underbrace{D_1 + D_2}_{=0 \,\text{under} \, \mathbb{H}_0} \overset{\mathcal{D}}{\to}U_0.
    \end{equation*}

    For (ii), we can, in addition to Prop.~\ref{prop.R1-3-zero} and Prop.~\ref{prop.convergence-u}, use Prop.~\ref{prop.asym-D1-D2} and Lem.~\ref{lem.Tilde-D>k}. Specifically, the triangle inequality yields that:
    \begin{align*}
        \sqrt{N} \|\gamma-\hat{\gamma}\|_\infty \leq \|U_0\|_\infty + \|D_1-\sqrt{N} \gamma\|_\infty + \underbrace{\|R_1\|_\infty + \|R_2\|_\infty + \|R_3\|_\infty + \|D_2\|_\infty}_{\text{$\overset{\mathcal{P}}{\to} 0$ due to Props.~\ref{prop.R1-3-zero}, \ref{prop.asym-D1-D2}}}.
    \end{align*}

    Therefore, we have due to Prop.~\ref{prop.convergence-u} and Lem.~\ref{lem.convergence-D1} that:
    \begin{align*}
        \limsup_{N\to\infty,\delta\to 0} P\left(\sqrt{N} \|\gamma-\hat{\gamma}\|_\infty > K\right) \leq P\left(\|U_0\|_\infty >K/2\right) +  \limsup_{N\to\infty,\delta\to 0} P(\|\Tilde{D}_1^{(n)}\|_{\infty}>K/2).
    \end{align*}

    The first and second terms on the right-hand side can be controlled into arbitrarily small due to Lem.~\ref{lem.U0-is-tight} and Lem.~\ref{lem.Tilde-D>k}, respectively. This concludes the proof.
\end{proof}

\subsection{Proof of Prop.~\ref{prop.consistency-vt}}
\label{appx.proof-prop.consistency-vt}

\noindent\textbf{Proposition~\ref{prop.consistency-vt}.} \emph{Assume Asms.~\ref{asm.Lp-consistency-estimator}, \ref{asm.moment-bound}, \ref{asm.Sigma-convergence-in-probability}, then we have $\sup_{t\in [0,T]} |\hat{\mathcal{V}}_t - \mathcal{V}_t| \overset{\mathcal{P}}{\to} 0$ as $N\to\infty, \delta\to 0$.}

\begin{proof}
    Consider the decomposition of $\hat{\mathcal{V}}$ given by $\hat{V}(t) = \|\hat{\Sigma}_\beta\|^2 \left(A_t + B_t + C_t\right)$, where
    \small
    \begin{align*}
        &A_t := \frac{1}{N} \sum_j \int_0^t G_{j,s}^2 ds, \quad B_t := \frac{1}{N} \sum_j \int_0^t (G_{j,s}-\hat{G}_{j,s})^2 ds, \quad C_t := \frac{1}{N} \sum_j \int_0^t G_{j,s} (G_{j,s}-\hat{G}_{j,s}) ds.
    \end{align*}
    \normalsize

    We have that $A_t$ is the empirical mean of $N$ i.i.d. samples of the process $\int_0^t  G_s^2 ds$. Further, $\int_0^t  G_s^2 ds$ is integrable under Asm.~\ref{asm.moment-bound}. Therefore, the Law of Large Numbers yields that $A_t \overset{\mathcal{P}}{\to} \int_0^t  G_s^2 ds$ for every $t\in [0,T]$. By the continuous mapping theorem, we further obtain that $\|\hat{\Sigma}_\beta\|^2 A_t \overset{\mathcal{P}}{\to} \mathcal{V}(t)$ for every $t\in [0,T]$. Note that $\|\hat{\Sigma}_\beta\|^2 A_t$ and $\mathcal{V}(t)$ are both nondecreasing, and the limiting process $\mathcal{V}(t)$ is stochastic equicontinuous:
    \begin{equation}
        \mathcal{V}(t) - \mathcal{V}(s) = \|\Sigma_\beta\|^2 \int_s^t \mathbb{E} G_s^2 ds \leq C(t-s).
        \label{eq.V(t)-is-Lipschitz-continuous}
    \end{equation}
    due to Asm.~\ref{asm.moment-bound} and  Coro.~\ref{coro.max-ineq-2.2.4}. Therefore, apply Lem.~\ref{lem.christlemb13}, we conclude that:
    \begin{equation*}
        \sup_{t\in [0,T]} \left| \|\hat{\Sigma}_\beta\|^2 A_t - \mathcal{V}(t)\right| \overset{\mathcal{P}}{\to} 0.
    \end{equation*}

    For $B_t$, we have under the Asm.~\ref{asm.Lp-consistency-estimator} that:
    \begin{align*}
        \mathbb{E}\big(\sup_{0\leq t \leq T} B_t\big) = \mathbb{E} B_T &= \mathbb{E} \Big\{\frac{1}{N} \sum_{j\in J} \int_0^T \left(G_{j,s}-\hat{G}_{j,s}\right)^2 ds\Big\} \\
        &= \mathbb{E} \left[\mathbb{E} \Big\{\frac{1}{N} \sum_{j\in J} \int_0^T \big(G_{j,s}-\hat{G}_{j,s}\big)^2 ds \mid \Tilde{\mathcal{G}}_T^c \Big\}\right]\\
        &= \mathbb{E} \int_0^t (G_{s}-\hat{G}_{s})^2 ds = \|G-\hat{G}\|_2\to 0.
    \end{align*}

    Therefore, we have $\sup_{0\leq t \leq T} B_t \to 0$ in probability. Applying the continuous mapping theorem, it follows that $\|\hat{\Sigma}_\beta\|^2 \sup_{0\leq t \leq T} B_t \overset{\mathcal{P}}{\to} \|\Sigma_\beta\|^2 \cdot 0 = 0$.

    Finally, we see that:
    \begin{align*}
        \mathbb{E} \Big|\sup_{t\in [0,T]} C_t\Big| \leq \mathbb{E} \sup_{t\in [0,T]} |C_t| &\leq \mathbb{E}\Big\{\frac{1}{N}\sum_{j\in J} \sup_{t\in [0,T]} \int_0^t |G_{j,s} (G_{j,s} - \hat{G}_{j,s})| ds\Big\} \\
        &= \mathbb{E}\Big\{\frac{1}{N}\sum_{j\in J} \int_0^T |G_{j,s}(G_{j,s} - \hat{G}_{j,s})| ds\Big\} \\
        &= \mathbb{E} \int_0^T |G_{s}(G_{s} - \hat{G}_{s})| ds \\
        &\leq \Big(\mathbb{E}\int_0^T |G_{s} - \hat{G}_{s}|^2 ds \Big)^\frac{1}{2} \Big(\mathbb{E}\int_0^T |G_{s}|^2 ds \Big)^\frac{1}{2} \to 0,
    \end{align*}
    under Asms.~\ref{asm.Lp-consistency-estimator}, \ref{asm.moment-bound}. Therefore, we have $\sup_t C_t \to 0$ in probability. Applying again the continuous mapping theorem, we have $\|\hat{\Sigma}_\beta\|^2 \sup_t C_t \overset{\mathcal{P}}{\to} 0$.

    Combining the above, we show by Slutsky's theorem that:
    \begin{equation*}
         \sup_{t\in [0,T]} |\hat{\mathcal{V}}(t) - \mathcal{V}(t)| \overset{\mathcal{P}}{\to} 0.
     \end{equation*}
\end{proof}

\section{Proof of Section~\ref{sec.lct}}
\label{appx.proof-thm.asym}

The asymptotics of any testing statistic under the null hypothesis $\mathbb{H}_0$ can be described by the following result, which is essentially an application of the continuous mapping theorem.

\begin{proposition}
    Let $\mathcal{J}:C[0,T]\times C[0,T]\mapsto \mathbb{R}$ be a functional that is continuous on the closed subset $C[0,T]\times \{\mathcal{V}(t): t\in [0,T]\}$ with respect to the uniform topology, that is, the topology generated by the norm $\|(f_1,f_2)\|:=\max\{\|f_1\|_\infty, \|f_2\|_\infty\}$ for $f_1, f_2 \in C[0,T]$. Define the testing statistic
    \begin{equation*}
        \hat{T} := \mathcal{J}(\sqrt{N}\hat{\gamma}, \hat{\mathcal{V}}).
    \end{equation*}

    Assume \eqref{eq.regularity-lambda-bounded-inL2} and Asms.~\ref{asm.Lp-consistency-estimator}, \ref{asm.moment-bound}, \ref{asm.Sigma-convergence-in-probability}, then we have:
    $$\hat{T} \overset{\mathcal{D}}{\to} \mathcal{J}(U_0,\mathcal{V})$$
    as $N\to \infty, \delta\to 0$, where $U_0$ is a mean zero continuous Gaussian martingale on $[0,T]$ with variance function $\mathcal{V}$.
    \label{prop.continuous-mapping}
\end{proposition}

\begin{proof}
    The result is a direct consequence of the continuous mapping theorem.
\end{proof}

\noindent\textbf{Theorem~\ref{thm.asym-level}.} \emph{Under \eqref{eq.regularity-lambda-bounded-inL2} and Asms.~\ref{asm.Lp-consistency-estimator}, \ref{asm.moment-bound}, \ref{asm.Sigma-convergence-in-probability}, and \ref{asm.bound-away-from-zero}, we have:
    \begin{equation*}
        \hat{T}_N \overset{\mathcal{D}}{\to} S
    \end{equation*}
    as $N\to\infty, \delta\to 0$. As a consequence, for any $\alpha\in (0,1)$, we have:
    \begin{equation*}
        \limsup_{N\to\infty, \delta\to 0} P(\hat{T}_N>z_{1-\alpha}) \leq \alpha.
    \end{equation*}
    In other words, our test has an asymptotic level $\alpha$.}

\begin{proof}
    We will apply Prop.~\ref{prop.continuous-mapping} with the function $\mathcal{J}$ given by
    \begin{equation*}
        \mathcal{J}(f_1,f_2) := 1(f_2\neq 0) \frac{\|f_1\|_\infty}{\sqrt{|f_2(T)|}}, \quad f_1,f_2\in C[0,T].
    \end{equation*}

    Under Asm.~\ref{asm.bound-away-from-zero}, it suffices to check the continuity of $\mathcal{J}$ on the set $\mathcal{A}$ given by
    \begin{equation*}
        \mathcal{A} := C[0,T] \times \left\{f\in C[0,T]\mid |f(T)|\geq \nu_0\right\},
    \end{equation*}
    where $\nu_0$ is the positive constant assumed in Asm.~\ref{asm.bound-away-from-zero}. To see that $\mathcal{J}$ is continuous on $\mathcal{A}$ with the uniform topology, note that it is the composition of two continuous maps, namely $(f_1,f_2)\mapsto (\|f_1\|_\infty, |f_2(T)|)$ that maps $\mathcal{A}$ to $\mathbb{R}^+\times [\nu_0,\infty)$; and $(x_1,x_2)\mapsto \frac{x_1}{\sqrt{x_2}}$ that maps $\mathbb{R}^+\times [\nu_0,\infty)$ to $\mathbb{R}$. Therefore, Prop.~\ref{prop.continuous-mapping} implies that:
    \begin{equation*}
        \hat{T} = \mathcal{J}(\sqrt{N}\hat{\gamma},\hat{\mathcal{V}}) \overset{\mathcal{D}}{\to} \mathcal{J}(U_0,\mathcal{V}) = \frac{\|U_0\|_\infty}{\mathcal{V}(T)} =: S.
    \end{equation*}

    Apply Prop.~\ref{prop.C.2.LCT2024}, we have:
    \begin{equation}
        S:=\frac{\|U_0\|_\infty}{\mathcal{V}(T)} \overset{\mathcal{D}}{=} \frac{\sup_{t\in [0,T]} \left|W_{\mathcal{V}(t)}\right|}{\mathcal{V}(T)}  \overset{\mathcal{D}}{=} \frac{\sup_{s\in [0,\mathcal{V}(T)]} \left|W_s\right|}{\mathcal{V}(T)} \overset{\mathcal{D}}{=} \sup_{t\in [0,T]} \left|W_t\right|,
        \label{eq.definition-S}
    \end{equation}
    where we have used that $\mathcal{V}$ is continuous under Asm.~\ref{asm.moment-bound} (ref. to \eqref{eq.V(t)-is-Lipschitz-continuous}), and that the Brownian motion is scale invariant. This has concluded the proof.
\end{proof}

\noindent\textbf{Theorem~\ref{thm.asym-power}.} \emph{Assume \eqref{eq.regularity-lambda-bounded-inL2} and Asms.~\ref{asm.Lp-consistency-estimator}, \ref{asm.moment-bound}, \ref{asm.Sigma-convergence-in-probability}. Then, for any $0<\alpha<\beta<1$, there exists $c>0$ such that:
    \begin{equation*}
        \liminf_{N\to\infty,\delta\to 0} \inf_{\gamma \in\mathcal{A}_{c,N}} P(\hat{T}_N>z_{1-\alpha})\geq \beta,
    \end{equation*}
    where $\mathcal{A}_{c,N} := \left\{ \sup_{t\in [0,T]} |\gamma_t | \geq c N^{-\frac{1}{2}} \right\}$. In other words, the type II error rate of the test can be controlled to arbitrarily small.}

\begin{proof}
    Let $0<\alpha<\beta<1$ be given. The second part of Thm.~\ref{thm.main-4.6} allows us to choose $K>0$ sufficiently large such that:
    \begin{equation*}
        \limsup_{N\to\infty, \delta\to 0} \sup_{\theta\in\mathcal{A}_{c,N}} P(\sqrt{N} \|\hat{\gamma}-\gamma\|_\infty > K) < 1-\beta.
    \end{equation*}

    Recall that $\mathcal{V}(t)$ is continuous under Asm.~\ref{asm.moment-bound} (ref. to \eqref{eq.V(t)-is-Lipschitz-continuous}), thus we have $\mathcal{V}(T)\leq C$ is bounded. We then choose $c>K+z_{1-\alpha} \sqrt{1+C^2}$ such that:
    \begin{equation*}
        \sqrt{N} \|\gamma\|_\infty - z_{1-\alpha} \sqrt{1+\mathcal{V}(T)} \geq c - z_{1-\alpha} \sqrt{1+C^2} > K.
    \end{equation*}

    The reverse triangle inequality now yields that:
    \begin{align*}
        \left\{\hat{T}\leq z_{1-\alpha}\right\} &= \left\{\|\hat{\gamma}\|_\infty \leq \sqrt{\hat{\mathcal{V}}(T)} \frac{z_{1-\alpha}}{\sqrt{N}}\right\} \\
        &\subseteq \left\{\|\gamma\|_\infty - \|\hat{\gamma}-\gamma\|_\infty \leq \sqrt{\hat{\mathcal{V}}(T)}  \frac{z_{1-\alpha}}{\sqrt{N}}\right\}.
    \end{align*}

    Therefore, we obtain that:
    \begin{align*}
         &\sqrt{N} \|\gamma\|_\infty \geq K + z_{1-\alpha} \sqrt{1+\mathcal{V}(T)} \\
         &\sqrt{N} \left(\|\gamma\|_\infty - \|\hat{\gamma}-\gamma\|_\infty\right) \leq z_{1-\alpha} \sqrt{\hat{\mathcal{V}}(T)}.
    \end{align*}

    Combining these two inequalities, we obtain that:
    \begin{equation*}
        \sqrt{N} \left(\|\gamma\|_\infty - \|\hat{\gamma}-\gamma\|_\infty\right) \geq K + z_{1-\alpha}\Big(\sqrt{1+\mathcal{V}(T)} - \sqrt{\hat{\mathcal{V}}(T)}\Big).
    \end{equation*}

    Then, denote the two events:
    \begin{align*}
        &\mathcal{E}_1 := \left\{\sqrt{N} \left(\|\gamma\|_\infty - \|\hat{\gamma}-\gamma\|_\infty\right) > K\right\} \\
        &\mathcal{E}_2 := \left\{\hat{\mathcal{V}}(T) > 1+\mathcal{V}(T)\right\} \subseteq \left\{\left|\hat{\mathcal{V}}(T) - \mathcal{V}(T) \right|>1\right\},
    \end{align*}
    we then have:
    \begin{equation*}
        \left\{\hat{T}\leq z_{1-\alpha}\right\} \subseteq \left\{\|\gamma\|_\infty - \|\hat{\gamma}-\gamma\|_\infty \leq \sqrt{\hat{\mathcal{V}}(T)} \frac{z_{1-\alpha}}{\sqrt{N}}\right\} \subseteq \mathcal{E}_1 \cup \mathcal{E}_2.
    \end{equation*}

    From Prop.~\ref{prop.consistency-vt}, we know that $ \lim_{N\to\infty, \delta\to 0} P(\mathcal{E}_2) = 0$. Hence, from the choice of $K$, we conclude that:
    \begin{equation*}
        \limsup_{N\to\infty, \delta\to 0} \sup_{\theta\in\mathcal{A}_{c,N}} P(\hat{T}\leq z_{1-\alpha}) \leq \limsup_{N\to\infty, \delta\to 0} P\left(\mathcal{E}_2 \right) < 1-\beta.
    \end{equation*}

    The desired statement now follows.
\end{proof}

\noindent\textbf{Theorem~\ref{thm.asym-level-Kfold}.} \emph{Suppose that Asms.~\ref{asm.Lp-consistency-estimator}, \ref{asm.Sigma-convergence-in-probability} are satisfied for every sample split $J_k \cup J_k^c$. Under \eqref{eq.regularity-lambda-bounded-inL2} and Asms.~\ref{asm.moment-bound}, \ref{asm.bound-away-from-zero}, we have:
    \begin{equation*}
        \hat{T}_K \overset{\mathcal{D}}{\to} S
    \end{equation*}
    as $N\to\infty, \delta\to 0$. As a consequence, for any $\alpha\in (0,1)$, we have:
    \begin{equation*}
        \limsup_{N\to\infty, \delta\to 0} P(\hat{T}_K>z_{1-\alpha}) \leq \alpha.
    \end{equation*}}

\begin{proof}
    We consider the decomposition \eqref{eq.decomposition} for each sample split $J_k \cup J_k^c$ and denote the corresponding processes by $U^k, R_1^k, R_2^k, R_3^k, D_1^k$ and $D_2^k$. For each $k=1,...,K$, we can then apply the results in Sec.~\ref{sec.asym-ana-ii} for a single data split to conclude that:
    \begin{itemize}
        \item By Prop.~\ref{prop.convergence-u}, we have $U^k\overset{\mathcal{P}}{\to} U_0$ in $C[0,T]$, where $U_0$ is a mean-zero Gaussian martingale with variance function $\mathcal{V}$.
        \item By Prop.~\ref{prop.R1-3-zero}, we have $\sup_t |R_i^k(t)| \to 0$ for $i=1,2,3$.
        \item Under $\mathbb{H}_0$, the processes $D_1^k, D_2^k$ are zero processes almost surely.
    \end{itemize}

    Recall we assume the folds have uniform asymptotic density, meaning $\frac{\sqrt{N_0}}{\sqrt{K N_k}}\to 1$ as $N_0\to \infty$. Thus, we further conclude that for each fixed $k$ and $i=1,2,3$,
    \begin{equation*}
        \frac{\sqrt{N_0}}{\sqrt{K N_k}} U^k \overset{\mathcal{D}}{\to} U_0, \quad \frac{\sqrt{N_0}}{\sqrt{K N_k}} R^k_i \overset{\mathcal{P}}{\to} 0.
    \end{equation*}

    Besides, we have $U^1 \ind ... \ind U^K$ are jointly independent. To see this, note that $U^k$ is constructed from $(G_j,M_j)_{j\in J_k}$ only, and by the i.i.d. assumption of the data, the collections $\{(G_j,M_j)_{j\in J_1}, ..., (G_j,M_j)_{j\in J_K}\}$ are jointly independent. 

    As a result, we have the summation of $U^1,..., U^K$ converges to the summation of $K$ independent copies of $U_0$. Using the convolution property of the Gaussian distribution, we then have:
    \begin{equation*}
        \bar{U}_K := \frac{1}{\sqrt{K}}\sum_{k=1}^K \frac{\sqrt{N_0}}{\sqrt{K N_k}} U^k \overset{\mathcal{D}}{\to} U_0.
    \end{equation*}

    Then, according to the Slutsky theorem, we can conclude that:
    \begin{equation*}
        \sqrt{N_0}\hat{\gamma}_K = \bar{U}_K + \sum_{k=1}^K \frac{\sqrt{N_0}}{K\sqrt{N_k}} (R_1^k+R_2^k+R_3^k+D_1^k+D_2^k)\overset{\mathcal{D}}{\to} U_0.
    \end{equation*}

    According to Lem.~\ref{lem.U0-is-tight}, the limit $U_0$ is tight in $C[0,T]$. Then, according to \cite[Prop. B.9]{christgau2024nonparametric}, we have $\sqrt{N_0}\|\hat{\gamma}_K\|_\infty {\to} \|U_0\|_\infty$.

    Furthermore, consider the cross-fitted variance estimator $\hat{\mathcal{V}}_K(T)$ at the time $T$. According to Prop.~\ref{prop.consistency-vt}, we have $\hat{\mathcal{V}}_K(T)$ is an average of $K$ random variables converging in probability to $\mathcal{V}(T)$. Hence, $\hat{\mathcal{V}}_K(T)$ also converge in probability to $\mathcal{V}(T)$. We can then apply the Slutsky theorem to conclude that:
    \begin{equation*}
        \hat{T}_K = \frac{\sqrt{N_0}\|\hat{\gamma}_K\|_\infty}{\hat{\mathcal{V}}_K(T)} \overset{\mathcal{D}}{\to} \frac{\|U_0\|_\infty}{\mathcal{V}(T)} \overset{\mathcal{D}}{=} S,
    \end{equation*}
    as $N\to\infty, \delta\to 0$, where the last equality has been established in \eqref{eq.definition-S}. This has concluded the proof.
\end{proof}
\section{Proof of Section~\ref{sec.estimation}}
\label{appx.asm-part1}

We first introduce some notational convenience. For a matrix $A\in\mathbb{R}^{m\times n}$, denote $\sigma_{\min}(A)$ and $\sigma_{\max}(A)$ as the smallest and largest singular values of $A$, respectively. Denote $\|A\|_{\mathrm{sp}}=\sigma_{\max}(A), \|A\|_p:=\|\mathrm{vec}(A)\|_p$ as the spectral norm and entrywise $\ell_p$-norm, respectively. We note that these norms are equivalent, for $1\leq p\leq q < \infty$: 
\begin{align}
    \|A\|_{\mathrm{sp}} \leq \|A\|_1 \leq mn \|A\|_\mathrm{sp}, \quad \|A\|_q \leq \|A\|_p \leq (mn)^{2(\frac{1}{p}-\frac{1}{q})} \|A\|_q.
    \label{eq.equivalence-matrix-norms}
\end{align}

We then show the complete form of \eqref{eq.filter-eq-general}. Recall $h_t=h(t,X_t)=h\{t,X_{V\backslash C}(t),X_C(t)\}$. We denote $h_1(t)=\frac{\partial h_t}{\partial t}, h_2(t)=\frac{\partial h_t}{\partial X_{V\backslash C}(t)}, h_3(t)=\frac{\partial h_t}{\partial X_{C}(t)}$ as the Jacobians, and denote $h_{22}(t)=\frac{\partial^2 h_t}{\partial X_{V\backslash C}(t)^2}, h_{33}(t)=\frac{\partial^2 h_t}{\partial X_{C}(t)^2}$ as the Hessians. Then in \eqref{eq.filter-eq-general}, we denote:
\begin{align*}
    &(\mathcal{L}h)_t := h_1(t) + h_2(t)\lambda_{V\backslash C}(t) + h_3(t) \lambda_C(t) \\
    &\manyquad[3] + \frac{1}{2}\left\{ h_{22}(t)(\Sigma_{V\backslash C}\Sigma_{V\backslash C}^\top \Sigma_{V\backslash C,C} +\Sigma_{V\backslash C,C}^\top) + h_{33}(t) (\Sigma_{C,V\backslash C} \Sigma_{C,V\backslash C}^\top + \Sigma_C\Sigma_C^\top)\right\} \\
    &(\mathcal{N}h)_t := h_1(t) \Sigma_{V\backslash C,C} + h_2(t) \Sigma_C,\numberthis\label{eq.filtering-eq-complete-form}
\end{align*}
with $\pi_t(\mathcal{L}h):=\mathbb{E}\{(\mathcal{L}h)_t|\mathcal{F}_t\}$ and $\pi_t(\mathcal{N}h):=\mathbb{E}\{(\mathcal{N}h)_t|\mathcal{F}_t\}$.

\subsection{Proof of Props.~\ref{prop.consistency-ou}-\ref{prop.filtering-eq-uniqueness-solution}}
\label{appx.proof-consistency-ou}

In this section, we introduce the proof of Props.~\ref{prop.consistency-ou}-\ref{prop.filtering-eq-uniqueness-solution}. We first show the consistency of $\bar{F}$ that is defined in line 2 of Alg.~\ref{alg:estimation-ou}.

\begin{proposition}
    Let $p\geq 1$, then there exists a positive constant $K=K(p,\Phi,\sigma)$ such that:  
    \begin{equation*}
        \mathbb{E}\|\bar{F}-F\|^p_\mathrm{sp} \leq K \left(N_c^{-\frac{1}{2}p}+\delta_c^{4p}\right)
    \end{equation*}
    for sufficiently large $N_c$ and sufficiently small $\delta_c$.
    \label{prop.consistency-hatF}
\end{proposition}

\begin{proof}
    We first show the estimator $\wh{F}$ defined in line 1 of Alg.~\ref{alg:estimation-ou} satisfies
    \begin{equation}
        \mathbb{E}\|\wh{F}-F\|^p_\mathrm{sp} \leq K_1 N_c^{-\frac{1}{2}p}
        \label{eq.hatF-F<=K-Nc}
    \end{equation}
    for some positive constant $K_1=K_1(p,\Phi,\sigma)$. Specifically, decompose
    \begin{align*}
        \wh{F} - F = C_N(-1)C_N(0)^{-1} - F = \left[C_N(-1)-\Gamma(-1) - F\left\{C_N(0)-\Gamma(0)\right\}\right] C_N(0)^{-1},
    \end{align*}
    it then follows that:
    \begin{equation*}
        \|\wh{F} - F\|^p_\mathrm{sp} \leq 2^{p}\left\{\|C_N(-1)-\Gamma(-1)\|^p_\mathrm{sp} + \|F\|^p_\mathrm{sp} \|C_N(0)-\Gamma(0)\|^p_\mathrm{sp}\right\} \|C_N(0)^{-1}\|^p_\mathrm{sp}.
    \end{equation*}
    Taking the expectation of both sides and using the Cauchy-Schwarz inequality, we have:
    \small
    \begin{equation}
        \mathbb{E}\|\wh{F} - F\|^p_\mathrm{sp} \leq 2^p\left\{\mathbb{E}\|C_N(-1)-\Gamma(-1)\|^{2p}_\mathrm{sp} + \|F\|^{2p}_\mathrm{sp} \mathbb{E}\|C_N(0)-\Gamma(0)\|^{2p}_\mathrm{sp}\right\}^{\frac{1}{2}} \left\{\mathbb{E}\|C_N(0)^{-1}\|^{2p}_\mathrm{sp}\right\}^{\frac{1}{2}}.
        \label{eq.E|hat{F}-F|}
    \end{equation}
    \normalsize

    According to Lem.~\ref{lem.convergence-CN-1-CN-0}, there exists $K_2=K_2(p,\Phi,\sigma)>0$ such that:
    \begin{equation}
        \max\left\{\mathbb{E}\|\mathrm{vec} \{C_N(-1)\} - \mathrm{vec} \{\Gamma(-1)\}\|^{2p}_{2p}, \, \mathbb{E}\|\mathrm{vec}\{ C_N(0)\}-\mathrm{vec} \{\Gamma(0)\}\|^{2p}_{2p}\right\} \leq K_2 N_c^{-p}.
        \label{eq.convergence-CN1-CN0}
    \end{equation}

    According to Lem.~\ref{lem.boundedness-CN0-inverse}, there exists $K_3=K_3(p,\Phi,\sigma)>0$ such that:
    \begin{equation}
        \mathbb{E}\|C_N(0)^{-1}\|^{2p}_\mathrm{sp}\leq K_3,
        \label{eq.CN-inverse-bounded}
    \end{equation}
    
    Furthermore, for $\delta_c<1$, we have:
    \begin{equation}
        \|F\|^{2p}_\mathrm{sp} = \|\exp(\Phi\delta_c)\|^{2p}_\mathrm{sp} < \exp({\|\Phi\|^{2p}_\mathrm{sp} \delta_c^{2p}})< \exp({\|\Phi\|^{2p}_\mathrm{sp}}).
        \label{eq.bound-norm-F}
    \end{equation}

    Combine \eqref{eq.E|hat{F}-F|}-\eqref{eq.bound-norm-F}, we have \eqref{eq.hatF-F<=K-Nc} holds with $K_1=2^p\{1+\exp({\|\Phi\|^{2p}_\mathrm{sp}})\}^{\frac{1}{2}} (K_2 K_3)^{\frac{1}{2}}$.
   
    We then show the consistency of $\bar{F}$. Specifically, denote the event $\mathcal{E} = \{1\leq\sigma_{\min}(I+\wh{F})\leq \sigma_{\max}(I+\wh{F})\leq3\}$. According to Lem.~\ref{lem.order-P(E)}, we have $P(\mathcal{E}) = 1 - O\left(N_c^{-2p}+\delta_c^{4p}\right)$, which means $\mathcal{E}$ holds almost surely. Then, we decompose
    \begin{equation*}
        \mathbb{E} \|F-\bar{F} \|^p_\mathrm{sp} = \mathbb{E} (\|F-\bar{F} \|^p_\mathrm{sp} 1_\mathcal{E}) + \mathbb{E} (\|F-\bar{F} \|^p_\mathrm{sp} 1_{^\neg \mathcal{E}}).
    \end{equation*}

    For the first term on the RHS, we have:
    \begin{equation*}
        \mathbb{E} (\|F-\bar{F} \|^p_\mathrm{sp} 1_\mathcal{E}) = \mathbb{E} (\|F-\wh{F} \|^p_\mathrm{sp} 1_\mathcal{E}) \leq \mathbb{E} (\|F-\wh{F} \|^p_\mathrm{sp})  \leq K_1 N_c^{-\frac{1}{2}p}.
    \end{equation*}

    For the second term on the RHS, we have, for $\delta_c< 1$:
    \begin{align*}
        \|F-\bar{F}\|^p_\mathrm{sp} &\leq 2^{p} (\|F\|^p_\mathrm{sp}  + \|\bar{F}+I-I\|^p_\mathrm{sp})\\
        &\overset{(1)}{\leq} 2^{p} \{\exp({\|\Phi\|^{2p}_\mathrm{sp}}) + 2^{p}\|\bar{F}+I\|^p_\mathrm{sp} + 2^{p} \|I\|^p_\mathrm{sp}\}\\
        &\overset{(2)}{\leq} 2^{p}\{(\exp({\|\Phi\|^{2p}_\mathrm{sp}}) + 2^{p} 3^p + 2^{p}\} := K_4
    \end{align*}
    where ``(1)'' uses \eqref{eq.CN-inverse-bounded} and ``(2)'' uses $\|\bar{F}+I\|_{\mathrm{sp}}=\sigma_{\max}(\bar{F}+I)\leq 3$ (line 2 of Alg.~\ref{alg:estimation-ou}). Hence, we have $\mathbb{E} (\|F-\bar{F} \|^p_\mathrm{sp} 1_{^\neg \mathcal{E}}) \leq K_4 P(^\neg \mathcal{E}) = O\left(N_c^{-2p}+\delta_c^{4p}\right)$.

    Therefore, we conclude:
    \begin{equation*}
        \mathbb{E} \|F-\bar{F} \|^p_\mathrm{sp} = \mathbb{E} (\|F-\bar{F} \|^p_\mathrm{sp} 1_\mathcal{E}) + \mathbb{E} (\|F-\bar{F} \|^p_\mathrm{sp} 1_{^\neg \mathcal{E}}) = O\left(N_c^{-\frac{1}{2}p} + \delta_c^{4p}\right).
    \end{equation*}
\end{proof}

\begin{lemma}
    Let $p\geq 1$, then there exists a positive constant $K=K(p,\Phi,\sigma)$ such that:  
    \begin{equation*}
        \max\left\{\mathbb{E}\|\mathrm{vec} \{C_N(-1)\} - \mathrm{vec} \{\Gamma(-1)\}\|^{2p}_{2p}, \, \mathbb{E}\|\mathrm{vec}\{ C_N(0)\}-\mathrm{vec} \{\Gamma(0)\}\|^{2p}_{2p}\right\} \leq K N_c^{-p}.
    \end{equation*}
    \label{lem.convergence-CN-1-CN-0}
\end{lemma}

\begin{proof}
    Denote $Y_{\delta_c}(j)=\mathrm{vec} \{X_{\delta_c}(j)X_0(j)^\top\} - \mathbb{E} \mathrm{vec}\{ X_{\delta_c}(j)X_0(j)^\top\}$. Then $\left\{Y_{\delta_c}(j) \mid j\in J_N^c\right\}$ are i.i.d copies of the zero-mean, $L_{2p}$-integrable r.v. $Y_{\delta_c}:=\mathrm{vec} (X_{\delta_c} X_0^\top) - \mathbb{E} \mathrm{vec} (X_{\delta_c} X_0^\top)$. 

    Apply Thm.~\ref{thm.tech-lem-Marc-Zyg-inequality}, we have there exists $K_1(p)>0$,
    \begin{equation*}
        \mathbb{E}\|\mathrm{vec} \{C_N(-1)\} - \mathrm{vec} \{\Gamma(-1)\}\|^{2p}_{2p} = \mathbb{E}\left\|\frac{1}{N_c} \sum_{j\in J_N^c} Y_{\delta_c}(j)\right\|^{2p}_{2p} \leq K_1(p) \mathbb{E}\|Y_{\delta_c}\|_{2p}^{2p} N^{-p}_c.
    \end{equation*}

    Furthermore, we note that:
    \begin{align*}
        \left|\mathbb{E}\|Y_{\delta_c}\|_{2p}^{2p}-\mathbb{E}\|Y_{0}\|_{2p}^{2p}\right| &\leq \mathbb{E}\|Y_{\delta_c}-Y_0\|_{2p}^{2p} \\
        &= \mathbb{E}\|(X_{\delta_c}X_0^\top - \mathbb{E}X_{\delta_c}X_0^\top) - (X_0 X_0^\top - \mathbb{E}X_0 X_0^\top)\|_{2p}^{2p} \\
        &\leq 2^{2p}\left( \mathbb{E}\|X_{\delta_c}X_0^\top-X_0 X_0^\top\|_{2p}^{2p} + \|\mathbb{E}X_{\delta_c}X_0^\top-\mathbb{E}X_0 X_0^\top\|_{2p}^{2p}\right) \\
        &\leq 2^{2p+1} \mathbb{E}\|X_{\delta_c}X_0^\top-X_0 X_0^\top\|_{2p}^{2p}\\
        &\leq 2^{2p+1} \left(\mathbb{E}\|X_{\delta_c}-X_0\|_{2p}^{4p}\right)^{\frac{1}{2}} \left(\mathbb{E}\|X_0\|_{2p}^{4p}\right)^{\frac{1}{2}} \\
        &\overset{(1)}{\leq} 2^{2p+1} \delta_c^{2p} \left(\mathbb{E}\|X_0\|_{2p}^{4p}\right)^{\frac{1}{2}} \to 0 \quad \text{as $\delta_c\to 0$},
    \end{align*}
    where ``(1)'' applies Lem.~\ref{lem.smooth}. Hence, for sufficiently small $\delta_c$, $\mathbb{E}\|Y_{\delta_c}\|_{2p}^{2p}\leq\mathbb{E}\|Y_{0}\|_{2p}^{2p}+1$. As a result, denote $K_2:=K_1(p) \left(\mathbb{E}\|Y_{0}\|_{2p}^{2p}+1\right)$ that relies on $(p,\Phi,\sigma)$, we have:
    \begin{equation*}
         \mathbb{E}\|\mathrm{vec} \{C_N(-1)\} - \mathrm{vec} \{\Gamma(-1) \}\|^{2p}_{2p} \leq K_2 N^{-p}_c.
    \end{equation*}

    In a similar way, we can show:
    \begin{equation*}
        \mE\|\mathrm{vec} \{C_N(0)\} - \mathrm{vec} \{\Gamma(0)\} \|^{2p}_{2p} \leq K_3 N^{-p}_c,
    \end{equation*}
    for some $K_3:=K_3(p,\Phi,\sigma)>0$. Then letting $K=\max(K_2,K_3)$ concludes the proof.
\end{proof}

\begin{lemma}
    Let $p\geq 1$, then there exists a positive constant $K=K(p,\Phi,\sigma)$ such that:  
    \begin{equation*}
        \mathbb{E}\|C_N(0)^{-1}\|^{2p}_\mathrm{sp}\leq K_3.
    \end{equation*}
    \label{lem.boundedness-CN0-inverse}
\end{lemma}

\begin{proof}
    Apply Thm.~\ref{thm.bhansali1991convergence} with $p=1, q=2p, A_j=X_0(j) X_0(j)^\top$. Since the sequence $\{X_0(j)\}_{j\in J_N^c}$ is i.i.d., following a Gaussian distribution, we can verify that conditions (i)-(iii) in Thm.~\ref{thm.bhansali1991convergence} hold. The theorem then states that there exists a number $N_0$ and a nonnegative $L_{2p}$-integrable random variable $\Lambda_0$, such that for all $N_c>N_0$,
    \begin{equation*}
        \|C_N(0)^{-1}\|_{\mathrm{sp}} = \left\|\left[\frac{1}{N}\sum_j X_0(j) X_0(j)^\top\right]^{-1}\right\|_{\mathrm{sp}} \leq \Lambda_0 \,\, \text{a.s.}
    \end{equation*}

    The statement required then follows.
\end{proof}

\begin{lemma}
    Let $p$ be the same as in Prop.~\ref{prop.consistency-hatF}, then we have:
    \begin{equation*}
        P\{\sigma_{\min}(I+\wh{F})\leq 1 \vee \sigma_{\max}(I+\wh{F})\geq 3\} =  O\left(N_c^{-2p}+\delta_c^{4p}\right)
    \end{equation*}
    \label{lem.order-P(E)}
\end{lemma}

\begin{proof}
    Use the Markov inequality,
    \begin{equation*}
        P\{\sigma_{\min}(I+\wh{F})\leq 1\} \leq P\{|\sigma_{\min}(I+\wh{F})-2|\geq 1\} \leq \mathbb{E}|\sigma_{\min}(I+\wh{F})-2|^{4p}.
    \end{equation*}

    Use Lem.~\ref{lem.lip-continuty-singular-value}, we further have:
    \begin{align*}
        \mE|\sigma_{\min}(I+\wh{F})-2| & = \mathbb{E}|\sigma_{\min}(I+\wh{F})-\sigma_{\min}(2I)|  \leq \|\wh{F}-F\|_\mathrm{sp} + \|H\|_\mathrm{sp},
    \end{align*}
    where $H:=e^{\Phi \delta_c}-I$ with $\|H\|_\mathrm{sp}=O(\delta_c)$. As a result, 
    \begin{equation*}
        \mathbb{E}|\sigma_{\min}(I+\wh{F})-2|^{4p} \leq 2^{4p-1}(\mathbb{E}\|\wh{F}-F\|_\mathrm{sp}^{4p} + \|H\|_\mathrm{sp}^{4p}) = O\left(N_c^{-2p}+\delta_c^{4p}\right).
    \end{equation*}

    Therefore, we have $P\{\sigma_{\min}(I+\wh{F})\leq 1\} =  O\left(N_c^{-2p}+\delta_c^{4p}\right)$. In a similar way, we can show that $P\{\sigma_{\max}(I+\wh{F})\geq 3\} =  O\left(N_c^{-2p}+\delta_c^{4p}\right)$.
\end{proof}

In the following, we prove Prop.~\ref{prop.consistency-ou}.

\noindent\textbf{Proposition ~\ref{prop.consistency-ou}.} \emph{Assume Asm.~\ref{asm.identifiable-ou}. Let $p\geq 1$, then there exists a positive constant $K:=K(p,\Phi,\sigma)$ such that:
    \begin{equation*}
        \mathbb{E} \|\Phi-\Tilde{\Phi}\|^p_\mathrm{sp} \leq K u^p \left(\delta_c^{-p} N_c^{-\frac{p}{2}} + \delta_c^{2p} \right)
    \end{equation*}
    for sufficiently large $N_c, u$ and sufficiently small $\delta_c$.}

\begin{proof}
    We first show the consistency of $\bar{\Phi}$ defined in line 3 of Alg.\ref{alg:estimation-ou}. Recall that:
    \begin{equation*}
        \Phi = \frac{2}{\delta_c}(F-I)(F+I)^{-1} + \nu, \quad \nu = -\frac{2}{\delta_c} \eta (F+I)^{-1},
    \end{equation*}
    where
    \begin{equation*}
        \eta = -\frac{(\Phi\delta_c)^3}{2} \sum_{i=3}^\infty \frac{(i-2)(\Phi\delta_c)^{i-3}}{i!}.
    \end{equation*}

    Hence,
    \begin{equation*}
        \Phi - \bar{\Phi} = \frac{2}{\delta_c} \left\{(F-I)(F+I)^{-1} - (\bar{F}-I)(\bar{F}+I)^{-1} \right\} + \nu.
    \end{equation*}

    It follows that:
    \begin{align*}
         \|\Phi - \bar{\Phi}\|_\mathrm{sp} &\leq \frac{2}{\delta_c} \left\|(F-I)(F+I)^{-1} - (\bar{F}-I)(\bar{F}+I)^{-1} \right\|_\mathrm{sp} + \|\nu\|_\mathrm{sp} \\
        &= \frac{2}{\delta_c} \left\| (I+F)^{-1} (F-\bar{F})(I+\bar{F})^{-1} \right\|_\mathrm{sp} + \|\nu\|_\mathrm{sp} \\
        &\leq \frac{2}{\delta_c} \| (I+F)^{-1} \|_\mathrm{sp} \|F-\bar{F} \|_\mathrm{sp} \| (I+\bar{F})^{-1} \|_\mathrm{sp} + \|\nu\|_\mathrm{sp}.
    \end{align*}

    For $\|(I+F)^{-1} \|_\mathrm{sp}$, we have $\|(I+F)^{-1} \|_\mathrm{sp} = 1/\sigma_{\min}(I+F)$. Since $F=e^{\Phi\delta_c} = I+H$, where $\|H\|_\mathrm{sp}=O(\delta_c)$, we have due to Lem.~\ref{lem.lip-continuty-singular-value} that:
    \begin{equation*}
        |\sigma_{\min}(I+F) - 2| = |\sigma_{\min}(I+F) - \sigma_{\min}(2I)| \leq \|H\|_\mathrm{sp} = O(\delta_c).
    \end{equation*}

    Hence, there exists $\delta_0>0$ such that $\sigma_{\min}(I+F)\geq 1$ for all $\delta_c<\delta_0$. Consequently, we have $\|(I+F)^{-1} \|_\mathrm{sp} = 1/\sigma_{\min}(I+F) \leq 1$ is bounded.

    For $\|F-\bar{F} \|_\mathrm{sp}$, we have shown in Prop.~\ref{prop.consistency-hatF} that:
    \begin{equation*}
        E\|F-\bar{F}\|^p_\mathrm{sp} = O\left(N_c^{-\frac{p}{2}} + \delta_c^{4p}\right).
    \end{equation*}

    For $\| (I+\bar{F})^{-1} \|_\mathrm{sp}$, we have $\| (I+\bar{F})^{-1} \|_\mathrm{sp} = 1/\sigma_{\min}(I+\bar{F}) \leq 1$ due to Alg.~\ref{alg:estimation-ou} (line 2). Moreover, we have $\|\nu\|_\mathrm{sp}=O(\delta_c^2)$. Therefore, we have:
    \begin{align*}
        \mathbb{E} \|\Phi - \bar{\Phi}\|^p_\mathrm{sp} &\leq 2^{p-1} \left[ \left(\frac{2}{\delta_c}\right)^p \| (I+F)^{-1} \|^p_\mathrm{sp} \mathbb{E} \left\{\|F-\bar{F} \|^p_\mathrm{sp} \| (I+\bar{F})^{-1} \|^p_\mathrm{sp}\right\} + \|\nu\|^p_\mathrm{sp}\right] \\
        &\leq 2^{p-1} \left\{ \left(\frac{2}{\delta_c}\right)^p \mathbb{E} \|F-\bar{F} \|^p_\mathrm{sp} + \|\nu\|^p_\mathrm{sp}\right\} = O\left(\delta_c^{-p} N_c^{-\frac{p}{2}} + \delta_c^{2p}\right).\numberthis \label{eq.order-barPhi-Phi}
    \end{align*}

    We then show the consistency of $\Tilde{\Phi}$. Denote $\mathcal{E}$ as the event ``$\sigma_{\min}(I-\bar{\Phi})\geq u^{-1}$ and $\sigma_{\max}(I-\bar{\Phi})\leq u$'', we decompose
    \begin{equation*}
        \mathbb{E}\|\Phi-\Tilde{\Phi}\|^p_\mathrm{sp} = \mathbb{E}\|\Phi-\Tilde{\Phi}\|^p_\mathrm{sp} 1_\mathcal{E}+ \mathbb{E}\|\Phi-\Tilde{\Phi}\|^p_\mathrm{sp} 1_{^\neg \mathcal{E}}.
    \end{equation*}

    For the first term on the RHS, we have:
    \begin{equation*}
        \mathbb{E}\|\Phi-\Tilde{\Phi}\|^p_\mathrm{sp} 1_\mathcal{E} = \mathbb{E}\|\Phi-\bar{\Phi}\|^p_\mathrm{sp} 1_\mathcal{E} \leq \mathbb{E}\|\Phi-\bar{\Phi}\|^p_\mathrm{sp}.
    \end{equation*}

    For the second term on the RHS, we have:
    \begin{align*}
        \mathbb{E}\|\Phi-\Tilde{\Phi}\|^p_\mathrm{sp} 1_{^\neg \mathcal{E}} &\leq \|\Phi-\Tilde{\Phi}\|^p_\mathrm{sp} P(^\neg \mathcal{E}) \leq 2^{p-1}(\|I-\Tilde{\Phi}\|^p_\mathrm{sp}+\|I-\Phi\|^p_\mathrm{sp}) P(^\neg \mathcal{E}) \\
        &\leq 2^{p-1}(u^p+\|I-\Phi\|^p_\mathrm{sp}) P(^\neg \mathcal{E}),
    \end{align*}
    where the last inequality is due to $\|I-\Tilde{\Phi}\|_\mathrm{sp}=\sigma_{\max}(I-\Tilde{\Phi})\leq u$ (line 4 of Alg.~\ref{alg:estimation-ou}).

    We note that: 
    \begin{equation*}
        P(^\neg \mathcal{E}) \leq P\{\sigma_{\min}(I-\bar{\Phi})< u^{-1}\} + P\{\sigma_{\max}(I-\bar{\Phi})>u\}.
    \end{equation*}

    Use Lem.~\ref{lem.lip-continuty-singular-value}, we have:
    \begin{equation*}
        P\{\sigma_{\min}(I-\bar{\Phi})<u^{-1}\}  \leq P\{\sigma_{\min}(I-{\Phi})<u^{-1}+\|\bar{\Phi}-\Phi\|_\mathrm{sp}\} \leq \frac{\mathbb{E}\|\bar{\Phi}-\Phi\|^p_\mathrm{sp}}{\{\sigma_{\min}(I-{\Phi})-u^{-1}\}^p}.
    \end{equation*}

    In a similar way,
    \begin{equation*}
        P\{\sigma_{\max}(I-\bar{\Phi})>u\} \leq P\{u-\sigma_{\max}(I-\Phi)<\|\bar{\Phi}-\Phi\|_\mathrm{sp}\} \leq \frac{\mathbb{E}\|\bar{\Phi}-\Phi\|^p_\mathrm{sp}}{\{u-\sigma_{\max}(I-\Phi)\}^p}.
    \end{equation*}

    Therefore,
    \begin{align*}
        \mathbb{E}\|\Phi-\Tilde{\Phi}\|^p_\mathrm{sp} 1_{^\neg \mathcal{E}} &\leq  2^{p-1}(u^p+\|I-\Phi\|^p_\mathrm{sp}) P(^\neg \mathcal{E}) \\
        &\leq 2^{p-1}(u^p+\|I-\Phi\|^p_\mathrm{sp}) \left[\frac{\mathbb{E}\|\bar{\Phi}-\Phi\|^p_\mathrm{sp}}{\{\sigma_{\min}(I-{\Phi})-u^{-1}\}^p}+\frac{\mathbb{E}\|\bar{\Phi}-\Phi\|^p_\mathrm{sp}}{\{u-\sigma_{\max}(I-\Phi)\}^p}\right]\\
        &=O(u^p \mathbb{E}\|\bar{\Phi}-\Phi\|^p_\mathrm{sp}) = O\left\{ u^p \left(\delta_c^{-p} N_c^{-\frac{p}{2}} + \delta_c^{2p} \right)\right\}
    \end{align*}

    Combining the above, we have shown the proposition.
\end{proof}

\textbf{Proposition \ref{prop.consistency-hat-Sigma}.} \emph{Assume Asm.~\ref{asm.identifiable-ou}, then we have $\hat{\Sigma} \overset{\mathcal{P}}{\to} \Sigma$ if $N_c\to \infty, \delta_c\to 0, u\to \infty$ and $\exists p\geq 1$ such that $u^p \{\delta_c^{-p} N_c^{-\frac{p}{2}} + \delta_c^{2p} \} \to 0$.}

\begin{proof}
    According to \cite[Coro.~3.2.1]{lutkepohl2005new}, we have $\hat{\Omega}\overset{\mathcal{P}}{\to} \Omega$ as $N_c\to \infty$. Moreover, we have $\hat{F}\overset{\mathcal{P}}{\to} F$ as $N_c\to\infty, \delta_c\to 0$ according to Prop.~\ref{prop.consistency-hatF}. Finally, under $u^p \{\delta_c^{-p} N_c^{-\frac{p}{2}} + \delta_c^{2p} \} \to 0$, we have $\Tilde{\Phi}\overset{\mathcal{P}}{\to} \Phi$ according to Prop.~\ref{prop.consistency-ou}. The desired statement then follows from the continuous mapping theorem.
\end{proof}

In the following, we prove Prop.~\ref{prop.form-of-filtering-eq-in-OU} and Prop.~\ref{prop.filtering-eq-uniqueness-solution}.

\textbf{Proposition~\ref{prop.form-of-filtering-eq-in-OU}.} \emph{The processes $m_t = \mathbb{E}\left\{X_{V\backslash C}(t)|\mathcal{F}_t\right\}, y_t = \frac{1}{\sigma^2}\mathrm{Var}\left\{X_{V\backslash C}(t)|\mathcal{F}_t\right\}$ satisfy the following system of filtering equations:}
    \begin{align*}
            &dm_t = \{A_t X_C(t) + B_t m_t\} dt + C_t d X_C(t),\\
            &\dot{y}_t = \Phi_{V\backslash C} y_t + y_t \Phi_{V\backslash C}^\top + I - y_t \Phi_{C, V\backslash C}^\top \Phi_{C, V\backslash C} y_t,
    \end{align*}
    \emph{where we denote $\Phi_{V\backslash C}$
 in short of $\Phi_{V\backslash C,V\backslash C}$ and}          \begin{equation*}
        A_t := \Phi_{V\backslash C, C} - y_t  \Phi_{C, V\backslash C}^\top  \Phi_{C}, \quad B_t := \Phi_{V\backslash C} - y_t \Phi_{C, V\backslash C}^\top\Phi_{C, V\backslash C}, \quad C_t := y_t \Phi_{C, V\backslash C}^\top.
    \end{equation*}

    \emph{Moreover, we have $X_0\sim\mathcal{N}(0,\sigma^2 \Upsilon)$, with $\Upsilon=(I-\Phi)^{-1} (I-\Phi^\top)^{-1}$. Therefore, the initial values of (\ref{eq.pi-t}a)-(\ref{eq.pi-t}b) are (respectively):}
    \begin{equation*}
        m_0 = \Upsilon_{V\backslash C,C} \Upsilon_C^{-1} X_C(0), \quad y_0 = \Upsilon_{V\backslash C} - \Upsilon_{V\backslash C,C} \Upsilon_C^{-1} \Upsilon_{C,V\backslash C}.
    \end{equation*}

\begin{proof}
    In \eqref{eq.filtering-setup}, substitute with
\begin{align*}
    & \quad \quad\theta = X_{V\backslash C}, \quad \xi=X_{C}, \quad W_1 = W_{V\backslash C}, \quad W_2 = W_{C}, \\
    &a_0 = \Phi_{V\backslash C, C}\cdot X_C(t), \quad a_1 = \Phi_{V\backslash C}, \quad b_1=\sigma I_{|V\backslash C|}, \quad b_2 = \mathrm{0}_{|V\backslash C|\times |C|}, \\
    &A_0 = \Phi_{C} \cdot X_C(t), \quad A_1=\Phi_{C, V\backslash C}, \quad B_1 = 0_{|C|\times |V\backslash C|}, \quad B_2 = \sigma I_{|C|},
\end{align*}
we have $k=|V\backslash C|$, $l=|C|$, and $m_t = \mathbb{E}\left\{X_{V\backslash C}(t)|\mathcal{F}_t^{C}\right\}$. Then, apply Thm.~\ref{eq.filtering-thm} with the substitution $y_t = \frac{1}{\sigma^2} \gamma_t$, we obtain Prop.~\ref{prop.form-of-filtering-eq-in-OU}.
\end{proof}

\textbf{Proposition~\ref{prop.filtering-eq-uniqueness-solution}.} \emph{The system of equations (\ref{eq.pi-t}a)-(\ref{eq.pi-t}b), subject to the initial condition on $m_0, \gamma_0$, has a unique, continuous, $\mathcal{F}_t$-adapted solution for any $t\in [0,T]$.}

\begin{proof}
    Prop.~\ref{prop.filtering-eq-uniqueness-solution} is the direct application of Thm.~\ref{thm.filtering-unique}.
\end{proof}

\subsection{Proof of Prop.~\ref{prop.consistency-projection-estimator}}
\label{appx.proof-consistency-projection-estimator}

We first introduce some notational convenience. For a nonsingular matrix $A$, we denote $\lambda_{\min}(A)$ and $\lambda_{\max}(A)$ as the smallest and largest eigenvalues of $A$, respectively. A \emph{principal submatrix} $B$ of $A$ is a square submatrix formed by keeping the same set of row and column indices. For any vector $v\in \mathbb{R}^d$, we note that $\|v\|_\mathrm{sp}=\|v\|_2$ equals  the Euclidean norm. 

For $u>1$, denote $U_y(u,T), U_A(u,T), U_B(u,T), U_C(u,T)$ as the upper bounds of $y=y(t,\Phi,y_0), A=A(t,\Phi,y_0), B=B(t,\Phi,y_0), C=C(t,\Phi,y_0)$, in terms of the spectral norm, over the set $\left\{(t,\Phi,y_0)\mid 0\leq t\leq T, \|\Phi\|_{\mathrm{sp}} \leq u-1, \|y_0\|_\mathrm{sp} \leq u^2+u^6\right\}$. In this regard, we have $\|\tilde{y}_t\|_\mathrm{sp}\leq U_y(u,T), \|\tilde{A}_t\|_\mathrm{sp}\leq U_A(u,T)$, etc., for every $t$. We show the detailed forms of these bounds in Coro.~\ref{coro.tech-lem-Uy-UA-UB-UC}. 

In the following, for $p\geq 1$, we introduce three lemmas that discuss the order of $\mathbb{E}\|m_0-\hat{m}_0\|_\mathrm{sp}^p$ (Lem.~\ref{lem.consistency-initial-value-estimate}), $\sup_{t\in [0,T]} \mathbb{E}\|m_t-\Tilde{m}_t\|_\mathrm{sp}^p$ (Lem.~\ref{lem.m-tildem-consistency}), and $\sup_{t\in [0,T]} \mathbb{E}\|\Tilde{m}_t-\hat{m}_t\|_\mathrm{sp}^p$ (Lem.~\ref{lem.tildem-hatm-consistency}), respectively.

\begin{lemma}
    Consider the estimators of the initial values of \eqref{eq.pi-t}, namely
    \begin{equation*}
        \wh{m}_0 = \Tilde{\Upsilon}_{V\backslash C,C} \Tilde{\Upsilon}_C^{-1} X_C, \quad y_0 = \Tilde{\Upsilon}_{V\backslash C} - \Tilde{\Upsilon}_{V\backslash C,C} \Tilde{\Upsilon}_C^{-1} \Tilde{\Upsilon}_{C,V\backslash C},
    \end{equation*}
    where $\Tilde{\Upsilon}=(I-\Tilde{\Phi})^{-1} (I-\Tilde{\Phi}^\top)^{-1}$. We have for $p\geq 1$,
    \begin{equation*}
        \mathbb{E}\|\wh{m}_0-m_0\|^p_\mathrm{sp} = O(u^{6p} \mathbb{E}\|\Tilde{\Phi}-\Phi\|^p_\mathrm{sp}), \quad \mathbb{E}\|\wh{y}_0-y_0\|^p_\mathrm{sp} = O(u^{6p} \mathbb{E}\|\Tilde{\Phi}-\Phi\|^p_\mathrm{sp}).
    \end{equation*}

    Furthermore, the estimator $\wh{y}_0$ has a spectral norm bounded by $u^2 + u^6$.
    \label{lem.consistency-initial-value-estimate}
\end{lemma}

\begin{proof}
    Denote $B:=(I-\Phi)^{-1}$ and $\Tilde{B}:=(I-\Tilde{\Phi})^{-1}$, with $\Upsilon=BB^\top$. We first discuss the consistency of $\wh{m}_0$. We have:
    \begin{equation*}
        \wh{m}_0 - m_0 = (\Tilde{\Upsilon}_{V\backslash C,C}\Tilde{\Upsilon}_C^{-1} - \Upsilon_{V\backslash C,C}\Upsilon_{C}^{-1}) X_C(0).
    \end{equation*}

    We note that:
    \small
    \begin{align*}
        \|\Tilde{\Upsilon}_{V\backslash C,C}\Tilde{\Upsilon}_C^{-1} &- \Upsilon_{V\backslash C,C}\Upsilon_{C}^{-1}\|_\mathrm{sp} = \|\Tilde{\Upsilon}_{V\backslash C,C}\Tilde{\Upsilon}_C^{-1} - \Tilde{\Upsilon}_{V\backslash C,C}{\Upsilon}_C^{-1} + \Tilde{\Upsilon}_{V\backslash C,C}{\Upsilon}_C^{-1} - \Upsilon_{V\backslash C,C}\Upsilon_{C}^{-1}\|_\mathrm{sp} \\
        &\leq \|\Tilde{\Upsilon}_{V\backslash C,C}\Tilde{\Upsilon}_C^{-1} - \Tilde{\Upsilon}_{V\backslash C,C}{\Upsilon}_C^{-1}\|_\mathrm{sp} + \|\Tilde{\Upsilon}_{V\backslash C,C}{\Upsilon}_C^{-1} - \Upsilon_{V\backslash C,C}\Upsilon_{C}^{-1}\|_\mathrm{sp} \\
        &\leq \|\Tilde{\Upsilon}_{V\backslash C,C}\|_\mathrm{sp} \|\Tilde{\Upsilon}_C^{-1}-{\Upsilon}_C^{-1}\|_\mathrm{sp} + \|\Tilde{\Upsilon}_{V\backslash C,C}-{\Upsilon}_{V\backslash C,C}\|_\mathrm{sp} \|{\Upsilon}_C^{-1}\|_\mathrm{sp}\\
        &= \|\Tilde{\Upsilon}_{V\backslash C,C}\|_\mathrm{sp} \|\Upsilon_{C}^{-1} (\Upsilon_{C}-\Tilde{\Upsilon}_C) \Tilde{\Upsilon}_C^{-1}\|_\mathrm{sp} + \|\Tilde{\Upsilon}_{V\backslash C,C}-{\Upsilon}_{V\backslash C,C}\|_\mathrm{sp} \|{\Upsilon}_C^{-1}\|_\mathrm{sp}\\
        &\leq \|\Tilde{\Upsilon}_{V\backslash C,C}\|_\mathrm{sp} \|\Upsilon_{C}^{-1}\|_\mathrm{sp}\|\Upsilon_{C}-\Tilde{\Upsilon}_C\|_\mathrm{sp} \|\Tilde{\Upsilon}_C^{-1}\|_\mathrm{sp} + \|\Tilde{\Upsilon}_{V\backslash C,C}-{\Upsilon}_{V\backslash C,C}\|_\mathrm{sp} \|{\Upsilon}_C^{-1}\|_\mathrm{sp} \\
        &\overset{(1)}{\leq} \|\Tilde{\Upsilon}\|_\mathrm{sp} \|\Upsilon^{-1}\|_\mathrm{sp}\|\Upsilon-\Tilde{\Upsilon}\|_\mathrm{sp} \|\Tilde{\Upsilon}^{-1}\|_\mathrm{sp} + \|\Tilde{\Upsilon}-{\Upsilon}\|_\mathrm{sp} \|{\Upsilon}^{-1}\|_\mathrm{sp} \\
        &\overset{(2)}{\leq} (u^4+1) \|\Upsilon^{-1}\|_\mathrm{sp}\|\Upsilon-\Tilde{\Upsilon}\|_\mathrm{sp},
    \end{align*}
    \normalsize
    where (1) applies $\|\Upsilon_C^{-1}\|_\mathrm{sp}=\lambda_{\min}(\Upsilon_C)^{-1}\leq \lambda_{\min}(\Upsilon)^{-1} = \|\Upsilon^{-1}\|_\mathrm{sp}$ and $\|\Tilde{\Upsilon}_C^{-1}\|_\mathrm{sp}=\lambda_{\min}(\Tilde{\Upsilon}_C)^{-1}\leq \lambda_{\min}(\Tilde{\Upsilon})^{-1} = \|\Tilde{\Upsilon}^{-1}\|_\mathrm{sp}$ due to the Cauchy interlacing Thm.~\ref{thm.tech-lem-cauchy-interlacing}; and (2) applies $\|\Tilde{\Upsilon}\|_\mathrm{sp}=\lambda_{\max}(\Tilde{\Upsilon})=\sigma_{\max}(\Tilde{B})^2\leq u^2$ and $\|\Tilde{\Upsilon}^{-1}\|_\mathrm{sp}=\lambda_{\min}(\Tilde{\Upsilon})^{-1}=\sigma_{\min}(\Tilde{B})^{-2}\leq u^{2}$ due to Alg.~\ref{alg:estimation-ou} (l. 4).
    
    Since $X_C(0)=X_{C}^{(j)}(0)$ is independent with $\Tilde{\Upsilon}$ that relies on $J_N^c$, we have:
    \begin{equation}
        \mathbb{E}\|\hat{m}_0-m_0\|_\mathrm{sp}^p \leq (u^4+1)^p \|\Upsilon^{-1}\|_\mathrm{sp}^p \mathbb{E}\|\Upsilon-\Tilde{\Upsilon}\|_\mathrm{sp}^p \mathbb{E}\|X_C(0)\|_\mathrm{sp}^p.
        \label{eq.hatm0-m0-1}
    \end{equation}

    Furthermore, since $\Tilde{\Upsilon}-\Upsilon=\Tilde{B}\Tilde{B}^\top-BB^\top=\Tilde{B}(\Tilde{B}^\top-B^\top)+(\Tilde{B}-B)B^\top$, we have:
    \begin{align}
         \mathbb{E}\|\Tilde{\Upsilon}-\Upsilon\|^p_\mathrm{sp} &\leq 2^p(\|\Tilde{B}\|^p_\mathrm{sp} + \|B\|^p_\mathrm{sp}) \mathbb{E}\|\Tilde{B}-B\|^p_\mathrm{sp} \notag\\
         &\overset{(1)}{\leq} 2^p(u^p+\|I-\Phi\|^p_\mathrm{sp}) \mathbb{E}\|\Tilde{B}-B\|^p_\mathrm{sp} \notag\\
         &\leq 2^p(1+u^p+\|\Phi\|^p_\mathrm{sp}) \mathbb{E}\|(I-\Phi)^{-1}-(I-\Tilde{\Phi})^{-1}\|^p_\mathrm{sp} \notag\\
         &=2^p(1+u^p+\|\Phi\|^p_\mathrm{sp}) \mathbb{E}\|(I-\Phi)^{-1}(\Phi-\Tilde{\Phi})(I-\Tilde{\Phi})^{-1}\|^p_\mathrm{sp} \notag\\
         &\leq 2^p(1+u^p+\|\Phi\|^p_\mathrm{sp}) \mathbb{E} \left\{\|(I-\Phi)^{-1}\|_\mathrm{sp}^p \|\Phi-\Tilde{\Phi}\|_\mathrm{sp}^p \|(I-\Tilde{\Phi})^{-1}\|^p_\mathrm{sp} \right\} \notag\\
         &\overset{(2)}{\leq} 2^p(1+u^p+\|\Phi\|^p_\mathrm{sp})u^p \|(I-\Phi)^{-1}\|_\mathrm{sp}^p \mathbb{E} \left\{\|\Phi-\Tilde{\Phi}\|_\mathrm{sp}^p\right\} \notag\\
         &= O(u^{2p} \mathbb{E}\|\Phi-\Tilde{\Phi}\|_\mathrm{sp}^p)
         \label{eq.hatm0-m0-2}
    \end{align}
    where (1), (2) apply $\|I-\Tilde{\Phi}\|_\mathrm{sp}\leq u, \|(I-\Tilde{\Phi})^{-1}\|_\mathrm{sp}\leq u$, respectively, due to Alg.~\ref{alg:estimation-ou} (l. 4).

    Combine \eqref{eq.hatm0-m0-1} and \eqref{eq.hatm0-m0-2}, we have:
    \begin{equation*}
        \mathbb{E}\|\hat{m}_0-m_0\|_\mathrm{sp}^p = O(u^{6p} \mathbb{E}\|\Phi-\Tilde{\Phi}\|_\mathrm{sp}^p).
    \end{equation*}

    We then discuss the consistency of $\wh{y}_0$. We have:
    \small
    \begin{align*}
        \|y_0 - \wh{y}_0\|_\mathrm{sp} &= \|\Upsilon_{V\backslash C} - \Tilde{\Upsilon}_{V\backslash C} + \Tilde{\Upsilon}_{V\backslash C,C} \Tilde{\Upsilon}_C^{-1} \Tilde{\Upsilon}_{C,V\backslash C} - \Upsilon_{V\backslash C,C}\Upsilon_C^{-1} \Upsilon_{C,V\backslash C} \|_\mathrm{sp} \\
        &\leq \|\Upsilon_{V\backslash C} - \Tilde{\Upsilon}_{V\backslash C}\|_\mathrm{sp} + \|\Tilde{\Upsilon}_{V\backslash C,C} \Tilde{\Upsilon}_C^{-1} \Tilde{\Upsilon}_{C,V\backslash C} - \Upsilon_{V\backslash C,C}\Upsilon_C^{-1} \Upsilon_{C,V\backslash C}\|_\mathrm{sp}\\
        &\leq \|\Upsilon - \Tilde{\Upsilon}\|_\mathrm{sp} + \|\Tilde{\Upsilon}_{V\backslash C,C} \Tilde{\Upsilon}_C^{-1} \Tilde{\Upsilon}_{C,V\backslash C} - \Tilde{\Upsilon}_{V\backslash C,C} {\Upsilon}_C^{-1} \Tilde{\Upsilon}_{C,V\backslash C} \\
        &\manyquad[6] + \,\, \Tilde{\Upsilon}_{V\backslash C,C} {\Upsilon}_C^{-1} \Tilde{\Upsilon}_{C,V\backslash C} - \Tilde{\Upsilon}_{V\backslash C,C} {\Upsilon}_C^{-1} {\Upsilon}_{C,V\backslash C} \\
        &\manyquad[6] + \,\,  \Tilde{\Upsilon}_{V\backslash C,C} {\Upsilon}_C^{-1} {\Upsilon}_{C,V\backslash C} -  \Upsilon_{V\backslash C,C}\Upsilon_C^{-1} \Upsilon_{C,V\backslash C}\|_\mathrm{sp} \\
        &\leq \|\Upsilon - \Tilde{\Upsilon}\|_\mathrm{sp} + \|\Tilde{\Upsilon}_{V\backslash C,C} \Upsilon_C^{-1} (\Upsilon_C - \Tilde{\Upsilon}_C) \Tilde{\Upsilon}_C^{-1} \Tilde{\Upsilon}_{C,V\backslash C}\|_\mathrm{sp}\\
        &\manyquad[6] + \,\, \|\Tilde{\Upsilon}_{V\backslash C,C} \Upsilon_C^{-1}(\Upsilon_{C,V\backslash C}-\Tilde{\Upsilon}_{C,V\backslash C})\|_\mathrm{sp} \\
        &\manyquad[6] + \,\,  \|(\Tilde{\Upsilon}_{V\backslash C,C}-\Upsilon_{V\backslash C,C}) \Upsilon_C^{-1} \Upsilon_{C,V\backslash C}\|_\mathrm{sp}
    \end{align*}
    \normalsize

    It follows that:
    \begin{align*}
        \|y_0 - \wh{y}_0\|_\mathrm{sp} &\leq \|\Upsilon - \Tilde{\Upsilon}\|_\mathrm{sp} + \|\Tilde{\Upsilon}\|_\mathrm{sp}^2 \|\Upsilon^{-1}\|_\mathrm{sp}\|\Tilde{\Upsilon}^{-1}\|_\mathrm{sp}\|\Upsilon-\Tilde{\Upsilon}\|_\mathrm{sp} + 2\|\Tilde{\Upsilon}\|_\mathrm{sp}  \|{\Upsilon}^{-1}\|_\mathrm{sp}\|\Upsilon-\Tilde{\Upsilon}\|_\mathrm{sp}\\
        &\leq (1+u^4\|\Upsilon^{-1}\|_\mathrm{sp} + 2u^2 \|\Upsilon^{-1}\|_\mathrm{sp}) \|\Upsilon-\Tilde{\Upsilon}\|_\mathrm{sp}.
    \end{align*}

    Apply \eqref{eq.hatm0-m0-2}, we then have:
    \begin{equation*}
        \mathbb{E}\|y_0 - \wh{y}_0\|_\mathrm{sp}^p = O(u^{6p} \mathbb{E}\|\Phi-\Tilde{\Phi}\|_\mathrm{sp}^p).
    \end{equation*}

    Besides, we have:
    \begin{equation*}
        \|\wh{y}_0\| \leq \|\Tilde{\Upsilon}_{V\backslash C}\| + \|\Tilde{\Upsilon}_{V\backslash C,C}\| \|\Tilde{\Upsilon}_C^{-1}\| \|\Tilde{\Upsilon}_{C,V\backslash C}\| \leq \|\Tilde{\Upsilon}\| + \|\Tilde{\Upsilon}\|^2 \|\Tilde{\Upsilon}_C^{-1}\| \leq u^6+u^2.
    \end{equation*}

    This has concluded the proof.
\end{proof}

Recall that $\Tilde{m}$ is the continuous process satisfying:
\begin{equation*}
    d\Tilde{m}_t = (\Tilde{A}_t \xi_t + \Tilde{B}_t \Tilde{m}_t)dt + \Tilde{C}_t d\xi_t, \quad \tilde{m}_0 = \hat{m}_0.
\end{equation*}

\begin{lemma}
    For $p\geq 2$, we have:
    \begin{equation*}
        \sup_{t\in[0,T]} \mathbb{E}\|m_t - \Tilde{m}_t\|^p_\mathrm{sp} = o\{\exp(e^{2puT}) \mathbb{E}\|\Phi-\bar{\Phi}\|_\mathrm{sp}^p\}.
    \end{equation*}
    \label{lem.m-tildem-consistency}
\end{lemma}

\begin{proof}
    Denote $z_t := \sup_{s\in[0,t]} \mathbb{E}\|m_s - \Tilde{m}_s\|^p_\mathrm{sp}$. We note $z$ is a monotone function. Therefore, it is integrable on $[0,T]$. In the following, we will use the Gronwall inequality (Lem.~\ref{lem.Gronwall}) to derive a bound for $z_T$. We first have:
    \small
    \begin{align*}
        &z_t = \sup_{s\in [0,t]} \mathbb{E}\left\|m_0-\Tilde{m}_0  + \int_0^s (A_\tau-\Tilde{A}_\tau)\xi_u + (B_\tau m_\tau-\Tilde{B}_\tau \Tilde{m}_\tau) d\tau + \int_0^s (C_\tau-\Tilde{C}_\tau) d\xi_\tau\right\|_\mathrm{sp}^p \\
        &\leq 2^{2p} \sup_{s\in [0,t]} z_0 + \mathbb{E}\left\|\int_0^s (A_\tau-\Tilde{A}_\tau)\xi_\tau + (B_\tau m_\tau-\Tilde{B}_\tau \Tilde{m}_\tau) d\tau\right\|_\mathrm{sp}^p  + \mathbb{E}\left\|\int_0^s (C_\tau-\Tilde{C}_\tau) d\xi_\tau\right\|_\mathrm{sp}^p \\
        &\leq 2^{3p} \sup_{s\in [0,t]} z_0 + \mathbb{E}\left\|\int_0^s (A_\tau-\Tilde{A}_\tau)\xi_\tau d\tau\right\|_\mathrm{sp}^p + \mathbb{E}\left\|\int_0^s (B_\tau m_\tau-\Tilde{B}_\tau \Tilde{m}_\tau) d\tau\right\|_\mathrm{sp}^p + \mathbb{E}\left\|\int_0^s (C_\tau-\Tilde{C}_\tau) d\xi_\tau\right\|_\mathrm{sp}^p.
    \end{align*}
    \normalsize

    We note that, according to Lem.~\ref{lem.tech-lem-p-norm-of-integral},
    \begin{align*}
        &\mathbb{E}\left\|\int_0^s (A_u-\Tilde{A}_\tau)\xi_\tau d\tau\right\|_\mathrm{sp}^p \leq T^{p-1} \mathbb{E} \int_0^s \|A_\tau-\Tilde{A}_\tau\|_\mathrm{sp}^p \|\xi_\tau \|_\mathrm{sp}^p d\tau, \\
        &\mathbb{E}\left\|\int_0^s (B_\tau m_\tau-\Tilde{B}_\tau \Tilde{m}_\tau) d\tau\right\|_\mathrm{sp}^p \leq 2^p T^{p-1}\mathbb{E}\int_0^t \|B_\tau -\Tilde{B}_\tau\|_\mathrm{sp}^p \|m_u\|_\mathrm{sp}^p + \|\Tilde{B}_\tau\|_\mathrm{sp}^p \|m_\tau-\Tilde{m}_\tau\|_\mathrm{sp}^p d\tau.
    \end{align*}

    Furthermore, according to Lem.~\ref{lem.tech-lem-p-norm-of-integral} and the Burkholder-Davis-Gundy (BDG) inequality (Lem.~\ref{lem.burkholder-david}),
    \begin{align*}
        &\mathbb{E}\left\|\int_0^s (C_\tau-\Tilde{C}_\tau) d\xi_\tau\right\|_\mathrm{sp}^p = \mathbb{E}\left\|\int_0^s (C_\tau-\Tilde{C}_\tau)\lambda_C(\tau)d\tau+ \sigma \int_0^s (C_\tau-\Tilde{C}_\tau)  dW_C(\tau)\right\|_\mathrm{sp}^p \\
        &\leq 2^p\mathbb{E}\left\|\int_0^s (C_\tau-\Tilde{C}_\tau)\lambda_C(\tau)d\tau \right\|_\mathrm{sp}^p + 2^p \sigma^p \mathbb{E}\left\|\int_0^s (C_\tau-\Tilde{C}_\tau)  dW_C(\tau)\right\|_\mathrm{sp}^p \\
        &\leq 2^p T^{p-1} \mathbb{E} \int_0^s \|C_\tau-\Tilde{C}_\tau\|_\mathrm{sp}^p \|\lambda_C(\tau)\|_\mathrm{sp}^p d\tau + 2^p \sigma^p K_p T^{\frac{p}{2}-1}  \mathbb{E} \int_0^s \|C_\tau-\Tilde{C}_\tau\|_\mathrm{sp}^p d\tau \\
        &\leq (2^p T^{p-1}+2^p \sigma^p K_p T^{\frac{p}{2}-1}) \mathbb{E}  \int_0^s \|C_\tau-\Tilde{C}_\tau\|_\mathrm{sp}^p (\|\lambda_C(\tau)\|_\mathrm{sp}^p+1) d\tau,
    \end{align*}
    where $K_p$ is a positive constant from the BDG inequality.

    Hence, we have:
    \begin{align*}
        z_t &\leq 2^{3p} z_0 +  2^{4p} (T^{p-1}+\sigma^p K_p T^{\frac{p}{2}-1}) \sup_{s\in[0,t]} \mathbb{E} \int_0^s \|A_\tau-\Tilde{A}_\tau\|_\mathrm{sp}^p \|\xi_\tau \|_\mathrm{sp}^p \\
        &\manyquad[2]+ \|B_\tau -\Tilde{B}_\tau\|_\mathrm{sp}^p \|m_\tau\|_\mathrm{sp}^p + \|\Tilde{B}_\tau m_\tau-\Tilde{B}_\tau \Tilde{m}_\tau\|_\mathrm{sp}^p + \|C_\tau-\Tilde{C}_\tau\|_\mathrm{sp}^p (\|\lambda_C(\tau)\|_\mathrm{sp}^p+1)d\tau\\
        &= 2^{3p} z_0 +  2^{4p} (T^{p-1}+\sigma^p K_p T^{\frac{p}{2}-1}) \mathbb{E} \int_0^t \|A_\tau-\Tilde{A}_\tau\|_\mathrm{sp}^p \|\xi_\tau \|_\mathrm{sp}^p \\
        &\manyquad[2]+ \|B_\tau -\Tilde{B}_\tau\|_\mathrm{sp}^p \|m_\tau\|_\mathrm{sp}^p + \|\Tilde{B}_\tau\|_\mathrm{sp}^p \|m_\tau-\Tilde{m}_\tau\|_\mathrm{sp}^p + \|C_\tau-\Tilde{C}_\tau\|_\mathrm{sp}^p (\|\lambda_C(\tau)\|_\mathrm{sp}^p+1)d\tau\\
        &\leq 2^{3p} z_0 +  2^{4p} (T^{p-1}+\sigma^p K_p T^{\frac{p}{2}-1}) \mathbb{E} \int_0^t \|A_\tau-\Tilde{A}_\tau\|_\mathrm{sp}^p \|\xi_\tau \|_\mathrm{sp}^p \\
        &\manyquad[2]+ \|B_\tau -\Tilde{B}_\tau\|_\mathrm{sp}^p \|m_\tau\|_\mathrm{sp}^p + \|C_\tau-\Tilde{C}_\tau\|_\mathrm{sp}^p (\|\lambda_C(\tau)\|_\mathrm{sp}^p+1) + U_B(u,T)^p \|m_\tau-\Tilde{m}_\tau\|_\mathrm{sp}^p d\tau.\\
    \end{align*}
    
    First, we note that the r.v.'s $\xi_\tau,m_\tau,\lambda_C(\tau)$ that rely on $j\in J_N$ are independent with $\|A_\tau-\Tilde{A}_\tau\|_\mathrm{sp}^p, \|B_\tau-\Tilde{B}_\tau\|_\mathrm{sp}^p, \|C_\tau-\Tilde{C}_\tau\|_\mathrm{sp}^p$ that rely on $J_N^c$. Besides, according to Coro.~\ref{coro.lp-bounded}, 
    \begin{equation*}
        K^\prime:=\max\left\{\sup_{u\in [0,T]} \mathbb{E} \|\xi_u \|_\mathrm{sp}^p, \sup_{u\in [0,T]} \mathbb{E} \|m_u \|_\mathrm{sp}^p, \sup_{u\in [0,T]} \mathbb{E} \|\lambda_C(u) \|_\mathrm{sp}^p + 1\right\}<\infty.
    \end{equation*}
    
    Second, denote $\zeta(\tau) = \mathbb{E} \left( \|A_\tau-\Tilde{A}_\tau\|^p_\mathrm{sp} + \|B_\tau-\Tilde{B}_\tau\|^p_\mathrm{sp} + \|C_\tau-\Tilde{C}_\tau\|^p_\mathrm{sp}\right)$, we have according to Lem.~\ref{lem.E|A-TildeA|-|B-TildeB|} that:
    \begin{align*}
        &\sup_{\tau\in [0,T]} \zeta(\tau) \leq 2^{4p}(u^p +\|\Phi\|_\mathrm{sp}^{2p}) \exp(2^p T^{p-1}\beta) \mathbb{E}\|y_0-\wh{y}_0\|^p_\mathrm{sp} \\
        &\manyquad[2]+ 2^{8p}\Big[\sup_t \|y_t\|_\mathrm{sp}^p + (1+u^p+\|\Phi\|_\mathrm{sp}^p)U_y(u,T)^p \\
        &\manyquad[2]+(u^p+\|\Phi\|_\mathrm{sp}^{2p})\{1+(u^p+\|\Phi\|_\mathrm{sp}^{p})\sup_t \|y_t\|^p_\mathrm{sp}\} T^p \exp(2^p T^{p-1}\beta) U_y(u,T)^p \Big] \mathbb{E} \|\Phi-\Tilde{\Phi} \|_\mathrm{sp}^p,
    \end{align*}
    which is finite due to Lem.~\ref{eq.hatm0-m0-2} and Prop.~\ref{prop.consistency-ou}, where
    \begin{equation*}
        \beta := 2^{3p+1} \|\Phi\|^p_\mathrm{sp} + 2^{4p+1}\left\{\|\Phi\|^{2p}_\mathrm{sp} \sup_t\|y_t\|^p_\mathrm{sp} + u^{2p} U_y(u,T)^p\right\}.
    \end{equation*}

    Finally, according to the Fubini's theorem,
    \begin{equation*}
        \mathbb{E}\int_0^t \|m_\tau-\Tilde{m}_\tau\|_\mathrm{sp}^p d\tau = \int_0^t  \mathbb{E} \|m_\tau-\Tilde{m}_\tau\|_\mathrm{sp}^p d\tau \leq \int_0^t z_\tau d\tau < \infty
    \end{equation*}
    are exchangeable since $z_t$ is integrable on $[0,T]$.

    Combining the above, we have:
    \small
    \begin{align*}
        z_t &\leq 2^{3p} z_0 +  2^{4p} (T^{p-1}+\sigma^p K_p T^{\frac{p}{2}-1}) \mathbb{E} \int_0^t \|A_\tau-\Tilde{A}_\tau\|_\mathrm{sp}^p \|\xi_\tau \|_\mathrm{sp}^p \\
        &\manyquad[2]+ \|B_\tau -\Tilde{B}_\tau\|_\mathrm{sp}^p \|m_\tau\|_\mathrm{sp}^p + \|C_\tau-\Tilde{C}_\tau\|_\mathrm{sp}^p (\|\lambda_C(\tau)\|_\mathrm{sp}^p+1) + U_B(u,T)^p \|m_\tau-\Tilde{m}_\tau\|_\mathrm{sp}^p d\tau\\
        &= 2^{3p} z_0 +  2^{4p} (T^{p-1}+\sigma^p K_p T^{\frac{p}{2}-1}) \int_0^t \mathbb{E}\|A_\tau-\Tilde{A}_\tau\|_\mathrm{sp}^p  \mathbb{E}\|\xi_\tau \|_\mathrm{sp}^p \\
        &\manyquad[2]+  \mathbb{E}\|B_\tau -\Tilde{B}_\tau\|_\mathrm{sp}^p  \mathbb{E}\|m_\tau\|_\mathrm{sp}^p +  \mathbb{E}\|C_\tau-\Tilde{C}_\tau\|_\mathrm{sp}^p (\|\lambda_C(\tau)\|_\mathrm{sp}^p+1) + U_B(u,T)^p  \mathbb{E}\|m_\tau-\Tilde{m}_\tau\|_\mathrm{sp}^p d\tau \\
        &\leq 2^{3p} z_0 +  2^{4p} (T^{p-1}+\sigma^p K_p T^{\frac{p}{2}-1}) \int_0^t K^\prime\zeta(\tau) + U_B(u,T)^p z_\tau d\tau \\
        &\leq 2^{3p} z_0 + 2^{4p} (T^{p}+\sigma^p K_p T^{\frac{p}{2}}) K^\prime\sup_{\tau\in[0,T]}\zeta(\tau)  + 2^{4p} (T^{p-1}+\sigma^p K_p T^{\frac{p}{2}-1}) \int_0^t U_B(u,T)^p z_\tau d\tau .
    \end{align*}
    \normalsize

    Denote $\bar{K}:=2^{4p} (T^{p}+\sigma^p K_p T^{\frac{p}{2}})$, then
    \begin{equation*}
        z_t \leq 2^{3p} z_0 + \bar{K}K^\prime \sup_{t\in[0,T]} \zeta(t) + \bar{K} \int_0^t U_B(u,T)^p z_\tau d\tau.
    \end{equation*}

    Apply the Gronwall inequality (Lem.~\ref{lem.Gronwall}), we obtain that:
    \begin{equation*}
        z_t \leq (2^{3p} z_0 + \bar{K}K^\prime \sup_{t\in[0,T]} \zeta(t))\exp\{\bar{K}U_B(u,T)^p t\}.
    \end{equation*}

    Take $t=T$, we then have:
    \begin{equation*}
        \sup_{t\in[0,T]} \mathbb{E}\|m_t - \Tilde{m}_t\|^p_\mathrm{sp} \leq \{2^{3p} \mathbb{E}\|m_0 - \Tilde{m}_0\|^p_\mathrm{sp}  + \bar{K}K^\prime \sup_{t\in[0,T]} \zeta(t)\}\exp\{\bar{K}U_B(u,T)^p T\}.
    \end{equation*}

    Recall that we have shown in Lem.~\ref{lem.consistency-initial-value-estimate},
    \begin{equation*}
        \mathbb{E}\|\wh{m}_0-m_0\|^p_\mathrm{sp} = O(u^{6p} \mathbb{E}\|\Tilde{\Phi}-\Phi\|^p_\mathrm{sp}).
    \end{equation*}

    Furthermore, according to Lem.~\ref{lem.consistency-initial-value-estimate} and Lem.~\ref{coro.tech-lem-Uy-UA-UB-UC}, it is not hard to show that:
    \begin{equation*}
        \sup_{t\in[0,T]} \zeta(t) = O\left[\left\{u^{18p} e^{2puT} + u^{7p} \exp(2^p T^{p-1} \beta)\right\}\mathbb{E}\|\Phi-\Tilde{\Phi}\|_\mathrm{sp}^p\right],
    \end{equation*}
    where $\beta\leq \underline{K}u^{8p} e^{puT}$ for some $\underline{K}=\underline{K}(p,\Phi,\sigma)>0$. Besides, according to Coro.~\ref{coro.tech-lem-Uy-UA-UB-UC},
    \begin{equation*}
        U_B(u,T)^p = O(u^{8p} e^{puT}).
    \end{equation*}

    Combine the above, we have:
    \begin{align*}
        \sup_{s\in[0,t]} \mathbb{E}\|m_s - \Tilde{m}_s\|^p_\mathrm{sp} &= O\left[\left\{u^{18p} e^{2puT} + u^{7p} \exp(2^p T^{p-1} \beta + \bar{K}U_B(u,T)^p T)\right\} \mathbb{E}\|\Phi-\Tilde{\Phi}\|_\mathrm{sp}^p\right]\\
        &= O\left\{\left(u^{18p} e^{2puT} + u^{7p} e^{\gamma}\right) \mathbb{E}\|\Phi-\Tilde{\Phi}\|_\mathrm{sp}^p\right\},
    \end{align*}
    where $\gamma:=2^p T^{p-1} \beta + \bar{K}U_B(u,T)^p T \leq K u^{8p} e^{puT}$ for some $K=K(p,\Phi,\sigma,T)>0$.

    In this regard, we have $\sup_{s\in[0,t]} \mathbb{E}\|m_s - \Tilde{m}_s\|^p_\mathrm{sp}$ is dominated by
    \begin{equation*}
        \sup_{s\in[0,t]} \mathbb{E}\|m_s - \Tilde{m}_s\|^p_\mathrm{sp} = o\{\exp(e^{2puT}) \mathbb{E}\|\Phi-\bar{\Phi}\|_\mathrm{sp}^p\}.
    \end{equation*}
\end{proof}

\begin{lemma}
    For $u> 1, p\geq 1$, we have:
    \small
    \begin{align*}
        &\sup_{t\in [0,T]} \max\{\mathbb{E} \|A_t-\Tilde{A}_t\|^p_\mathrm{sp}, \mathbb{E} \|B_t-\Tilde{B}_t\|^p_\mathrm{sp}\}\\
        &\quad \leq 2^{3p} \|\Phi\|_\mathrm{sp}^{2p} \exp(2^p T^{p-1}\beta) \mathbb{E}(\|y_0-\wh{y}_0\|^p_\mathrm{sp} + \alpha T^p)  + 2^{3p} (1+u^p+\|\Phi\|_\mathrm{sp}^p)U_y(u,T)^p \mathbb{E} \|\Phi-\Tilde{\Phi} \|_\mathrm{sp}^p,\\
        &\sup_{t\in [0,T]} \mathbb{E} \|C_t-\Tilde{C}_t\|^p_\mathrm{sp}  \\
        &\quad \leq 2^p \sup_t \|y_t\|_\mathrm{sp}^p \mathbb{E}\|\Phi-\Tilde{\Phi}\|_\mathrm{sp}^p + 2^{2p} u^p \exp(2^p T^{p-1}\beta)  \mathbb{E}(\|y_0-\wh{y}_0\|^p_\mathrm{sp} +  \alpha T^p),
    \end{align*}
    \normalsize
    where
    \begin{align*}
        &\alpha := \{2^{3p+1}+2^{4p}\sup_t \|y_t\|^p_\mathrm{sp}(u^p + \|\Phi\|_\mathrm{sp}^p)\} U_y(u,T)^p \|\Phi-\Tilde{\Phi}\|^p_\mathrm{sp}\\
        &\beta := 2^{3p+1} \|\Phi\|^p_\mathrm{sp} + 2^{4p+1}\left\{\|\Phi\|^{2p}_\mathrm{sp} \sup_t\|y_t\|^p_\mathrm{sp} + u^{2p} U_y(u,T)^p\right\}.
    \end{align*}
    \label{lem.E|A-TildeA|-|B-TildeB|}
\end{lemma}

\begin{proof}
    Recall that $A_t = \Phi_{V\backslash C,C}-y_t \Phi_{C,V\backslash C}^\top \Phi_C$, then we have:
    \small
    \begin{align*}
        \sup_{t\in [0,T]} \mathbb{E} \|A_t-\Tilde{A}_t\|^p_\mathrm{sp} &=  \sup_{t\in [0,T]}  \mathbb{E} \left\|\left(\Phi_{V\backslash C,C}- \Tilde{\Phi}_{V\backslash C,C}\right) - \left(y_t \Phi_{C,V\backslash C}^\top {\Phi}_C -\tilde{y}_t \Tilde{\Phi}_{C,V\backslash C}^\top \Tilde{\Phi}_C\right)\right\|_\mathrm{sp}^p \\
        &\leq 2^p \sup_{t\in [0,T]} \left( \mathbb{E} \|\Phi- \Tilde{\Phi}\|_\mathrm{sp}^p + \mathbb{E}\|y_t \Phi_{C,V\backslash C}^\top {\Phi}_C -\tilde{y}_t \Tilde{\Phi}_{C,V\backslash C}^\top \Tilde{\Phi}_C \|_\mathrm{sp}^p\right).
    \end{align*}
    \normalsize

    We note that:
    \small
    \begin{align*}
        &\mathbb{E}\|y_t \Phi_{C,V\backslash C}^\top {\Phi}_C -\tilde{y}_t \Tilde{\Phi}_{C,V\backslash C}^\top \Tilde{\Phi}_C \|_\mathrm{sp}^p \\
        &\manyquad[5] = \mathbb{E} \|y_t \Phi_{C,V\backslash C}^\top {\Phi}_C - \Tilde{y}_t \Phi_{C,V\backslash C}^\top {\Phi}_C + \Tilde{y}_t \Phi_{C,V\backslash C}^\top {\Phi}_C -\tilde{y}_t \Tilde{\Phi}_{C,V\backslash C}^\top \Tilde{\Phi}_C \|_\mathrm{sp}^p \\
        &\manyquad[5] \leq 2^p \mathbb{E} \|y_t - \Tilde{y}_t \|^p_\mathrm{sp}  \|\Phi_{C,V\backslash C}^\top {\Phi}_C \|_\mathrm{sp}^p+ 2^p \mathbb{E} \|\Tilde{y}_t \|_\mathrm{sp}^p \| \Phi_{C,V\backslash C}^\top{\Phi}_C - \Tilde{\Phi}_{C,V\backslash C}^\top \Tilde{\Phi}_C \|_\mathrm{sp}^p \\
        &\manyquad[5] \leq 2^p \|\Phi\|_\mathrm{sp}^{2p} \mathbb{E} \|y_t - \Tilde{y}_t \|^p_\mathrm{sp} + 2^{2p} \mathbb{E} \|\Tilde{y}_t \|_\mathrm{sp}^p \|\Phi-\Tilde{\Phi} \|_\mathrm{sp}^p (\|\Phi\|_\mathrm{sp}^p+\|\Tilde{\Phi}\|_\mathrm{sp}^p)\\
        &\manyquad[5] \overset{(1)}{\leq} 2^p \|\Phi\|_\mathrm{sp}^{2p} \mathbb{E} \|y_t - \Tilde{y}_t \|^p_\mathrm{sp} + 2^{2p} \{(u-1)^p+\|\Phi\|_\mathrm{sp}^p\} \mathbb{E} \|\Tilde{y}_t \|_\mathrm{sp}^p \|\Phi-\Tilde{\Phi} \|_\mathrm{sp}^p \\
        &\manyquad[5] \overset{(2)}{\leq} 2^p \|\Phi\|_\mathrm{sp}^{2p} \mathbb{E} \|y_t - \Tilde{y}_t \|^p_\mathrm{sp} + 2^{2p} \{u^p+\|\Phi\|_\mathrm{sp}^p\} U_y(u,T)^p \mathbb{E} \|\Phi-\Tilde{\Phi} \|_\mathrm{sp}^p,
    \end{align*}
    \normalsize
    where ``(1)'' is due to line 4 of Alg.~\ref{alg:estimation-ou} and ``(2)'' is due to Lem.~\ref{coro.tech-lem-Uy-UA-UB-UC}.

    Furthermore, in Prop.~\ref{prop.up-bound-yt}, substitute $y_1=y, y_2=\Tilde{y}, a_1=\Phi_{V\backslash C}, a_2=\Tilde{\Phi}_{V\backslash C},  b_1=\Phi_{C,V\backslash C}^\top \Phi_{C,V\backslash C}, b_2=\Tilde{\Phi}_{C,V\backslash C}^\top \Tilde{\Phi}_{C,V\backslash C}$, we obtain:
    \begin{equation*}
        \sup_{t\in [0,T]} \mathbb{E}\|y_t-\Tilde{y}_t\|^p_\mathrm{sp} \leq 2^p\mathbb{E}(\|y_0-\wh{y}_0\|^p_\mathrm{sp} +  \alpha T^p) \exp(2^p T^{p-1}\beta),
    \end{equation*}
    where
    \begin{align*}
        &\alpha := \{2^{3p+1}+2^{4p}\sup_t \|y_t\|^p_\mathrm{sp}(u^p + \|\Phi\|_\mathrm{sp}^p)\} U_y(u,T)^p \|\Phi-\Tilde{\Phi}\|^p_\mathrm{sp}\\
        &\beta := 2^{3p+1} \|\Phi\|^p_\mathrm{sp} + 2^{4p+1}\left\{\|\Phi\|^{2p}_\mathrm{sp} \sup_t\|y_t\|^p_\mathrm{sp} + u^{2p} U_y(u,T)^p\right\}.
    \end{align*}

    It follows that:
    \small
    \begin{align*}
         &\mathbb{E}\|y_t \Phi_{C,V\backslash C}^\top {\Phi}_C -\tilde{y}_t \Tilde{\Phi}_{C,V\backslash C}^\top \Tilde{\Phi}_C \|_\mathrm{sp}^p \\
         &\manyquad[2] \leq 2^p \|\Phi\|_\mathrm{sp}^{2p} \mathbb{E} \|y_t - \Tilde{y}_t \|^p_\mathrm{sp} + 2^{2p} \{u^p+\|\Phi\|_\mathrm{sp}^p\} U_y(u,T)^p \mathbb{E} \|\Phi-\Tilde{\Phi} \|_\mathrm{sp}^p\\
         &\manyquad[2] \leq 2^{2p} \|\Phi\|_\mathrm{sp}^{2p} \exp(2^p T^{p-1}\beta) \mathbb{E}(\|y_0-\wh{y}_0\|^p_\mathrm{sp} + \alpha T^p)  + 2^{2p} (u^p+\|\Phi\|_\mathrm{sp}^p) U_y(u,T)^p \mathbb{E} \|\Phi-\Tilde{\Phi} \|_\mathrm{sp}^p
    \end{align*}
    \normalsize

    Hence, we have:
    \small
    \begin{align*}
        &\sup_{t\in [0,T]} \mathbb{E} \|A_t-\Tilde{A}_t\|^p_\mathrm{sp} \leq 2^p \sup_{t\in [0,T]} \left( \mathbb{E} \|\Phi- \Tilde{\Phi}\|_\mathrm{sp}^p + \mathbb{E}\|y_t \Phi_{C,V\backslash C}^\top {\Phi}_C -\tilde{y}_t \Tilde{\Phi}_{C,V\backslash C}^\top \Tilde{\Phi}_C \|_\mathrm{sp}^p\right)\\
        &\quad \leq 2^{3p} \|\Phi\|_\mathrm{sp}^{2p} \exp(2^p T^{p-1}\beta) \mathbb{E}(\|y_0-\wh{y}_0\|^p_\mathrm{sp} + \alpha T^p)  + 2^{3p} (1+u^p+\|\Phi\|_\mathrm{sp}^p)U_y(u,T)^p \mathbb{E} \|\Phi-\Tilde{\Phi} \|_\mathrm{sp}^p.
    \end{align*}
    \normalsize

    In a similar way, we can show that $\sup_{t\in [0,T]} \mathbb{E} \|A_t-\Tilde{A}_t\|^p_\mathrm{sp}$ enjoys the same bounds as $\sup_{t\in [0,T]} \mathbb{E} \|A_t-\Tilde{A}_t\|^p_\mathrm{sp}$, and that:
    \small
    \begin{equation*}
        \sup_{t\in [0,T]} \mathbb{E} \|C_t-\Tilde{C}_t\|^p_\mathrm{sp} \leq 2^p \sup_t \|y_t\|_\mathrm{sp}^p \mathbb{E}\|\Phi-\Tilde{\Phi}\|_\mathrm{sp}^p + 2^{2p} u^p \exp(2^p T^{p-1}\beta)  \mathbb{E}(\|y_0-\wh{y}_0\|^p_\mathrm{sp} +  \alpha T^p).
    \end{equation*}
    \small
\end{proof}

\begin{lemma}
    For $p\geq 2$, assume $\sup_{t} \mathbb{E}\| m_t - \Tilde{m}_t \|^p_\mathrm{sp}=o(1)$, then we have:
    \begin{equation*}
        \sup_{t\in[0,T]} \mathbb{E}\|\Tilde{m}_t - \hat{m}_t\|^p_\mathrm{sp} = o\{\delta^{\frac{p}{2}}\exp(e^{2puT})\}.
    \end{equation*}
    \label{lem.tildem-hatm-consistency}
\end{lemma}

\begin{proof}
    Recall that the Euler-Maruyama estimator $\hat{m}$ satisfies:
    \begin{equation*}
        \wh{m}_t = \int_0^{\underline{t}} \Tilde{A}_{\underline{s}} \xi_{\underline{s}}+\Tilde{B}_{\underline{s}}\wh{m}_{\underline{s}} ds + \int_0^{\underline{t}} \Tilde{C}_{\underline{s}} d\xi_s,
    \end{equation*}
    where $\underline{t}:=k\delta$ for $t\in [k\delta,(k+1)\delta)$.

    Denote $z_t := \sup_{s\in[0,t]} \mathbb{E}\|\Tilde{m}_s - \hat{m}_s\|^p_\mathrm{sp}$. We note $z$ is a monotone function. Therefore, it is integrable on $[0,T]$. In the following, we will use the Gronwall inequality (Lem.~\ref{lem.Gronwall}) to derive a bound for $z_T$. We first have: 
    \begin{equation}
        z_t \leq 4^p\sup_{s\in [0,t]} \{I_1(s)+I_2(s)+I_3(s)+I_4(s)\},
        \label{eq.zt<=I1+...+I4}
    \end{equation}
    where
    \small
    \begin{align*}
        &I_1(s) = \mathbb{E} \left\|\int_0^{\underline{s}} \Tilde{A}_\tau \xi_\tau - \Tilde{A}_{\underline{\tau}} \xi_{\underline{\tau}} + \Tilde{B}_\tau \Tilde{m}_\tau - \Tilde{B}_{\underline{\tau}} \wh{m}_{\underline{\tau}} d\tau  \right\|^p_\mathrm{sp}, \,\, I_2(s) = \mathbb{E} \left\| \int_0^{\underline{s}} \Tilde{C}_\tau - \Tilde{C}_{\underline{\tau}} d\xi_\tau\right\|^p_\mathrm{sp}, \\
        &I_3(s) = \mathbb{E} \left\| \int_{\underline{s}}^s \Tilde{A}_\tau \xi_\tau + \Tilde{B}_\tau \Tilde{m}_\tau d\tau \right\|^p_\mathrm{sp}, \,\,\, I_4(s) = \mathbb{E} \left\| \int_{\underline{s}}^s \Tilde{C}_\tau d \xi_\tau\right\|^p_\mathrm{sp}.
    \end{align*}
    \normalsize

    In the following, we analyze $I_1,I_2,I_3,I_4$, respectively. We first discuss the bounds of $I_3$ and $I_4$ since they are simpler. For $I_3$, we have:
    \small
    \begin{align*}
        I_3(s) &= \mathbb{E} \left\| \int_{\underline{s}}^s \Tilde{A}_\tau \xi_\tau + \Tilde{B}_\tau \Tilde{m}_\tau d\tau \right\|^p_\mathrm{sp} \leq \delta^{p-1} \mathbb{E} \int_{\underline{s}}^s \left\| \Tilde{A}_\tau \xi_\tau + \Tilde{B}_\tau \Tilde{m}_\tau \right\|^p_\mathrm{sp} d\tau \\
        &\leq \delta^{p-1} \mathbb{E} \int_{\underline{s}}^s \| \Tilde{A}_\tau\|^p_\mathrm{sp} \| \xi_\tau \|^p_\mathrm{sp} + \|\Tilde{B}_\tau\|^p_\mathrm{sp} \| \Tilde{m}_\tau \|^p_\mathrm{sp} d\tau \\
        &\leq  \delta^{p-1} \mathbb{E} \int_{\underline{s}}^s U_A(u,T)^p \| \xi_\tau \|^p_\mathrm{sp} + U_B(u,T)^p 2^p(\| {m}_\tau \|^p_\mathrm{sp}+\| m_\tau - \Tilde{m}_\tau \|^p_\mathrm{sp})d\tau\\
        &\leq \delta^{p-1} \int_{\underline{s}}^s U_A(u,T)^p \mathbb{E}\| \xi_\tau \|^p_\mathrm{sp} + 2^p U_B(u,T)^p (\mathbb{E}\| {m}_\tau \|^p_\mathrm{sp}+\mathbb{E}\| m_\tau - \Tilde{m}_\tau \|^p_\mathrm{sp})d\tau\\
        &\leq \delta^p \Bigg\{ U_A(u,T)^p \sup_{\tau\in[0,T]} \mathbb{E}\| \xi_\tau \|^p_\mathrm{sp} + 2^p U_B(u,T)^p \Big(\sup_\tau \mathbb{E}\| {m}_\tau \|^p_\mathrm{sp}+\sup_{\tau\in[0,T]} \mathbb{E}\| m_\tau - \Tilde{m}_\tau \|^p_\mathrm{sp}\Big)\Bigg\} \\
        &=: \delta^p K_3(p,\Phi,\sigma,u,T), \numberthis \label{eq.bound-I3}
    \end{align*}
    \normalsize
    where according to Thm.~\ref{thm.tech-lem-Lp-boundedness-of-solution}, Coro.~\ref{coro.lp-bounded}, and Lem.~\ref{lem.m-tildem-consistency}, $K_3$ is a finite positive number.

    For $I_4$, we have according to the Burkholder-Davis-Gundy (BDG) inequality (Lem.~\ref{lem.burkholder-david}),
    \small
    \begin{align*}
        I_4(s) &= \mathbb{E} \left\| \int_{\underline{s}}^s \Tilde{C}_\tau d \xi_\tau\right\|^p_\mathrm{sp} = \mathbb{E} \left\| \int_{\underline{s}}^s \Tilde{C}_\tau \lambda_C(\tau) d\tau + \int_{\underline{s}}^s \Tilde{C}_\tau \sigma dW_C(\tau) \right\|^p_\mathrm{sp} \\
        &\leq 2^p \mathbb{E} \left\| \int_{\underline{s}}^s \Tilde{C}_\tau \lambda_C(\tau) d\tau\right\|^p_\mathrm{sp} + 2^p \mathbb{E} \left\| \int_{\underline{s}}^s \Tilde{C}_\tau \sigma dW_C(\tau) \right\|^p_\mathrm{sp} \\
        &\leq 2^p \delta^{p-1} \mathbb{E} \int_{\underline{s}}^s \|\Tilde{C}_\tau \|^p_\mathrm{sp} \| \lambda_C(\tau) \|^p_\mathrm{sp} d\tau + 2^p \sigma^p  \delta^{\frac{p}{2}-1} \mathbb{E} \int_{\underline{s}}^s \|\Tilde{C}_\tau\|_\mathrm{sp}^p d\tau \\
        &\leq 2^p \delta^p U_C(u,T)^p \sup_{\tau\in [0,T]}\mathbb{E} \| \lambda_C(\tau) \|^p_\mathrm{sp} + 2^p \sigma^p  \delta^{\frac{p}{2}} U_C(u,T)^p \\
        &= \delta^{\frac{p}{2}} \Bigg\{2^p U_C(u,T)^p \Big(\sup_{\tau\in [0,T]}\mathbb{E} \| \lambda_C(\tau) \|^p_\mathrm{sp}+\sigma^p\Big) \Bigg\} =:  \delta^{\frac{p}{2}} K_4(p,\Phi,\sigma,u,T) \quad \text{for $\delta\leq 1$}, \numberthis \label{eq.bound-I4}
    \end{align*}
    \normalsize
    where according to Coro.~\ref{coro.lp-bounded}, $K_4$ is a finite positive number.

    We now discuss $I_1$ and $I_2$. For $I_1$, we have:
    \small
    \begin{align*}
        &I_1(s) = \mathbb{E} \left\|\int_0^{\underline{s}} \Tilde{A}_\tau \xi_\tau - \Tilde{A}_{\underline{\tau}} \xi_{\underline{\tau}} + \Tilde{B}_\tau \Tilde{m}_\tau - \Tilde{B}_{\underline{\tau}} \wh{m}_{\underline{\tau}} d\tau  \right\|^p_\mathrm{sp} \\
        &= \mathbb{E} \left\|\int_0^{\underline{s}} \Tilde{A}_\tau \xi_\tau - \Tilde{A}_\tau \xi_{\underline{\tau}} + \Tilde{A}_\tau \xi_{\underline{\tau}} - \Tilde{A}_{\underline{\tau}} \xi_{\underline{\tau}} + \Tilde{B}_\tau \Tilde{m}_\tau - \Tilde{B}_\tau \Tilde{m}_{\underline{\tau}} + \Tilde{B}_\tau \Tilde{m}_{\underline{\tau}} - \Tilde{B}_{\underline{\tau}} \Tilde{m}_{\underline{\tau}} + \Tilde{B}_{\underline{\tau}} \Tilde{m}_{\underline{\tau}} -\Tilde{B}_{\underline{\tau}} \wh{m}_{\underline{\tau}} d\tau  \right\|^p_\mathrm{sp} \\
        &\leq 2^{4p} T^{p-1} \mathbb{E} \int_0^{\underline{s}} \|\Tilde{A}_\tau\|^p_\mathrm{sp} \|\xi_\tau - \xi_{\underline{\tau}}\|^p_\mathrm{sp} + \|\Tilde{A}_\tau - \Tilde{A}_{\underline{\tau}}\|^p_\mathrm{sp} \|\xi_{\underline{\tau}}\|^p_\mathrm{sp} + \|\Tilde{B}_\tau\|^p_\mathrm{sp} \| \Tilde{m}_\tau - \Tilde{m}_{\underline{\tau}} \|^p_\mathrm{sp} \\
        &\manyquad[3] + \|\Tilde{B}_\tau  - \Tilde{B}_{\underline{\tau}} \|^p_\mathrm{sp} \|\Tilde{m}_{\underline{\tau}}\|^p_\mathrm{sp} + \|\Tilde{B}_{\underline{\tau}} \|^p_\mathrm{sp} \| \Tilde{m}_{\underline{\tau}} -\wh{m}_{\underline{\tau}} \|^p_\mathrm{sp} d\tau\\
        &\leq 2^{4p} T^{p-1} \mathbb{E} \int_0^{\underline{s}} U_A(u,T)^p \|\xi_\tau - \xi_{\underline{\tau}}\|^p_\mathrm{sp} + \|\Tilde{A}_\tau - \Tilde{A}_{\underline{\tau}}\|^p_\mathrm{sp} \|\xi_{\underline{\tau}}\|^p_\mathrm{sp} + U_B(u,T)^p  \| \Tilde{m}_\tau - \Tilde{m}_{\underline{\tau}} \|^p_\mathrm{sp} \\
        &\manyquad[3] + \|\Tilde{B}_\tau  - \Tilde{B}_{\underline{\tau}} \|^p_\mathrm{sp} \|\Tilde{m}_{\underline{\tau}}\|_\mathrm{sp}^p + U_B(u,T)^p \| \Tilde{m}_{\underline{\tau}} -\wh{m}_{\underline{\tau}} \|^p_\mathrm{sp} d\tau.
    \end{align*}
    \normalsize
    
    According to Coro.~\ref{lem.smooth}, there exists $K^\prime=K^\prime(p,\Phi,\sigma)$ such that:
    \begin{equation*}
        \mathbb{E} \|\xi_\tau - \xi_{\underline{\tau}}\|^p_\mathrm{sp} \leq K^\prime \delta^{\frac{p}{2}}.
    \end{equation*}

    Due to Thm.~\ref{thm.tech-lem-Lp-boundedness-of-solution}, we have $ \sup_\tau \mathbb{E}\|\xi_\tau \|^p_\mathrm{sp} < \infty$. According to Lem.~\ref{lem.smooth-of-tilde-processes}, we have:
    \begin{equation*}
        \max\{\|\Tilde{A}_\tau - \Tilde{A}_{\underline{\tau}}\|_\mathrm{sp}^p, \|\Tilde{B}_\tau - \Tilde{B}_{\underline{\tau}}\|_\mathrm{sp}^p\} \leq \delta^p  4^p u^{2p} \left\{1 + 2u^p U_y(u,T)^p + u^{2p} U_y(u,T)^{2p}\right\}
    \end{equation*}
    
    Moreover, repeat to the analyses of $I_3$ and $I_4$, with $s,\underline{s}$ replaced with $\tau,\underline{\tau}$, we have:
    \begin{equation*}
        \mathbb{E} \|\Tilde{m}_{{\tau}} -\Tilde{m}_{\underline{\tau}} \|^p_\mathrm{sp} \leq \delta^p K_3 + \delta^{\frac{p}{2}} K_4.
    \end{equation*}
    
    Further, we note that:
    \begin{equation*}
        \mathbb{E}\|\Tilde{m}_{\underline{\tau}}\|_\mathrm{sp}^p \leq 2^p\left(\sup_\tau \mathbb{E}\|{m}_{\underline{\tau}}\|_\mathrm{sp}^p + \sup_\tau \mathbb{E}\|{m}_{\underline{\tau}} - \Tilde{m}_{\underline{\tau}}\|_\mathrm{sp}^p\right).
    \end{equation*}

    Now, denote $\underline{K}:=\sup_\tau \mathbb{E}\|\xi_\tau \|^p_\mathrm{sp} + \sup_{\tau}\mathbb{E}\| {m}_\tau \|^p_\mathrm{sp}+ \sup_{\tau} \mathbb{E}\| m_\tau - \Tilde{m}_\tau \|^p_\mathrm{sp}$. Combining the above analyses, we have, for $\delta\geq 1$,
    \small
    \begin{align*}
        I_1(s) &\leq 2^{4p} T^{p-1} \int_0^{\underline{s}} U_A(u,T)^p K^\prime \delta^{\frac{p}{2}} + U_B(u,T)^p (\delta^p K_3 + \delta^{\frac{p}{2}} K_4) \\
        &\manyquad[3] +\delta^p  2^{3p} u^{2p} \left\{1 + 2u^p U_y(u,T)^p + u^{2p} U_y(u,T)^{2p}\right\} \underline{K} + U_B(u,T)^p z_\tau d\tau \\
        &\leq \delta^{\frac{p}{2}}  2^{4p} T^{p} \Big[K^\prime  U_A(u,T)^p + U_B(u,T)^p (K_3 + K_4) \\
        &\manyquad[3] + 2^{3p} u^{2p} \left\{1 + 2u^p U_y(u,T)^p + u^{2p} U_y(u,T)^{2p}\right\} \underline{K}\Big] + 2^{4p} T^{p-1} U_B(u,T)^p \int_0^s z_\tau d\tau \\
        &=:\delta^{\frac{p}{2}} K_1 + 2^{4p} T^{p-1} U_B(u,T)^p \int_0^s z_\tau d\tau, \numberthis \label{eq.bound-I1}
    \end{align*}
    \normalsize
    where according to Lem.~\ref{lem.m-tildem-consistency}, $K_1=K_1(p,\Phi,\sigma,u,T)$ is a finite positive number.

    Finally, for $I_2$, we have, according to the BDG inequality (Lem.~\ref{lem.burkholder-david}),    
    \small
    \begin{align*}
        I_2(s) &= \mathbb{E} \left\| \int_0^{\underline{s}} (\Tilde{C}_\tau - \Tilde{C}_{\underline{\tau}}) d\xi_\tau\right\|^p_\mathrm{sp} = \mathbb{E} \left\| \int_0^{\underline{s}} (\Tilde{C}_\tau - \Tilde{C}_{\underline{\tau}})\lambda_C d\tau +  \int_0^{\underline{s}} (\Tilde{C}_\tau - \Tilde{C}_{\underline{\tau}}) \sigma dW_C(\tau)\right\|^p_\mathrm{sp} \\
        &\leq 2^p \mathbb{E} \left\| \int_0^{\underline{s}} (\Tilde{C}_\tau - \Tilde{C}_{\underline{\tau}})\lambda_C d\tau\right\|^p_\mathrm{sp}  + 2^p \mathbb{E} \left\|  \int_0^{\underline{s}} (\Tilde{C}_\tau - \Tilde{C}_{\underline{\tau}}) \sigma dW_C(\tau)\right\|^p_\mathrm{sp} \\
        &\leq 2^p T^{p-1} \mathbb{E} \int_0^{\underline{s}} \|\Tilde{C}_\tau - \Tilde{C}_{\underline{\tau}}\|^p_\mathrm{sp} \|\lambda_C(\tau)\|^p_\mathrm{sp} d\tau + 2^p T^{\frac{p}{2}-1} \mathbb{E} \int_0^{\underline{s}} \|\Tilde{C}_\tau - \Tilde{C}_{\underline{\tau}}\|^p_\mathrm{sp} \sigma^p du
    \end{align*}
    \normalsize

    According to Lem.~\ref{lem.smooth-of-tilde-processes}, we have:
    \begin{equation*}
        \|\Tilde{C}_\tau - \Tilde{C}_{\underline{\tau}}\|_\mathrm{sp}^p \leq 4^p \delta^p u^{p} \left\{1 + 2u^p U_y(u,T)^p + u^{2p} U_y(u,T)^{2p}\right\}.
    \end{equation*}

    According to Coro.~\ref{coro.lp-bounded}, we have $\sup_\tau \mathbb{E} \|\lambda_C(\tau)\|^p_\mathrm{sp} < \infty$. Therefore, we have:
    \begin{align*}
        I_2(s) &\leq \left\{2^p T^{\frac{p}{2}} \sigma^p + 2^p T^p \sup_t \mathbb{E}\|\lambda_C(t)\|_\mathrm{sp}^p\right\} \sup_\tau \|\Tilde{C}_\tau - \Tilde{C}_{\underline{\tau}}\|_\mathrm{sp}^p\\
        &\leq \delta^p 2^{3p} u^{p} \left\{T^{\frac{p}{2}} \sigma^p + T^p \sup_t \mathbb{E}\|\lambda_C(t)\|_\mathrm{sp}^p\right\}  \left\{1 + 2u^p U_y(u,T)^p + u^{2p} U_y(u,T)^{2p}\right\} \\
        &:= \delta^p K_2(p,\Phi,\sigma,u,T). \numberthis \label{eq.bound-I2}
    \end{align*}
    where according to Lem.~\ref{lem.m-tildem-consistency}, $K_1=K_1(p,\Phi,\sigma,u,T)$ is a finite positive number.

    Combine \eqref{eq.zt<=I1+...+I4}-\eqref{eq.bound-I2}, we have for $\delta\leq 1$,
    \begin{align*}
        z_t &\leq 4^p\sup_{s\in [0,t]} \{I_1(s)+I_2(s)+I_3(s)+I_4(s)\}\\
        &\leq 4^p\sup_{s\in [0,t]} \{\delta^{\frac{p}{2}} K_1 + 2^{4p} T^{p-1} U_B(u,T)^p \int_0^s z_\tau d\tau + \delta^p K_2 + \delta^p K_3 + \delta^{\frac{p}{2}} K_4\}\\
        &\leq 4^p (K_1+K_2+K_3+K_4) \delta^{\frac{p}{2}} + 2^{6p} T^{p-1} U_B(u,T)^p \int_0^s z_\tau d\tau.
    \end{align*}

    Then, apply the Gronwall inequality (Lem.~\ref{lem.Gronwall}), we have:
    \begin{equation*}
        z_t \leq 4^p (K_1+K_2+K_3+K_4) \delta^{\frac{p}{2}} \exp\{2^{6p} T^{p-1} U_B(u,T)^p t\}.
    \end{equation*}

    Take $t=T$, we then have:
    \begin{equation*}
        \sup_{t\in[0,T]} \mathbb{E}\|\Tilde{m}_t - \hat{m}_t\|^p_\mathrm{sp} \leq 4^p (K_1+K_2+K_3+K_4) \delta^{\frac{p}{2}} \exp\{2^{6p} T^p U_B(u,T)^p\}.
    \end{equation*}

    According to Prop.~\ref{thm.tech-lem-Lp-boundedness-of-solution}, Coro.~\ref{coro.tech-lem-Uy-UA-UB-UC}, and that we assume $\sup_{t} \mathbb{E}\| m_t - \Tilde{m}_t \|^p_\mathrm{sp}=O(1)$, we have:
    \begin{align*}
        &K_3 = O\{U_A(u,T)^p + U_B(u,T)^p\} = O(u^{8p} e^{puT}), \\
        &K_4 = O\{U_C(u,T)^p\} = O(u^{7p} e^{puT}),\\ 
        &K_1 = O\{U_A(u,T)^p + U_B(u,T)^p (K_3+K_4) + u^{4p} U_y(u,T)^{2p}\} = O(u^{16p} e^{2puT}), \\
        &K_2 = O\{u^{2p} U_y(u,T)^p\} = O(u^{8p} e^{puT}).
    \end{align*}

    Therefore, we have:
    \begin{equation*}
        \sup_{t\in[0,T]} \mathbb{E}\|\Tilde{m}_t - \hat{m}_t\|^p_\mathrm{sp} = O(u^{16p} e^{2puT}\delta^{\frac{p}{2}} e^\gamma),
    \end{equation*}
    where $\gamma=2^{6p} T^p U_B(u,T)^p\leq K u^{8p} e^{puT}$ for some positive constant $K$. In this regard, we have $\sup_{t\in[0,T]} \mathbb{E}\|\Tilde{m}_t - \hat{m}_t\|^p_\mathrm{sp}$ is dominated by
    \begin{equation*}
        \sup_{t\in[0,T]} \mathbb{E}\|\Tilde{m}_t - \hat{m}_t\|^p_\mathrm{sp} = o\{\delta^{\frac{p}{2}}\exp(e^{2puT})\}.
    \end{equation*}
\end{proof}

\begin{lemma}
    For $p\geq 2$ and any $s,t\in [0,T]$, we have:
    \begin{align*}
        \max\{&\|\Tilde{A}_s - \Tilde{A}_t\|_\mathrm{sp}^p,\|\Tilde{B}_s - \Tilde{B}_t\|_\mathrm{sp}^p\} \leq 4^p (t-s)^p u^{2p} \left\{1 + 2u^p U_y(u,T)^p + u^{2p} U_y(u,T)^{2p}\right\}, \\
        &\|\Tilde{C}_s - \Tilde{C}_t\|_\mathrm{sp}^p \leq 4^p (t-s)^p u^{p} \left\{1 + 2u^p U_y(u,T)^p + u^{2p} U_y(u,T)^{2p}\right\}.
    \end{align*}
    \label{lem.smooth-of-tilde-processes}
\end{lemma}

\begin{proof}
    Recall $\Tilde{A}_t = \Tilde{\Phi}_{V\backslash C,C}-\Tilde{y}_t \Tilde{\Phi}_{C,V\backslash C}^\top \Tilde{\Phi}_C$, we then have:
    \begin{align*}
        \|\Tilde{A}_s - \Tilde{A}_t\|_\mathrm{sp}^p &= \|\Tilde{y}_s \Tilde{\Phi}_{C,V\backslash C}^\top \Tilde{\Phi}_C - \Tilde{y}_t \Tilde{\Phi}_{C,V\backslash C}^\top \Tilde{\Phi}_C\|_\mathrm{sp}^p \leq \|\Tilde{y}_s -\Tilde{y}_t\|_\mathrm{sp}^p \|\Tilde{\Phi}\|_\mathrm{sp}^{2p} \leq u^{2p}\|\Tilde{y}_s -\Tilde{y}_t\|_\mathrm{sp}^p.
    \end{align*}

    Furthermore, recall $\Tilde{y}$ is the solution of
    \begin{equation*}
        \dot{\Tilde{y}}_t = \Tilde{\Phi}_{V\backslash C} \Tilde{y}_t + \Tilde{y}_t \Tilde{\Phi}_{V\backslash C}^\top + I - \Tilde{y}_t \Tilde{\Phi}_{C, V\backslash C}^\top \Tilde{\Phi}_{C, V\backslash C} \Tilde{y}_t,
    \end{equation*}
    hence,
    \begin{align*}
        \|\Tilde{y}_s -\Tilde{y}_t\|_\mathrm{sp}^p &= \left\|\int_s^t \Tilde{\Phi}_{V\backslash C} \Tilde{y}_\tau + \Tilde{y}_\tau \Tilde{\Phi}_{V\backslash C}^\top + I - \Tilde{y}_\tau \Tilde{\Phi}_{C, V\backslash C}^\top \Tilde{\Phi}_{C, V\backslash C} \Tilde{y}_\tau d\tau \right\|_\mathrm{sp}^p \\
        &\leq 4^p (t-s)^{p-1} \int_s^t 2\|\Tilde{\Phi}_{V\backslash C}\|_\mathrm{sp}^p \| \Tilde{y}_\tau\|_\mathrm{sp}^p + 1 + \|\Tilde{y}_\tau\|_\mathrm{sp}^{2p} \|\Tilde{\Phi}_{C, V\backslash C}\|_\mathrm{sp}^{2p} d\tau \\
        &\leq 4^p (t-s)^p \left\{1 + 2u^p U_y(u,T)^p + u^{2p} U_y(u,T)^{2p}\right\}.
    \end{align*}

    Therefore,
    \begin{equation*}
        \|\Tilde{A}_s - \Tilde{A}_t\|_\mathrm{sp}^p \leq 4^p (t-s)^p u^{2p} \left\{1 + 2u^p U_y(u,T)^p + u^{2p} U_y(u,T)^{2p}\right\}.
    \end{equation*}

    In a similar way, we can show $\|\Tilde{B}_s - \Tilde{B}_t\|_\mathrm{sp}^p$ has the same bound, and that:
    \begin{equation*}
         \|\Tilde{C}_s - \Tilde{C}_t\|_\mathrm{sp}^p \leq 4^p (t-s)^p u^{p} \left\{1 + 2u^p U_y(u,T)^p + u^{2p} U_y(u,T)^{2p}\right\}.
    \end{equation*}
\end{proof}

\noindent\textbf{Proposition~\ref{prop.consistency-projection-estimator}.}\emph{Assume Asm.~\ref{asm.identifiable-ou}, then for $q>2$, we have:
    \begin{equation*}
        \max \left\{\|\mu-\hat{\mu}\|_q,\, \|G-\hat{G}\|_q, \, N^{\frac{1}{2}}\|\mu-\hat{\mu}\|_2 \|G-\hat{G}\|_2\right\} \to 0
    \end{equation*}
    as $N_c\to\infty, \delta_c=\Theta(N_c^{-\frac{1}{6}}), u = O(\frac{1}{4T} \log\log N_c^{\frac{1}{6}}), \delta=O(N_c^{-\frac{2}{3}})$.}

\begin{proof}
    Recall Prop.~\ref{prop.consistency-ou} (Eq.~\eqref{eq.order-barPhi-Phi}) shows that for any $p\geq 1$,
    \begin{equation*}
        \mathbb{E} \|\Phi - \bar{\Phi}\|^p_\mathrm{sp} = O\left(\delta_c^{-p} N_c^{-\frac{p}{2}} + \delta_c^{2p}\right).
    \end{equation*}

    Hence, under $N_c\to\infty, \delta_c=\Theta(N_c^{-\frac{1}{6}})$, we have:
    \begin{align*}
        N_c^{\frac{1}{2}} \mathbb{E} \|\Phi - \bar{\Phi}\|^2_\mathrm{sp} = O(N_c^{-\frac{1}{6}}) \leq K N_c^{-\frac{1}{6}}, \quad \mathbb{E} \|\Phi - \bar{\Phi}\|^q_\mathrm{sp} = O(N_c^{-\frac{1}{3}q}) \leq K N_c^{-\frac{1}{3}q}
    \end{align*}
    for some positive constant $K$. Further, recall Lem.~\ref{lem.m-tildem-consistency} shows that for any $p\geq 2$,
    \begin{equation*}
        \sup_{t\in[0,T]} \mathbb{E}\|m_t - \Tilde{m}_t\|^p_\mathrm{sp} = o\{\exp(e^{2puT}) \mathbb{E}\|\Phi-\bar{\Phi}\|_\mathrm{sp}^p\}.
    \end{equation*}

    Thus, under $u = O(\frac{1}{4T} \log\log N_c^{\frac{1}{6}})$, which slower than $O(\frac{1}{2qT} \log\log N_c^{\frac{1}{3}q})$, we have:
    \begin{align*}
        &N_c^{\frac{1}{2}}\sup_{t\in[0,T]} \mathbb{E}\|m_t - \Tilde{m}_t\|^2_\mathrm{sp} = o\{\exp(e^{4uT}) K N_c^{-\frac{1}{6}}\} = o(K) \to 0, \numberthis\label{eq.o(K)1}\\
        &\sup_{t\in[0,T]} \mathbb{E}\|m_t - \Tilde{m}_t\|^q_\mathrm{sp} = o\{\exp(e^{2quT}) K N_c^{-\frac{1}{3}q}\} = o(K)\to 0.\numberthis\label{eq.o(K)2}
    \end{align*}

    Moreover, recall Lem.~\ref{lem.tildem-hatm-consistency} shows that for any $p\geq 2$,
    \begin{equation*}
        \sup_{t\in[0,T]} \mathbb{E}\|\Tilde{m}_t - \hat{m}_t\|^p_\mathrm{sp} = o\{\delta^{\frac{p}{2}}\exp(e^{2puT})\}.
    \end{equation*}

    Therefore, under $\delta=O(N_c^{-\frac{2}{3}})\leq K^\prime N_c^{-\frac{2}{3}}$ for some $K^\prime>0$, we have:
    \begin{align*}
        N_c^{\frac{1}{2}}\sup_{t\in[0,T]} \mathbb{E}\|\Tilde{m}_t - \hat{m}_t\|^2_\mathrm{sp} &= o\{\delta\exp(e^{4uT})\} \\
        &= o\{\exp(e^{4uT}) K^\prime N_c^{-\frac{2}{3}}\} = o\{\exp(e^{4uT}) K^\prime N_c^{-\frac{1}{6}}\} \overset{(1)}{=} o(K^\prime) \to 0\\
        \sup_{t\in[0,T]} \mathbb{E}\|\Tilde{m}_t - \hat{m}_t\|^q_\mathrm{sp} &= o\{\delta^{\frac{q}{2}}\exp(e^{2quT})\} = o\{\exp(e^{2quT}) K^\prime N_c^{-\frac{1}{3}q}\} \overset{(2)}{=} o(K^\prime)\to 0,
    \end{align*}
    where ``(1)'' and ``(2)'' have been shown in \eqref{eq.o(K)1} and \eqref{eq.o(K)2}, respectively.

    Combining the above, we have:
    \begin{equation*}
        \max\Big\{N_c^{\frac{1}{2}}\sup_{t\in[0,T]} \mathbb{E}\|m_t - \hat{m}_t\|^2_\mathrm{sp}, \sup_{t\in[0,T]} \mathbb{E}\|m_t - \hat{m}_t\|^p_\mathrm{sp}\Big\} \to 0
    \end{equation*}
    under the specified asymptotics. Since for any $p\geq 1$,
    \begin{align*}
        \|G-\hat{G}\|_p = \|\Pi-\hat{\Pi}\|_p &= \Big(\mathbb{E}\int_0^T |\Pi_t-\hat{\Pi}_t|^p dt\Big)^{\frac{1}{p}} \\
        &\leq  \Big(T\sup_{t\in [0,T]} \mathbb{E}|\Pi_t-\hat{\Pi}_t|^p \Big)^{\frac{1}{p}} \leq \Big(T\sup_{t\in [0,T]} \mathbb{E}\|m_t-\hat{m}_t\|^p_\mathrm{sp} \Big)^{\frac{1}{p}},
    \end{align*}
    we conclude that:
    \begin{equation*}
        \max \left\{\|G-\hat{G}\|_q, \, N^{\frac{1}{4}}\|G-\hat{G}\|_2\right\} \to 0.
    \end{equation*}

    Besides, recall $\hat{\mu}_t = \Tilde{\Phi}_{\beta,V\backslash C} \hat{m}_t + \Tilde{\Phi}_{\beta,C} X_C(t)$. Hence, we have:
    \small
    \begin{align*}
        \mathbb{E}|\mu_t -\hat{\mu}_t|^p &= \mathbb{E}\Big|(\Phi_{\beta,V\backslash C} m_t - \Tilde{\Phi}_{\beta,V\backslash C} m_t)  + (\Tilde{\Phi}_{\beta,V\backslash C} m_t-\tilde{\Phi}_{\beta,V\backslash C} \hat{m}_t) + (\Phi_{\beta,C}-\Tilde{\Phi}_{\beta,C}) \Big|^p \\
        &\leq \mathbb{E}\|\Phi-\Tilde{\Phi}\|_\mathrm{sp}^p \big(\sup_{t\in [0,T]} \mathbb{E}\|m_t\|_\mathrm{sp}^p+\sup_{t\in [0,T]} \mathbb{E}\|X_C(t)\|_\mathrm{sp}^p\big) + (u+1)^p \sup_{t\in [0,T]} \mathbb{E}\|m_t-\hat{m}_t\|_\mathrm{sp}^p.
    \end{align*}
    \normalsize

    Then, apply Coro.~\ref{coro.lp-bounded}, Prop.~\ref{prop.consistency-ou}, Lem.~\ref{lem.m-tildem-consistency}, and Lem.~\ref{lem.tildem-hatm-consistency}, it is not hard to verify that:
    \begin{equation*}
        \sup_{t\in [0,T]} \mathbb{E}|\mu_t -\hat{\mu}_t|^p = o\{\exp(e^{2puT}) \mathbb{E}\|\Phi-\bar{\Phi}\|_\mathrm{sp}^p+\delta^{\frac{p}{2}}\exp(e^{2puT})\}.\
    \end{equation*}
    has the same bound as $\sup_t \mathbb{E}\|m_t-\hat{m}_t\|_\mathrm{sp}^p$. So, repeat the analysis above, we have:
    \begin{equation*}
        \max \left\{\|\mu-\hat{\mu}\|_q, \, N^{\frac{1}{4}}\|\mu-\hat{\mu}\|_2\right\} \to 0.
    \end{equation*}

    This concludes the proof.
\end{proof}
\section{Technical lemmas}
\label{appx.tech-lemmas}

\subsection{Lemmas used in Appx.~\ref{appx.proof-sec-setup}}

\begin{proposition}
    Consider the Itô process $(X_t)$. Assume \eqref{eq.regularity-lambda-bounded-inL2}, then the processes $X_t, \mu_t, \Pi_t$, and $M_t$ are all bounded in $L_2$.
    \label{prop.tech-lem-L2-boundedness-under-(4)}
\end{proposition}

\begin{proof}
    For any $t\in [0,T]$, we have:
    \begin{align*}
        \mathbb{E}\|X_t\|^2 &\leq 3\mathbb{E}\|X_0\|^2 + 3\mathbb{E} \left\| \int_0^t \lambda_s ds \right\|^2 + 3d^2 \sum_{i,j=1}^d \mathbb{E}\left\{ \Sigma_{i j} W_{j}(t) \right\}^2 \\
        &\overset{(1)}{\leq} 3\mathbb{E}\|X_0\|^2 + 3T\mathbb{E} \int_0^t \|\lambda_s\|^2 ds + 3d^2 T \|\Sigma\|^2 \\
        &\overset{(2)}{\leq} 3\mathbb{E}\|X_0\|^2 + 3T^2 \sup_{s\in [0,T]}\mathbb{E}\|\lambda_s\|^2 + 3d^2 T \|\Sigma\|^2 <\infty,
    \end{align*}
    where ``(1)" applies the Cauchy-Schwarz inequality and ``(2)" applies the Fubini's theorem. 

    Furthermore, apply the Jensen inequality,
    \begin{align*}
        \mathbb{E}|\mu_t|^2 = \mathbb{E}|\mathbb{E}\{\lambda_\beta(t)|\mathcal{F}_t\}|^2 \leq \mathbb{E}\mathbb{E}\{|\lambda_\beta(t)|^2 |\mathcal{F}_t\} =  \mathbb{E}|\lambda_\beta(t)|^2 \leq \sup_{t\in [0,T]} \mathbb{E}\|\lambda_s\|^2 < \infty.
    \end{align*}

    In a similar way, the process $\Pi_t = \mathbb{E}\{X_\alpha(t)|\mathcal{F}_t\}$ is also bounded in $L_2$.
    
    Finally, we note that:
    \small
    \begin{align*}
        \mathbb{E}|M_t|^2 = \mathbb{E}\left|X_\beta(t) - X_\beta(0) - \int_0^t \mu_s ds\right|^2 &\leq 3\mathbb{E}|X_\beta(t)|^2 + 3\mathbb{E}|X_\beta(0)|^2 + 3T\int_0^t \mathbb{E}|\mu_s|^2 ds \\
        &\leq 3\mathbb{E}\|X_t\|^2 + 3\mathbb{E}\|X_0\|^2 + 3T^2 \sup_{t\in [0,T]} \mathbb{E}|\mu_t|^2 < \infty.
    \end{align*}
    \normalsize
\end{proof}

\begin{lemma}
    Consider a random process $f(\omega,t)$, defined on the probability space $(\Omega,\mathcal{F},P)$, if $f$ is progressively measurable, then
    \begin{equation*}
        \mathbb{E} \int_0^t f(s)ds = \int_0^t \mathbb{E}f(s) ds
    \end{equation*}
    for every $t$, provided that one (then the other) of the integrals of the absolute value of $f$ exists.
    \label{lem.fubini-progressive-measurable}
\end{lemma}

\begin{proof}
    Recall that Fubini's theorem states that the integrals are exchangeable if $f$, treated as a mapping from $\Omega\times [0,t]\mapsto \mathbb{R}$, is measurable to $\mathcal{F}\times\mathcal{B}([0,t])$ for every $t$. This is exactly the definition of $f$ being progressively measurable.
\end{proof}

\begin{definition}[Quadratic variation]
    Consider a real-valued stochastic process $X_t$. Its quadratic variation, denoted as $[X]_t$, is the process defined as:
    \begin{equation*}
        [X]_t = \lim_{|\mathcal{P}|\to 0} \sum_{i=1}^n (X_{t_i}-X_{t_{i-1}})^2,
    \end{equation*}
    where $\mathcal{P}$ ranges over partitions of the interval $[0,t]$ and the norm of $\mathcal{P}$ is the mesh. This limit, if exists, is defined using convergence in probability.

    More generally, the covariation of two processes $X$ and $Y$ is:
    \begin{equation*}
        [X,Y]_t = \lim_{|\mathcal{P}|\to 0} \sum_{i=1}^n (X_{t_i}-X_{t_{i-1}})(Y_{t_i}-Y_{t_{i-1}}).
    \end{equation*}
    \label{def.tech-lem-quadratic-variation}
\end{definition}

\begin{theorem}[\cite{kuo2006introduction}, Thm. 6.5.8 and Thm. 6.5.9, p. 88]
    Let $(M_t)$ be a cadlag martingale bounded in $L_2$, and $X_t$ be a predictable process satisfying:
    \begin{equation*}
        \mathbb{E}\int_0^T |X_t|^2 d[M]_t < \infty.
    \end{equation*}
    Then the stochastic process $I_t = \int_0^t X_s dM_s, t\in [0,T]$ is a martingale with
    \begin{equation*}
        \mathbb{E}|I_t|^2 = \mathbb{E}\int_0^t |X_s|^2 d[M]_s, \quad [I]_t = \int_0^t |X_s|^2 d[M]_s
    \end{equation*}
    \label{thm.tech-lem-integral-process-is-martingale}
\end{theorem}

\begin{lemma}
    Consider the process $M_t$, under \eqref{eq.regularity-lambda-bounded-inL2}, we have the quadratic variation process $[M]_t=\|\Sigma_\beta\|^2_2 t$, where $\Sigma_\beta:=[\Sigma_{\beta 1}, ..., \Sigma_{\beta d}]^\top$.
    \label{lem.tech-lem-[M]_t}
\end{lemma}

\begin{proof}
    Recall that $M_t=A_t+B_t$, where
   \begin{equation*}
       A_t = \int_0^t \lambda_\beta(s) - \mathbb{E}\{\lambda_\beta(s)|\mathcal{F}_s\} ds, \quad B_t = \sum_{j=1}^d \Sigma_{\beta j} W_j(t).
   \end{equation*}

   By Def.~\ref{def.tech-lem-quadratic-variation}, we have $[M]_t=[A]_t+2[A,B]_t+[B]_t$. It is clear that $[B]_t= \| \Sigma_\beta \|^2_2 t$. In the following, we will show $[A]_t=0$ and $[A,B]_t=0$. 

   Denote $A_t=\int_0^t f(s,\omega)ds$. According to Prop.~\ref{prop.tech-lem-L2-boundedness-under-(4)}, we have $f(\cdot,\omega)\in L_2[0,T]$ almost surely. Therefore, $t\mapsto A_t$ is uniformly continuous on $[0,T]$ almost surely. If $\mathcal{P}=\{0=t_0<\cdots<t_n=T\}$ is a partition of $[0,T]$ with mesh $|\mathcal{P}|:=\max_{i} |t_{i+1}-t_i|$, then
   \begin{align*}
        \sum_i |A_{t_{i+1}}-A_{t_i}|^2 &\leq \sup_{|s-t|<|\mathcal{P}|} |A_s-A_t| \sum_i |A_{t_{i+1}}-A_{t_i}| \\
        &\leq \sup_{|s-t|<|\mathcal{P}|} |A_s-A_t| \int_0^T |g(s)|ds.
    \end{align*}
    Due to the uniform continuity of $A_t$, the RHS converges to $0$ almost surely as $|\mathcal{P}|\to 0$. This gives us $[A]_t=0$. In a similar way, we can show $[A,B]_t=0$. 
\end{proof}

\subsection{Lemmas used in Appx.~\ref{appx.proof-asym-ana-ii} and Appx.~\ref{appx.proof-thm.asym}}

\begin{definition}[Stochastic equicontinuity]
    A sequence of random processes $X^{(n)}$ indexed over a metric space $(\mathcal{T},d)$ is called stochastic equicontinuous if for all $\eta,\epsilon>0$, there exists $\delta>0$ such that:
    \begin{equation*}
        \limsup_{n\to \infty} P\left\{\sup_{s,t\in\mathcal{T}: d(s,t)\leq \delta} |X_s-X_t|>\eta\right\} <\epsilon.
    \end{equation*}
    \label{def.stoch-equicontinuous}
\end{definition}

\begin{theorem}[\cite{newey1991uniform}, Thm. 2.1]
    Let $X^{(n)}$ be a sequence of random processes indexed by a compact metric space $\mathcal{T}$. Assume $X^{(n)}$ is stochastic equicontinuous and that $X_t^{(n)}\overset{\mathcal{P}}{\to} 0$ for every $t\in\mathcal{T}$. Then $\sup_{t\in\mathcal{T}} |X_t^{(n)}| \overset{\mathcal{P}}{\to} 0$ as $n\to\infty$.
    \label{thm.pointwise-convergence+stoch-equicon=uniform-convergence}
\end{theorem}

\begin{definition}[Orlicz norms; \cite{wellner2013weak}, Sec. 2.2.1, p. 144]
    A Young function is a convex function $\psi:\mathbb{R}^+\mapsto \mathbb{R}^+$ with $\psi(0)=0$ that is not identically zero. For a given Young function and a random variable, the Orlicz norm $\|X\|_\psi$ is defined as:
    \begin{equation*}
        \|X\|_\psi:= \inf\left\{C>0: \mathbb{E}\psi(|X|/C)\leq 1\right\}.
    \end{equation*}
    When $\psi(x)=x^p$ for $p\geq 1$, $\|X\|_\psi=(\mathbb{E} |X|^p)^{1/p}$ equals the $L_p$-norm. 
\end{definition}

Let $D(\epsilon,d)$ denote the packing number, i.e., the maximum number of $\epsilon$-separated points in the space $\mathcal{T}$ under the semimetric $d$. For an invertible function $\psi$, let $\psi^{-1}$ denote its inverse function.

\begin{theorem}[\cite{wellner2013weak}; Thm. 2.2.4, p. 147]
    Let $\psi$ be a convex, nondecreasing, nonzero function with $\psi(0)=0$ and $\limsup_{x,y\to\infty} \frac{\psi(x)\psi(y)}{\psi(cxy)}<\infty$ for some constant $c$. Let $\{X_t: t\in\mathcal{T}\}$ be a separable random process with $\|X_t - X_s\|_\psi \leq Cd(t,s)$, for some semimetric $d$ on $\mathcal{T}$ and a constant $C$. Then for any $\nu,\delta>0$, and a constant $K=K(\psi,C)$,
    \begin{equation*}
        \Big\|\sup_{d(t,s)<\delta} |X_t-X_s|\Big\|_{\psi} \leq K\left[\int_0^\nu \psi^{-1}\{D(\epsilon,d)\} d\epsilon + \delta\psi^{-1}\{D^2(\nu,d)\}\right].
    \end{equation*}
    \label{thm.max-ineq-2.2.4}
\end{theorem}

\begin{lemma}[\cite{billingsley2017probability}, Exam. 38.3, p. 527]
    If a random process indexed over $[0,T]$ is sample-path continuous, then it is separable.
\end{lemma}

\begin{corollary}
    Let $\{X_t: t\in [0,T]\}$ be a continuous random process with $(\mathbb{E} |X_t-X_s|^2)^{1/2} \leq C|t-s|$ for some positive constant $C$. Then $X$ is stochastic equicontinuous.
    \label{coro.max-ineq-2.2.4}
\end{corollary}

\begin{proof}
    Apply Thm.~\ref{thm.max-ineq-2.2.4}, with $d(s,t)=|s-t|$ and $\psi(x)=x^2$. It is clear that the packing number $D(\epsilon,d)$ of the interval $[0,1]$ is strictly less than $1/\epsilon$. Hence, for any $\delta$ and $\nu$,
    \begin{equation*}
        \Big(\mathrm{E}\big|\sup_{d(t,s)<\delta} |X_t-X_s|\big|^2\Big)^{\frac{1}{2}} \leq K\Big[\int_0^\nu \frac{1}{\sqrt{\epsilon}} d\epsilon + \frac{\delta}{\nu}\Big] = K\Big(2\sqrt{\nu} + \frac{\delta}{\nu}\Big),
    \end{equation*}
    which implies:
    \begin{equation*}
        \mathrm{E}\Big|\sup_{d(t,s)<\delta} |X_t-X_s|\Big|^2 \leq K^2 \left(2\sqrt{\nu} + \frac{\delta}{\nu}\right)^2.
    \end{equation*}

    Using the Markov inequality, we obtain that, for any $\eta$,
    \begin{equation*}
        P\Big(\sup_{d(t,s)<\delta} |X_t-X_s|>\eta\Big) \leq \frac{\mathrm{E}\left|\sup_{d(t,s)<\delta} |X_t-X_s|\right|^2}{\eta^2} \leq K^2 \Big(\frac{2\sqrt{\nu}}{\eta} + \frac{\delta}{\nu\eta}\Big)^2.
    \end{equation*}

    Then, for any $\epsilon>0$, we can first choose $\nu$ such that $\frac{2\sqrt{\nu}}{\eta}\geq \frac{1}{2K^\prime \epsilon}$; then choose $\delta$ such that $\frac{\delta}{\nu\eta}\leq \frac{1}{2K^\prime\epsilon}$, so that the right-hand side of the above inequality is less than $\epsilon$. This means the stochastic equicontinuity holds.
\end{proof}

\begin{lemma}[Doob's submartingale inequality]
    Let $X_t$ be a cadlag submartingale, then for any $C>0$,
    \begin{equation*}
        P\Big(\sup_{t\in [0,T]} X_t \geq C\Big) \leq \frac{\mathbb{E}\max(X_T,0)}{C}.
    \end{equation*}
    \label{lem.doob-submartingale-ineq}
\end{lemma}

\begin{theorem}[Martingale CLT; \cite{andersen2012statistical}, Thm. II.5.2]
    Let $\left\{U^{(n)}\right\}_{n\in \mathbb{N}}$ be a sequence of square-integrable local martingales in $[0,T]$, possibly defined on different sample spaces and with different filtrations for each $n$. Let $U_0$ be a continuous Gaussian martingale with continuous variance function $\mathcal{V}:[0,T]\mapsto \mathbb{R}^+$, and assume that $U^{(n)}(0) = U_0(0) = 0$ has initial values equal to zero. Suppose further that for every $t\in [0,T]$,
    \begin{equation*}
        [U^{(n)}]_t \to_{\mathcal{P}} \mathcal{V}_t
    \end{equation*}
    as $n\to\infty$. Then it holds that $U^{(n)}\to_{\mathcal{D}} U_0$ in $C[0,T]$ as $n\to \infty$.
    \label{thm.martingale-clt}
\end{theorem}

\begin{lemma}[\cite{christgau2024nonparametric}, Lem. B.13]
    Let $(X_n)$ be a sequence of $C[0,T]$-valued random variables with nondecreasing sample paths, and let $f\in C[0,T]$ be an equicontinuous nondecreasing function. If $X_n(t)\overset{\mathcal{P}}{\to} f(t)$ for every $t$, then
    \begin{equation*}
        \sup_{t\in[0,T]} |X_n(t)-f(t)|\overset{\mathcal{P}}{\to} 0.
    \end{equation*}
    \label{lem.christlemb13}
\end{lemma}

\begin{definition}[Tightness; \cite{billingsley2013convergence}, p. 8]
    Let $(\mathbb{D}, d)$ denote a general metric space. A probability measure $\mu$ on $\mathbb{D}$ is \emph{tight} if for any $\epsilon>0$, there exists a compact set $K\subseteq \mathbb{D}$ such that $\mu(K^c)<\epsilon$. A sequence of probability measures $\{\mu_n\}$ is \emph{tight} if for any $\epsilon>0$, there exists a compact set $K\subseteq \mathbb{D}$ such that $\sup_n \mu(K^c)<\epsilon$. 
    
    A random variable $X$ is \emph{tight} if its induced measure $P\circ X^{-1}$ is tight. A sequence of random variables $\{X_n\}$ is \emph{tight} if the sequence of induced measures $\{P\circ X_n^{-1}\}$ is tight.
    \label{def.tightness}
\end{definition}


\begin{lemma}[\cite{christgau2024nonparametric}; Lem. B.14]
    Let $X$ and $\{X_n\}$ be random variables on $(C[0,T], \ell_\infty)$. Suppose that the finite-dimensional marginals converge in distribution, i.e., $\forall\, 0\leq t_1 < \cdots <t_k \leq T$,
    \begin{equation*}
        \left\{X_n(t_1), ..., X_n(t_k)\right\} \overset{\mathcal{D}}{\to} \left\{X(t_1),...,X(t_k)\right\}
    \end{equation*}
    as random vectors in $\mathbb{R}^k$; and that $X$ and $\{X_n\}$ are both tight. Then, we have $X_n \overset{\mathcal{D}}{\to} X$ on $(C[0,T], \ell_\infty)$ as $n\to \infty$.
    \label{lem.tight+identification}
\end{lemma}

\begin{proposition}[\cite{billingsley2013convergence}; Thm. 7.3, p. 82]
    Let $\{X_n\}$ be a sequence of random variables on $C[0,T]$ such that $X_n(0)=0$ almost surely for every $n\in\mathbb{N}$. Then, $\{X_n\}$ is tight if and only if it is stochastic equicontinuous.
    \label{prop.stoch-equi<=>tight}
\end{proposition}

\begin{proposition}[\cite{christgau2024nonparametric}; Prop. C.2]
    Let $(W_t)$ be a Brownian motion. For every non-decreasing function $f\in C[0,T]$, the process $(W_{f(t)})_{t\in[0,T]}$ is a mean-zero Gaussian martingale on $[0,T]$ with variance function $f$. Consequently, if $U=(U_t)_{t\in [0,T]}$ is a mean-zero Gaussian martingale with a continuous variance function $V$, then $U$ has the distributional representation:
    \begin{equation*}
        (U_t)_{t\in[0,T]} \overset{\mathcal{D}}{=} \{W_{V(t)}\}_{t\in[0,T]}.
    \end{equation*}
    \label{prop.C.2.LCT2024}
\end{proposition}

\subsection{Lemmas used in Appx.~\ref{appx.proof-consistency-ou}}

\begin{theorem}[Existence and unique. of a solution; \cite{arnold1974stochastic}, Thm.~6.2.2, p.~105]
    Suppose that we have a stochastic differential equation:
    \begin{equation}
        dX_t = f(t,X_t)dt + G(t,X_t)dW_t, \quad X_{t_0}=c, \quad t_0\leq t\leq T<\infty,
        \label{eq.tech-lem-sde}
    \end{equation}
    where $W_t$ is an $\mathbb{R}^m$-valued Wiener process and $c$ is a random variable independent of $W_t-W_{t_0}$ for $t\geq t_0$. Suppose that the $R^d$-valued function $f(t,x)$ and the ($d\times m$ matrix)-valued function $G(t,x)$ are defined and measurable on $[t_0,T]\times \mathbb{R}^d$, and have the following properties: there exists a constant $L>0$ such that:

    \noindent a) \,(Lipschitz condition) for all $t\in [t_0,T], x\in \mathbb{R}^d, y\in\mathbb{R}^d$,
    \begin{equation*}
        \|f(t,x)-f(t,y)\| + \|G(t,x)-G(t,y)\| \leq L\|x-y\|,
    \end{equation*}
    where $\|G\|^2=\mathrm{tr} (G G^\top)$ is the Frobenius norm of $G$.

    \noindent b) \, (Restriction on growth) for all $t\in [t_0,T]$ and $x\in\mathbb{R}^d$,
    \begin{equation*}
        \|f(t,x)\|^2 + \|G(t,x)\|^2 \leq L^2(1+\|x\|^2).
    \end{equation*}

    Then, \eqref{eq.tech-lem-sde} has on $[t_0,T]$ a unique $\mathbb{R}^d$-valued solution $X_t$, continuous with probability $1$, that satisfied the initial condition $X_{t_0}=c$.
    \label{thm.tech-lem-existence-unique-solution}
\end{theorem}

\begin{theorem}[$L_{p}$-boundedness of the solution; \cite{arnold1974stochastic}, Thm. 7.1.2, p. 116]
    Suppose that the conditions of Thm.~\ref{thm.tech-lem-existence-unique-solution} are satisfied and that $\mathbb{E}\|c\|^{2p}<\infty$, where $p$ is a positive integer. Then for the solution $X_t$ of the SDE \eqref{eq.tech-lem-sde},
    \begin{equation*}
        \mathbb{E}\|X_t\|^{2p}\leq (1+\mathbb{E}\|c\|^{2p}) e^{C(t-t_0)},
    \end{equation*}
    where $C=2p(2p+1)k^2.$
    \label{thm.tech-lem-Lp-boundedness-of-solution}
\end{theorem}

\begin{corollary}
    Under the conditions in Thm.~\ref{thm.tech-lem-Lp-boundedness-of-solution}, the intensity process $\lambda_t$, the process $m_t$, the addictive residual process $G_t$, and the process $\mu_t$ are all bounded in $L_{2p}$.
    \label{coro.lp-bounded}
\end{corollary}

\begin{proof}
    We have for any $t\in [0,T]$,
    \begin{equation*}
        \mathbb{E}\|\lambda_t\|^{2p} =  \mathbb{E}\|\lambda(X_t)\|^{2p} \leq L^{2p}\left(1+ \sup_t\mathbb{E}\|X_t\|^{2p}\right) < \infty.
    \end{equation*}

    Furthermore, Jensen's inequality yields that:
    \begin{align*}
        \mathbb{E}\|m_t\|^{2p} &= \sup_{t}\mathbb{E}\|\mathbb{E}\{X_{V\backslash C}(t)|\mathcal{F}_t\}\|^{2p} \\
        &\leq \sup_{t}\mathbb{E}[\mathbb{E}\{\|X_{V\backslash C}(t)\|^{2p}|\mathcal{F}_t\}] = \sup_{t}\mathbb{E}\|X_{V\backslash C}(t)\|^{2p} < \infty.
    \end{align*}

    In a similar way, we can show that $G_t=X_\alpha(t)-\mathbb{E}\{X_\alpha(t)|\mathcal{F}_t\}$ and $\mu_t=\mathbb{E}\{\lambda_\beta(X_t)|\mathcal{F}_t\}$ are both bounded in $L_{2p}$.
\end{proof}

\begin{lemma}
    Consider a function $f\in L_p([0,T])$, where $p\geq 1$, then 
    \begin{equation*}
        \left\|\int_0^T f(t) dt \right\|^p \leq T^{p-1} \int_0^T \|f(t)\|^p dt.
    \end{equation*}
    \label{lem.tech-lem-p-norm-of-integral}
\end{lemma}

\begin{proof}
    Let $q:=\frac{p}{p-1}$, and apply the Hölder inequality,
    \begin{equation*}
        \left\|\int_0^T f(t) dt \right\|^p \leq \left\|\left(\int_0^T 1^q dt\right)^\frac{1}{q} \left(\int_0^T \|f(t)\|^p dt\right)^{\frac{1}{p}} \right\|^p = T^{p-1} \int_0^T \|f(t)\|^p dt.
    \end{equation*}
\end{proof}

\begin{lemma}[Burkholder-Davis-Gundy ineq.; \cite{marinelli2016maximal}, Thm. 1.1]
    Let $M$ be a $\mathbb{R}^d$-valued local martingale with $M_{t_0}=0$. Then for any $1\leq p<\infty$, there exists universal positive constants (independent of $M$) $c_p, C_p$, such that for any stopping time $\tau$,
    \begin{equation*}
        c_p \mathbb{E} [M]_\tau^{\frac{p}{2}} \leq \mathbb{E}\left(\sup_{t_0\leq t\leq \tau} \|M_t\|^p\right) \leq C_p \mathbb{E} [M]_\tau^{\frac{p}{2}}.
    \end{equation*}
    \label{lem.burkholder-david}
\end{lemma}

\begin{corollary}[$L_p$-continuity of the solution]
     Under the conditions in Thm.~\ref{thm.tech-lem-Lp-boundedness-of-solution}, there exists a constant $K>0$ such that for any $s, t\in [0,T]$, $\mathbb{E}\|X_s - X_t\|^{2p} \leq K|s-t|^p$.
    \label{lem.smooth}
\end{corollary}

\begin{proof}
    Without loss of generalizability, fix an $s\in [0,T]$ and let $t\geq s$. We first have:
    \begin{align*}
        \mathbb{E}\|X_t - X_s\|^{2p} &= \mathbb{E}\left\|\int_s^t f(u,X_u)du  + \int_s^t G(u,X_u) dW_u\right\|^{2p} \\
        &\leq 2^{2p} \left\{\mathbb{E}\left\|\int_s^t f(u,X_u)du \right\|^{2p} + \mathbb{E}\left\|\int_s^t G(u,X_u) dW_u \right\|^{2p}\right\}.
    \end{align*}

    For the first term on the RHS, apply Lem.~\ref{lem.tech-lem-p-norm-of-integral}, we have:
    \begin{align*}
        \mathbb{E}\left\|\int_s^t f(u,X_u)du \right\|^{2p} &\leq (t-s)^{2p-1} \mathbb{E}\int_s^t \|f(u,X_u)\|^{2p} du \\
        &\leq T^p (t-s)^{p-1} \mathbb{E}\int_s^t \|f(u,X_u)\|^{2p} du.
    \end{align*}

    For the second term on the RHS, we note that the process $M_t := \int_s^t G(u,X_u) dW_u$ is a martingale with $M_s=0$ and $[M]_t=\int_s^t \|G(u,X_u)\|^2 du$. Hence, apply Lem.~\ref{lem.burkholder-david}, we have:
    \begin{align*}
        \mathbb{E}\left\|\int_s^t G(u,X_u) dW_u \right\|^{2p} &\leq C_p \mathbb{E}\left(\int_s^t \|G(u,X_u)\|^2 du \right)^p \\ &\overset{(1)}{\leq} C_p (t-s)^{p-1} \mathbb{E}\int_s^t \|G(u,X_u)\|^{2p} du. 
    \end{align*}
    where ``(1)'' applies Lem.~\ref{lem.tech-lem-p-norm-of-integral}.

    Therefore, we have:
    \begin{align*}
        &\mathbb{E}\|X_t - X_s\|^{2p} \leq 2^{2p} \left\{\mathbb{E}\left\|\int_s^t f(u,X_u)du \right\|^{2p} + \mathbb{E}\left\|\int_s^t G(u,X_u) dW_u \right\|^{2p}\right\} \\
        &\manyquad[2] \leq 2^{2p} \left\{T^p (t-s)^{p-1} \mathbb{E}\int_s^t \|f(u,X_u)\|^{2p} du + C_p (t-s)^{p-1} \mathbb{E}\int_s^t \|G(u,X_u)\|^{2p} du\right\} \\
        &\manyquad[2] \leq 2^{2p} \max(T^p, C_p) (t-s)^{p-1}  \mathbb{E}\int_s^t \|f(u,X_u)\|^{2p} + \|G(u,X_u)\|^{2p} du.
    \end{align*}
    
    According to condition (b) in Thm.~\ref{thm.tech-lem-Lp-boundedness-of-solution}, we have:
    \begin{align*}
        \|f(u,X_u)\|^{2p} + \|G(u,X_u)\|^{2p} &\leq \{\|f(u,X_u)\|^2+\|G(u,X_u)\|^2\}^p \\
        &\leq L^{2p} (1+\|X_u\|^2)^p \\
        &\leq L^{2p} 2^p (1+\|X_u\|^{2p})\\
        &\leq L^{2p} 2^p (1+\sup_{u\in[0,T]} \|X_u\|^{2p}) <\infty.
    \end{align*}

    Hence,
    \begin{align*}
        \mathbb{E}\|X_t - X_s\|^{2p} &\leq 2 (2L)^{2p} \max(T^p, C_p) (1+\sup_{u\in[0,T]} \|X_u\|^{2p}) (t-s)^p =: K (t-s)^p.
    \end{align*}

    We note that $K$ does not rely on the time steps $s,t$. This has concluded the proof.
\end{proof}

\begin{theorem}[Marcinkiewicz–Zygmund inequality; \cite{ren2001best}, the main thm.]
    Let $\{X_n, n\geq 1\}$ be a sequence of r.v.'s with $\mathbb{E}X_n=0$, then
    \begin{equation*}
        \mathbb{E}\left\|\sum_{i=1}^n X_i\right\|^p \leq C_p n^{\frac{p}{2}-1} \sum_{i=1}^n \mathbb{E}\|X_i\|^p, \quad p\geq 2.
    \end{equation*}
    \label{thm.tech-lem-Marc-Zyg-inequality}
\end{theorem}

\begin{theorem}[\cite{bhansali1991convergence}, the main theorem]
    \label{thm.bhansali1991convergence}
    Let $\{X_i\}$ be a sequence of r.v.'s satisfying:
    \begin{enumerate}[label=(\roman*)]
        \item the sequence is strictly stationary.
        \item the $X_i$'s have finite moments of all orders.
        \item for any finite set $\{j,l_1,l_2,...,l_k\}$ of distinct integers, the joint distribution of $X_j$, $X_{l_1}, X_{l_2}, ..., X_{l_k}$ is absolute continuous, and there exists a constant $K$ and a version $f_{jl_1,...,l_k}(\cdot|x_1,...,x_k)$ of the conditional probability density function of $X_j$ given $X_{l_1}=x_1, ..., X_{l_k}=x_k$ such that $f_{jl_1,...,l_k}(x|x_1,...,x_k) \leq K$ for all $x,x_1,...,x_k$.
    \end{enumerate}

    Let $p$ be a positive integer and for each $j=0,1,2,...$, consider the positive-definite symmetric matrix:
    \begin{equation*}
        A_j = \begin{bmatrix}
               X_{p+j} X_{p+j} & \cdots & X_{p+j}X_{1+j}\\
               \vdots & \, & \vdots \\
               X_{1+j} X_{p+j}& \cdots &X_{1+j} X_{1+j}
               \end{bmatrix}.
    \end{equation*} 

    Then for each $q>0$ there is a nonnegative r.v. $\Lambda_0$ and a number $r>0$ such that $\mathbb{E}\Lambda_0^q<\infty$ and for all $N\geq r$,
    \begin{equation*}
        \left\|\left(\frac{1}{N}\sum_{j=0}^{N-1} A_j\right)^{-1}\right\|_{\mathrm{sp}} \leq \Lambda_0 \quad almost \,\, surely.
    \end{equation*}
\end{theorem}

\begin{lemma}[\cite{golub2013matrix} Coro. 8.6.2]
    Denote $\sigma_i(A)$ as the $i$-th largest singular value of the matrix $A$. If $A$ and $A+E$ are in $\mathbb{R}^{p\times q}$ with $p\geq q$, then for $k=1,...,q$,
    \begin{equation*}
        |\sigma_k(A+E)-\sigma_k(A)| \leq \|E\|_{\mathrm{sp}}.
    \end{equation*}
    \label{lem.lip-continuty-singular-value}
\end{lemma}

Let $\theta=[\theta_1,...,\theta_k]^\top, \xi=[\xi_1,...,\xi_l]^\top$ be a continuous-time diffusion process with
\begin{subequations}
\label{eq.filtering-setup}
\begin{align}
    &d\theta_t = \left\{a_0(t,\xi) + a_1(t,\xi) \theta_t\right\} dt + \sum_{i=1}^2 b_i(t,\xi) dW_i(t), \\
    &d\xi_t = \left\{A_0(t,\xi) + A_1(t,\xi) \theta_t \right\} dt + \sum_{i=1}^2 B_i(t,\xi) dW_i(t),
\end{align}
\end{subequations}
where $a_0(t,x):\mathcal{T}\times \mathcal{X}\mapsto \mathbb{R}^k, A_0(t,x):\mathcal{T}\times \mathcal{X}\mapsto \mathbb{R}^l$ are vector-valued functionals, $a_1: k\times k, A_1: l\times k, b_1: k\times k, b_2:k\times l, B_1:l\times k, B_2: l\times l$ are matrix-valued functionals, and $W_1=[W_{11},...,W_{1k}]^\top, W_2=[W_{21},...,W_{2l}]^\top$ are vector Wiener processes. Denote $m_t:=\mathbb{E}(\theta_t|\mathcal{F}_t^\xi)=[\mathbb{E}(\theta_{1,t}|\mathcal{F}_t^\xi),...,\mathbb{E}(\theta_{k,t}|\mathcal{F}_t^\xi)]^\top$ as the projection process of $\theta$ relative to $\mathcal{F}_t^\xi$. Further, let $\gamma_t :=\mathrm{Var}(\theta_t|\mathcal{F}_t^\xi)=\mathbb{E}\{(\theta_t - m_t) (\theta_t - m_t)^\top|\mathcal{F}_t^\xi\}$ be the conditional variance process.

\begin{theorem}[Filtering equation; \cite{liptser2013statistics}, Thm. 12.7, Vol. II, p. 36]
    Let $(\theta,\xi)$ be a random process with differentials given in (\ref{eq.filtering-setup}a) and (\ref{eq.filtering-setup}b). Suppose that some regularity conditions for integrability hold, and that $P(\theta_0\leq a|\xi_0)\sim \mathcal{N}(m_0,\gamma_0)$ is Gaussian. Use the notations: 
    \begin{equation*}
        b\circ b = b_1 b_1^\top + b_2 b_2^\top, \quad b\circ B = b_1 B_t^\top + b_2 B_2^\top, \quad B\circ B = B_1 B_1^\top + B_2 B_2^\top.
    \end{equation*}
    Then, $m_t$ and $\gamma_t$ satisfy the equations:
    \begin{subequations}
    \label{eq.filtering-eq}
    \begin{align}
        &dm_t = (a_0 + a_1 m_t) + (b\circ B + \gamma_t A_t^{-1})(B\circ B)^{-1} \left\{d\xi_t - (A_0+A_1 m_t)dt\right\}, \\
        &\dot{\gamma}_t = a_t \gamma_t + \gamma_t a_1^\top + b\circ b - (b\circ B+\gamma_t A_1^\top)(B\circ B)^{-1}(b\circ B+\gamma_t A_1\top)^\top,
    \end{align}
    \end{subequations}
    subject to the condition $m_0=\mathbb{E}(\theta_0|\xi_0), \gamma_0=\mathbb{E}\left\{(\theta_0-m_0)(\theta_0-m_0)^\top|\xi_0\right\}$.
    \label{eq.filtering-thm}
\end{theorem}

\begin{theorem}[\cite{liptser2013statistics}, Thm. 12.7, Vol. II, p.36]
    Let the conditions of Thm.~\ref{eq.filtering-thm} be satisfied. Then the system of equations (\ref{eq.filtering-eq}a)-(\ref{eq.filtering-eq}b), subject to the initial condition on $m_0, \gamma_0$, has a unique, continuous, $\mathcal{F}_t^\xi$-adapted solution for any $t$, $0\leq t\leq T$.
    \label{thm.filtering-unique}
\end{theorem}

\subsection{Lemmas used in Appx.~\ref{appx.proof-consistency-projection-estimator}}

If polynomials $f(x)$ and $g(x)$ have all real roots $r_1\leq r_2\leq \cdots \leq r_n$ and $s_1\leq s_2\leq \cdots \leq s_{n-1}$, then we say $f$ and $g$ \emph{interlace} if and only if
\begin{equation*}
    r_1\leq s_1\leq r_2\leq s_2\leq \cdots \leq s_{n-1}\leq r_n.
\end{equation*}

\begin{theorem}[Cauchy interlacing theorem; \cite{fisk2005very}, Coro. 1]
    If $A$ is a Hermitian matrix, and $B$ is a principal submatrix of $A$, then the eigenvalues of $B$ interlace the eigenvalues of $A$.
    \label{thm.tech-lem-cauchy-interlacing}
\end{theorem}

\begin{lemma}[Gronwall inequality; \cite{platen1992numerical}, Lem. 4.5.1, p. 129]
    Let $u, \alpha: [0,T]\mapsto \mathbb{R}$ be integrable with
    \begin{equation*}
        0 \leq u(t) \leq \alpha(t)+  \int_0^t \beta u(s)ds
    \end{equation*}
    for some $\beta>0$. Then, for every $t\in [0,T]$,
    \begin{equation*}
        u(t)\leq \alpha(t) + \beta \int_0^t e^{\beta(t-s)} \alpha(s) ds.
    \end{equation*}
    \label{lem.Gronwall}
\end{lemma}

\begin{proposition}[Thm. 2-2 of \cite{potter1965matrix}]
    Consider the RDE
    \begin{equation}
        \dot{y}_t = ay_t + y_t a^\top + I - y_t b y_t, \quad y_0=y(0),
        \label{appx-eq.rde}
    \end{equation}
    where $a, b, y(0), y_t$ are all $k\times k$ matrices. Suppose that $y(0), b$ are positive semidefinite. Then the solution $y(t)$ exists on $[0,\infty)$ and
    \begin{equation*}
        \sup_{t\in [0,T]} \|y_t\|_1 \leq \frac{1}{2} k \left\{\mathrm{tr}(y_0) + kT\right\} e^{T \|a\|_1}
    \end{equation*}
    \label{prop.up-bound-yt}
\end{proposition}

\begin{proposition}
    Consider the solutions of the RDE under two different parameters:
    \begin{align*}
        \dot{y}_1 = a_1 y_1 + y_1 a_1^\top + I - y_1 b_1 y_1 \quad \text{with initial value $y_1(0)$}, \\
        \dot{y}_2 = a_2 y_2 + y_2 a_2^\top + I - y_2 b_2 y_2  \quad \text{with initial value $y_2(0)$}.
    \end{align*}
    Then, for $p\geq 1$, we have:
    \begin{equation*}
        \sup_{t\in [0,T]} \|y_1(t)-y_2(t)\|^p_\mathrm{sp} \leq 2^p\left(\|y_1(0)-y_2(0)\|^p_\mathrm{sp} + \alpha T^p\right) \exp(2^p T^{p-1} \beta),
    \end{equation*}
    where
    \begin{align*}
        &\alpha := 2^{3p+1}\|a_1 - a_2\|^p_\mathrm{sp} \sup_t \|y_2(t)\|^p_\mathrm{sp} + 2^{4p}\sup_t \|y_1(t)\|^p_\mathrm{sp} \sup_t \|y_2(t)\|^p_\mathrm{sp} \|b_1 - b_2\|^p_\mathrm{sp},\\
        &\beta := 2^{3p+1} \|a_1\|^p_\mathrm{sp} + 2^{4p+1}\left\{\|b_1\|^p_\mathrm{sp} \sup_t\|y_1(t)\|^p_\mathrm{sp} + \|b_2\|^p_\mathrm{sp} \sup_t \|y_2(t)\|^p_\mathrm{sp}\right\}.
    \end{align*}
    \label{prop.lip-yt}
\end{proposition}

\begin{proof}
    Note that the solutions $y_i(t), i=1,2$ satisfy:
    \begin{align*}
        y_i(t)=y_i(0) + \int_0^t a_i y_i(s) + y_i(s)a_i ^\top + I - y_i(s)b_i y_i(s) ds.
    \end{align*}
    Then, we consider $z(t) := \|y_1(t) - y_2(t)\|^p_{\mathrm{sp}}$, which has been shown to be bounded by Prop.~\ref{prop.up-bound-yt}. Therefore, it is integrable. Applying the inequality in Lem.~\ref{lem.tech-lem-p-norm-of-integral}, we have:
    \small
    \begin{align*}
        z(t) &\leq 2^p\|y_1(0)-y_2(0)\|^p_\mathrm{sp} \\
        &\quad + 2^p\left\|\int_0^t a_1 y_1(s) - a_2 y_2(s) + y_1(s) a_1^\top - y_2(s) a_2^\top + y_1(s) b_1 y_1(s) - y_2(s) b_2 y_2(s)\right\|_\mathrm{sp}^p \\
        &\leq 2^p z(0) + 2^p T^{p-1} \int_0^t 2^{2p+1}\|a_1 y_1(s) - a_2 y_2(s)\|^p_\mathrm{sp} + 2^{2p}\|y_1(s) b_1 y_1(s) - y_2(s) b_2 y_2(s)\|^p_\mathrm{sp} ds.
    \end{align*}
    \normalsize
    For $\|a_1 y_1(s) - a_2 y_2(s)\|^p_\mathrm{sp}$, we have:
    \begin{align*}
        \|a_1 y_1(s) - a_2 y_2(s)\|^p_\mathrm{sp} &\leq 2^p\|a_1 y_1(s) - a_1 y_2(s)\|^p_\mathrm{sp} + 2^p\|a_1 y_2(s) - a_2 y_2(s)\|^p_\mathrm{sp} \\
        &= 2^p\|a_1\|^2_\mathrm{sp} \|y_1(s) - y_2(s)\|^p_\mathrm{sp} + 2^p\|a_1 - a_2\|^p_\mathrm{sp} \|y_2(s)\|^p_\mathrm{sp} \\
        &\leq 2^p\|a_1\|^p_\mathrm{sp} z(s) + 2^p\|a_1 - a_2\|^p_\mathrm{sp} \sup_t \|y_2(t)\|^p_\mathrm{sp} 
    \end{align*}
    For $\|y_1(s) b_1 y_1(s) - y_2(s)b_2 y_2(s)\|^p_\mathrm{sp}$, we have:
    \small
    \begin{align*}
        &\|y_1(s) b_1 y_1(s) - y_2(s)b_2 y_2(s)\|^p_\mathrm{sp} \\
        & = \|\{y_1(s) - y_2(s)\} b_1 y_1(s) + y_2(s)(b_1 - b_2)y_1(s) + y_2(s) b_2 \{y_1(s) - y_2(s)\}\|^p_\mathrm{sp} \\
        & \leq 2^{2p+1}\left\{\|b_1\|^2_\mathrm{sp} \|y_1(s)\|^p_\mathrm{sp} + \|b_2\|^2_\mathrm{sp} \|y_2(s)\|^p_\mathrm{sp}\right\} \|y_1(s) - y_2(s)\|^2_\mathrm{sp} + 2^{2p}\|y_1(s)\|^2_\mathrm{sp} \|y_2(s)\|^2_\mathrm{sp} \|b_1 - b_2\|^p_\mathrm{sp} \\
        & \leq 2^{2p+1}\left\{\|b_1\|^2_\mathrm{sp} \sup_t \|y_1(t)\|^p_\mathrm{sp} + \|b_2\|^p_\mathrm{sp} \sup_t \|y_2(t)\|^p_\mathrm{sp}\right\} z(s) + 2^{2p}\sup_t \|y_1(t)\|^p_\mathrm{sp} \sup_t \|y_2(t)\|^p_\mathrm{sp} \|b_1 - b_2\|^p_\mathrm{sp}
    \end{align*}
    \normalsize
    Thus, we have:
    \small
    \begin{align*}
        &z(t) \leq 2^p z(0) + 2^p T^{p-1} \int_0^t 2^{3p+1}\|a_1\|^p_\mathrm{sp} z(s) + 2^{3p+1}\|a_1 - a_2\|^p_\mathrm{sp} \sup_t \|y_2(t)\|^p_\mathrm{sp} \\
        & + 2^{4p+1}\left\{\|b_1\|^p_\mathrm{sp} \sup_t\|y_1(t)\|^p_\mathrm{sp} + \|b_2\|^p_\mathrm{sp} \sup_t \|y_2(t)\|^p_\mathrm{sp}\right\} z(s) + 2^{4p}\sup_t \|y_1(t)\|^p_\mathrm{sp} \sup_t \|y_2(t)\|^p_\mathrm{sp} \|b_1 - b_2\|^p_\mathrm{sp} ds.
    \end{align*}
    \normalsize
    Denote
    \begin{align*}
        &\alpha := 2^{3p+1}\|a_1 - a_2\|^p_\mathrm{sp} \sup_t \|y_2(t)\|^p_\mathrm{sp} + 2^{4p}\sup_t \|y_1(t)\|^p_\mathrm{sp} \sup_t \|y_2(t)\|^p_\mathrm{sp} \|b_1 - b_2\|^p_\mathrm{sp},\\
        &\beta := 2^{3p+1} \|a_1\|^p_\mathrm{sp} + 2^{4p+1}\left\{\|b_1\|^p_\mathrm{sp} \sup_t\|y_1(t)\|^p_\mathrm{sp} + \|b_2\|^p_\mathrm{sp} \sup_t \|y_2(t)\|^p_\mathrm{sp}\right\}.
    \end{align*}
    We then have:
    \begin{equation*}
        z(t)\leq 2^p\left\{ z(0) + T^p \bar{\alpha} + T^{p-1}\int_0^t \bar{\beta} z(s) ds\right\}.
    \end{equation*}
    Applying the Gronwall Lem.~\ref{lem.Gronwall}, it yields that:
    \begin{equation*}
        z(t) \leq 2^p \{z(0)+T^p \alpha\} \exp(2^p T^{p-1} \beta)
    \end{equation*}
    which concludes the proof.
\end{proof}

\begin{corollary}
    Let $u>1$ and denote $U_y(u,T), U_A(u,T), U_B(u,T), U_C(u,T)$ as the upper bounds of $y=y(t,\Phi,y_0), A=A(t,\Phi,y_0), B=B(t,\Phi,y_0), C=C(t,\Phi,y_0)$, in terms of the spectral norm, over the set
    \begin{equation*}
        \left\{(t,\Phi,y_0)\mid 0\leq t\leq T, \|\Phi\|_{\mathrm{sp}} \leq u-1, \|y_0\|_\mathrm{sp} \leq u^2+u^6\right\}.
    \end{equation*}
    Then we have:
    \begin{align*}
        &U_y(u,T) = \frac{1}{2}d\{(u^2+u^6)+dT\} e^{uT} = O(u^6 e^{uT})\\
        &U_A(u,T) = U_B(u,T) = u + \frac{1}{2}d\{(u^2+u^6)+dT\} u^2 e^{uT} = O(u^8 e^{uT})\\
        &U_C(u,T) = \frac{1}{2}d\{(u^2+u^6)+dT\} u e^{uT} = O(u^7 e^{uT})
    \end{align*}
    \label{coro.tech-lem-Uy-UA-UB-UC}
\end{corollary}

\begin{proof}
    Recall $y$ is the solution of
    \begin{equation*}
        \dot{y}_t = \Phi_{V\backslash C} y_t + y_t \Phi_{V\backslash C}^\top + I - y_t \Phi_{C, V\backslash C}^\top \Phi_{C, V\backslash C} y_t \quad \text{with initial value} \,\,\, y_0.
    \end{equation*}

    Then, apply \eqref{eq.equivalence-matrix-norms} and Prop.~\ref{prop.up-bound-yt}, 
    \begin{align*}
        &U_y(u,T) = \sup_{t\in[0,T]} \|y_t\|_\mathrm{sp} \leq \sup_{t\in[0,T]} \|y_t\|_1 \leq  \frac{1}{2}d\{\|y_0\|_1+dT\} e^{\|\Phi_{V\backslash C}\|_1 T} \\
        &\manyquad[2] \leq \frac{1}{2}d\{\|y_0\|_\mathrm{sp}+dT\} e^{\|\Phi_{V\backslash C}\|_\mathrm{sp} T} \leq \frac{1}{2}d\{\|y_0\|_\mathrm{sp}+dT\} e^{\|\Phi\|_\mathrm{sp} T} \leq \frac{1}{2}d\{(u^2+u^6)+dT\} e^{uT}.\\
    \end{align*}

    Furthermore, recall $A := \Phi_{V\backslash C, C} - y_t  \Phi_{C, V\backslash C}^\top  \Phi_{C}$. Hence,
    \begin{align*}
        \|A\|_{\mathrm{sp}} &\leq \|\Phi_{V\backslash C, C}\|_{\mathrm{sp}} + \|y_t\|_{\mathrm{sp}} \|\Phi_{C, V\backslash C}^\top\|_{\mathrm{sp}}  \|\Phi_{C}\|_\mathrm{sp} \\
        &\leq \|\Phi\|_{\mathrm{sp}} + U_y(u,T) \|\Phi\|_{\mathrm{sp}}^2 \leq u + \frac{1}{2}d\{(u^2+u^6)+dT\} u^2 e^{uT} =: U_A(u,T).
    \end{align*}

    In a similar way, we can show $U_B(u,T)$ and $U_C(u,T)$.
\end{proof}

\section{Additional simulation results}
\label{app.experiment}

We first investigate conditional local independence testing with nonlinear data and data from non-isometric noise. For the nonlinear case, we generate data from the following SDE:
\begin{equation*}
    dX_t = \Phi\{X_t+\sin(2\pi X_t)\}dt + \sigma dW_t,
\end{equation*}
where $\sin(\cdot)$ denote the element-wise sine function. For the non-isometric case, we consider the diffusion matrix $\Sigma=\mathrm{diag}(1,2,3,3,2,1,1,2,3,1)$. For both cases, we set the rest of the parameters, $\Phi,\sigma,\delta,T$, to be the same as in Sec.~\ref{sec.exp-CL-test}. We consider different sample sizes $N_0\in \{50,100,150,200,250\}$ and the scale $d=10$. The results are shown in Fig.~\ref{fig:exp-3-4}. We see that our test establishes type I error rates close to the significance level $\alpha=0.05$ and high recalls, which demonstrates the robustness of our filtering estimator to nonlinear and non-isometric data. 

\begin{figure}[htp]
    \centering
    \includegraphics[width=0.75\linewidth]{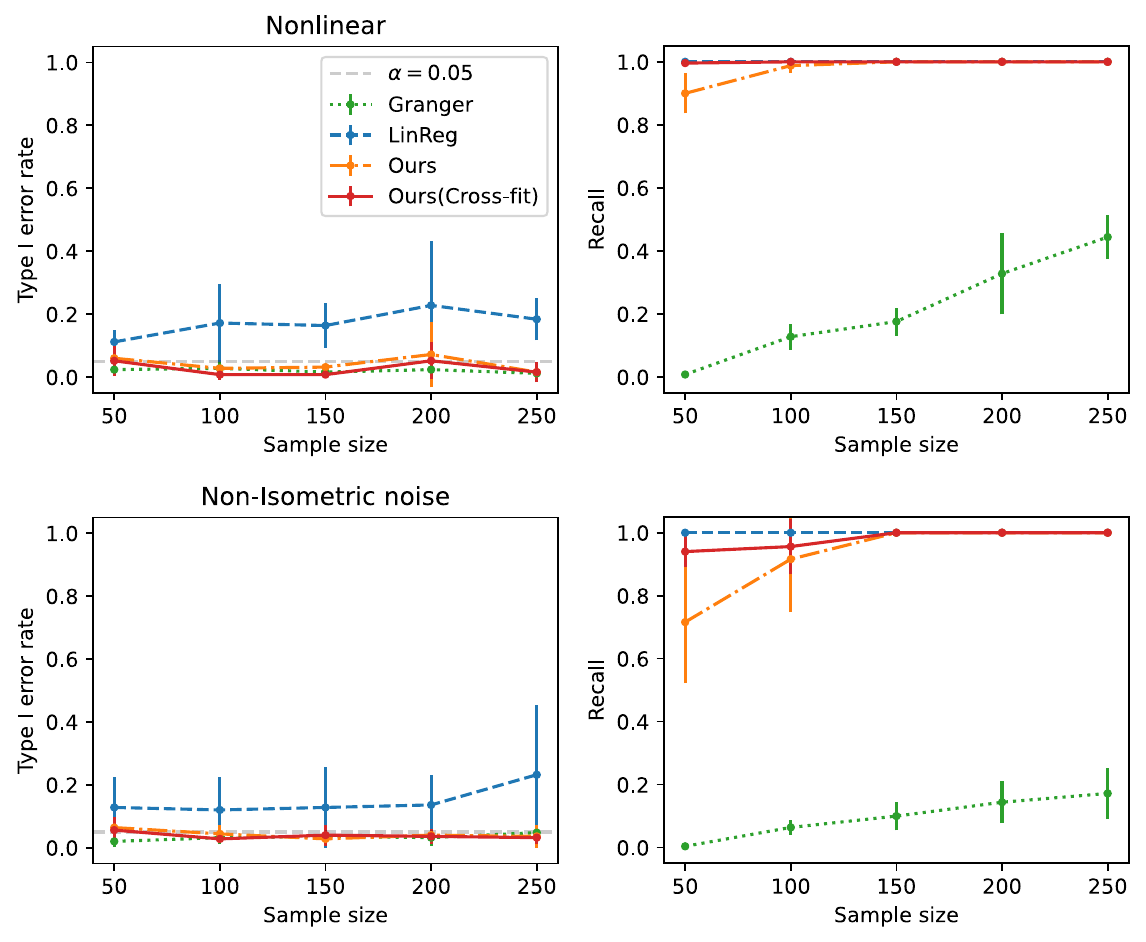}
    \caption{Type I error rate and recall of different test procedures under nonlinear data (upper row) and data with non-isometric noise (lower row).}
    \label{fig:exp-3-4}
\end{figure}

\begin{figure}[htp]
    \centering
    \includegraphics[width=\linewidth]{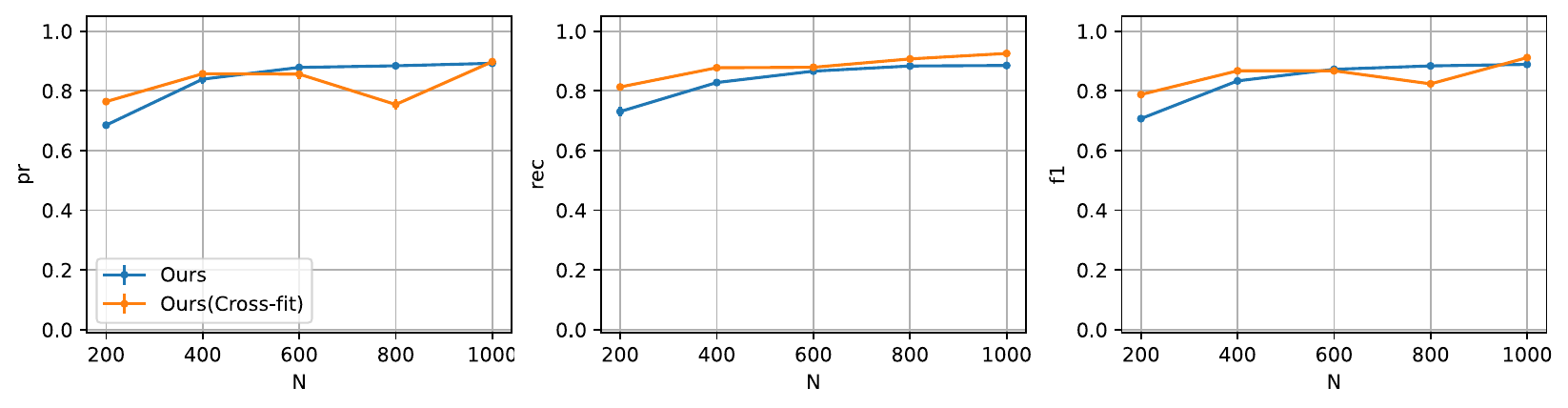}
    \caption{Precision, recall, and F1-score for the causal discovery of LIG with our tests.}
    \label{fig:exp-7}
\end{figure}

We then evaluate the causal discovery of LIG on synthetic data. For data generation, we sample the parameters $\Phi$ in the same way as Sec.~\ref{sec.exp-CL-test}. We set $\sigma=1,\delta=0.001, T=1$. We consider different sample sizes $N_0\in \{200,400,600,800,1000\}$ and the scale $d=50$. For evaluation, we report the precision, recall, and F1-score of the recovered causal graph, where precision and recall measure (respectively) the accuracy and completeness of identified causal edges, and $\mathrm{F_1}:=2\frac{\mathrm{precision}-\mathrm{recall}}{\mathrm{precision}+\mathrm{recall}}$. The results are shown in Fig.~\ref{fig:exp-7}. As shown, our test can present precision and recall close to $90\%$ when $N_0=1000$. This shows that our procedure can be effectively applied to causal discovery for large-scale graphs.

\end{document}